\RequirePackage{snapshot}
\documentclass[acmtog,authorversion,natbib=false]{acmart}
\usepackage[utf8]{inputenc}

\usepackage{mathtools}
\usepackage{stmaryrd}
\usepackage{breqn}
\usepackage{color}
\usepackage[inline]{enumitem}
\usepackage{siunitx}

\usepackage{booktabs}
\usepackage{tabularx}
\usepackage{tabulary}
\usepackage{tabularray}

\usepackage{graphbox}
\usepackage{wrapfig}
\usepackage{pgfplots}
\pgfplotsset{compat=1.18}
\usepackage{pgfplotstable}
\usepackage{currfile}

\usetikzlibrary{external}
\usepackage[ruled]{algorithm2e}
\SetAlFnt{\small}
\SetAlCapFnt{\small}
\SetAlCapNameFnt{\small}
\SetAlCapHSkip{0pt}
\IncMargin{-\parindent}

\PassOptionsToPackage{breaklinks}{hyperref}
\usepackage[capitalize,nameinlink]{cleveref}

\newcommand{\R}{\mathbb{R}}
\newcommand{\Sph}{\mathbf{S}}
\newcommand{\Lie}{\mathcal{L}}
\newcommand{\dext}{\mathrm{d}}
\newcommand{\Dext}{\mathrm{D}}
\newcommand{\SO}{\mathrm{SO}}
\DeclareMathOperator{\id}{id}
\DeclareMathOperator{\Span}{span}
\DeclareMathOperator{\gr}{graph}
\DeclareMathOperator{\Area}{Area}
\DeclareMathOperator{\arcsinh}{arcsinh}
\DeclareMathOperator{\tr}{tr}
\newcommand{\Mass}{\mathbf{M}}
\newcommand{\subj}{\text{s.t.}}
\newcommand{\dist}{\mathrm{dist}}

\newcommand{\hodge}{{\star}}

\newcommand{\dsc}{\mathsf}
\newcommand{\dscm}{\mathbf}

\newcommand{\vect}[1]{\mathbf{#1}}
\newcommand{\vvf}{\vect{t}}

\DeclareMathOperator*{\argmin}{arg\,min}
\DeclareMathOperator*{\argmax}{arg\,max}

\newtheorem{definition}{Definition}
\newtheorem{proposition}{Proposition}

\setcopyright{rightsretained}
\acmJournal{TOG}
\acmYear{2024}
\acmVolume{43}
\acmNumber{4}
\acmArticle{60}
\acmMonth{7}
\acmDOI{10.1145/3658198}

\acmSubmissionID{767}

\RequirePackage[
  backend=biber,
  datamodel=acmdatamodel,
  style=acmauthoryear,
  uniquelist=false,
  maxsortnames=2,
  maxcitenames=2,
  sorting=nyt
  ]{biblatex}
  
\begin{document}

\title{Lifting Directional Fields to Minimal Sections}

\author{David Palmer}
\email{dpalmer@seas.harvard.edu}
\orcid{0000-0002-1931-5673}
\affiliation{%
  \institution{Harvard University}
  \streetaddress{29 Oxford St 307A}
  \city{Cambridge}
  \state{Massachusetts}
  \country{USA}
  \postcode{02138}
}

\author{Albert Chern}
\email{alchern@ucsd.edu}
\orcid{0000-0002-9802-3619}
\affiliation{%
  \institution{University of California San Diego}
  \streetaddress{EBU3B, Room 4112}
  \streetaddress{9500 Gilman Dr, MC 0404}
  \city{La Jolla}
  \state{California}
  \postcode{92093-0404}
  \country{USA}
}

\author{Justin Solomon}
\email{jsolomon@mit.edu}
\orcid{0000-0002-7701-7586}
\affiliation{%
  \institution{Massachusetts Institute of Technology}
  \streetaddress{32 Vassar St, 32-D460}
  \city{Cambridge}
  \state{Massachusetts}
  \country{USA}
  \postcode{02139}}

\begin{abstract}
Directional fields, including unit vector, line, and cross fields, are essential tools in the geometry processing toolkit. The topology of directional fields is characterized by their singularities. While singularities play an important role in downstream applications such as meshing, existing methods for computing directional fields either require them to be specified in advance, ignore them altogether, or treat them as zeros of a relaxed field. While fields are ill-defined at their singularities, the graphs of directional fields with singularities are well-defined surfaces in a circle bundle. By lifting optimization of fields to optimization over their graphs, we can exploit a natural convex relaxation to a \emph{minimal section} problem over the space of currents in the bundle. This relaxation treats singularities as first-class citizens, expressing the relationship between fields and singularities as an explicit boundary condition.  As curvature frustrates finite element discretization of the bundle, we devise a hybrid spectral method for representing and optimizing minimal sections. Our method supports field optimization on both flat and curved domains and enables more precise control over singularity placement.

\end{abstract}




\begin{teaserfigure}
\newcommand{\imgwidth}{0.185\textwidth}
\centering
\begin{tabulary}{\textwidth}{@{}CCCCC@{}}
\includegraphics[width=\imgwidth]{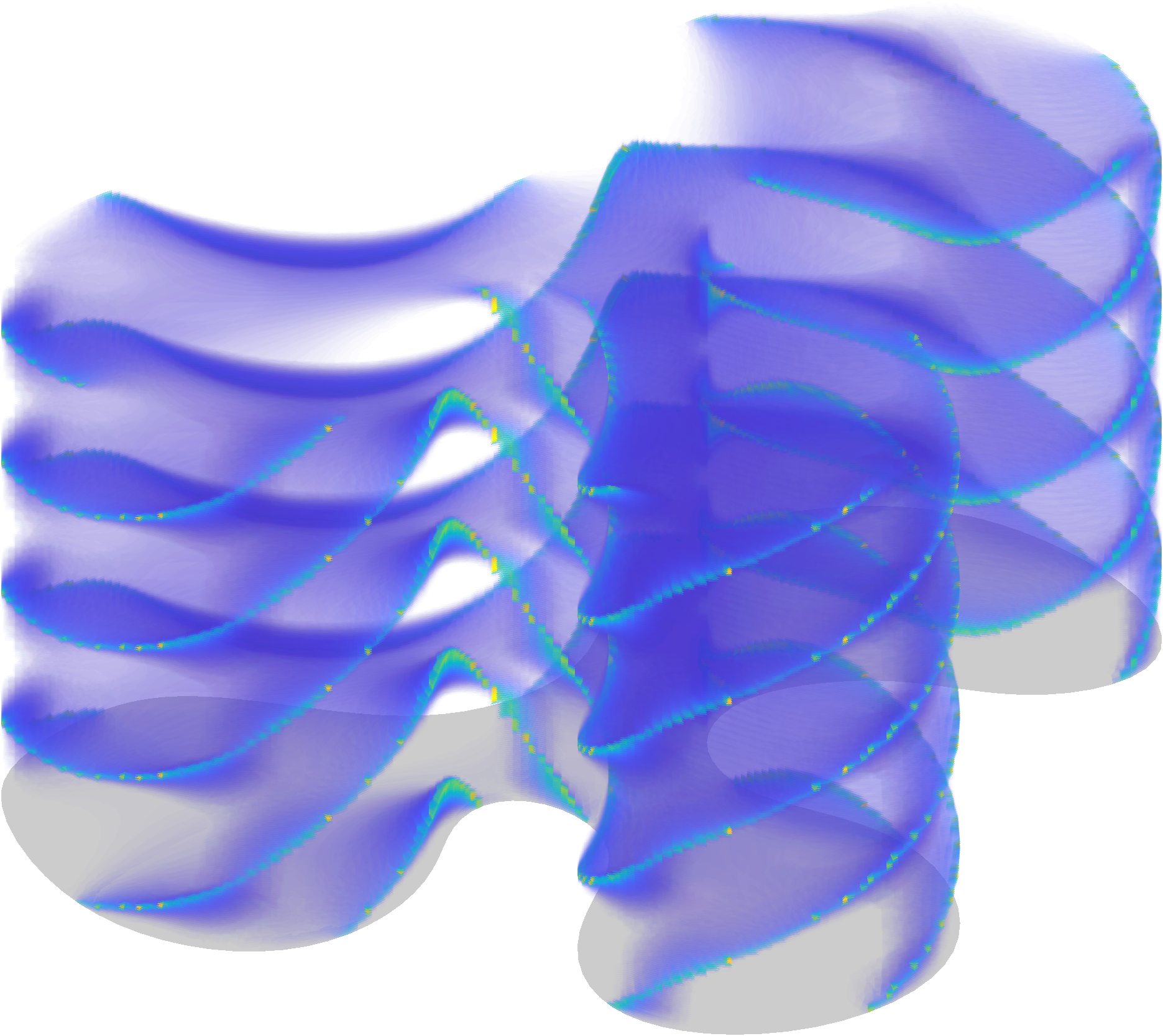}&%
\includegraphics[width=\imgwidth]{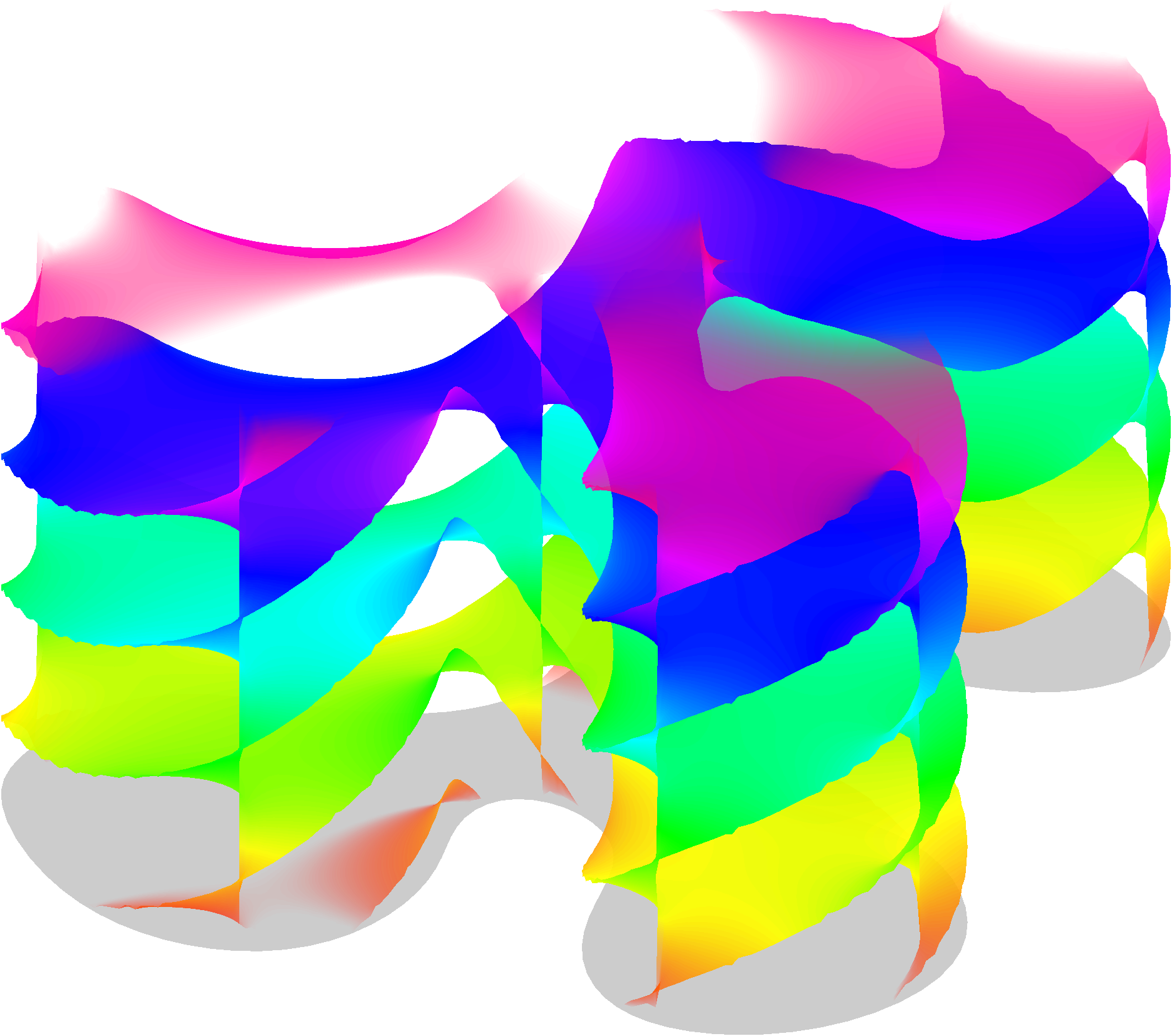}&%
\includegraphics[width=\imgwidth]{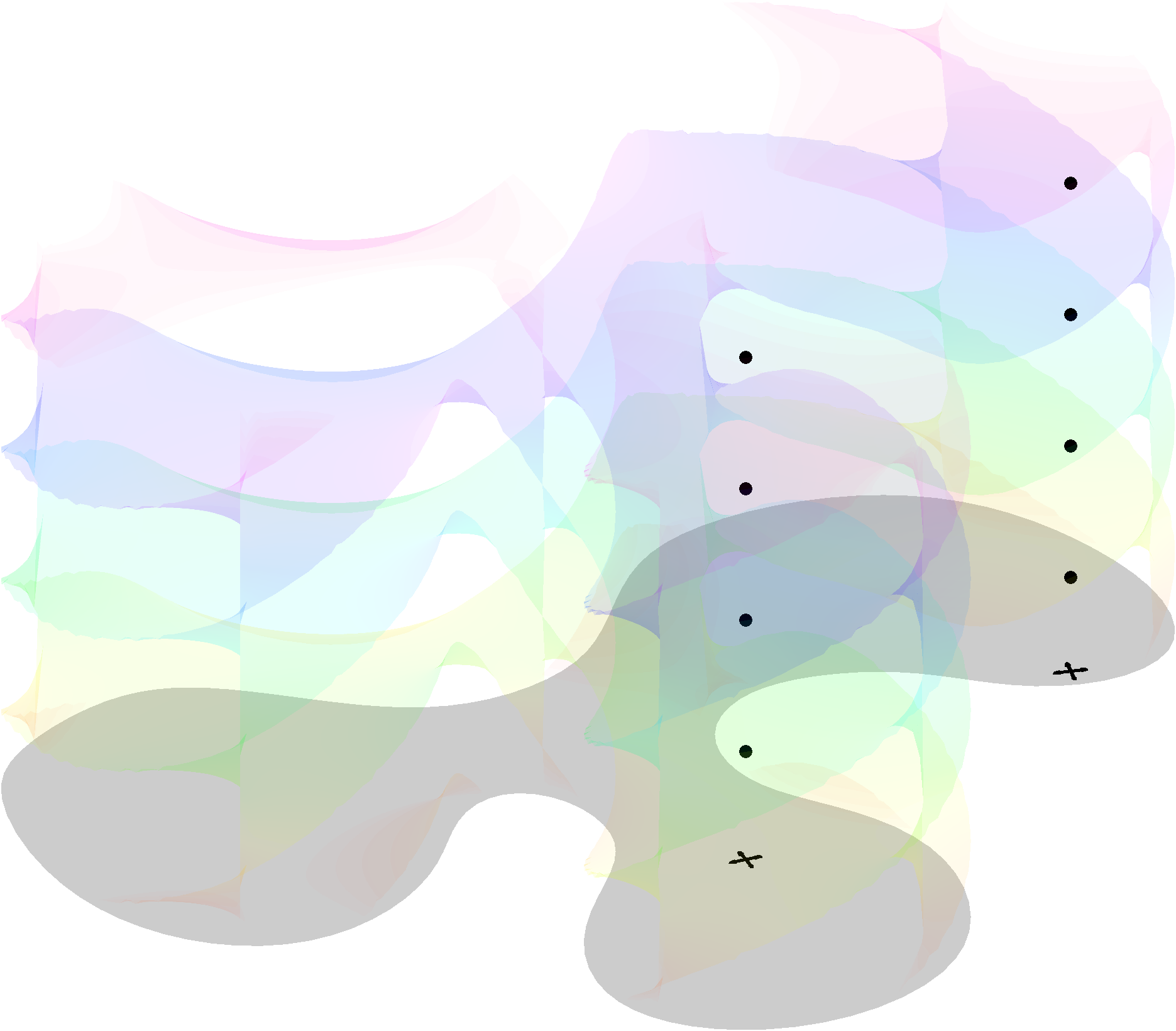}&%
\includegraphics[width=\imgwidth]{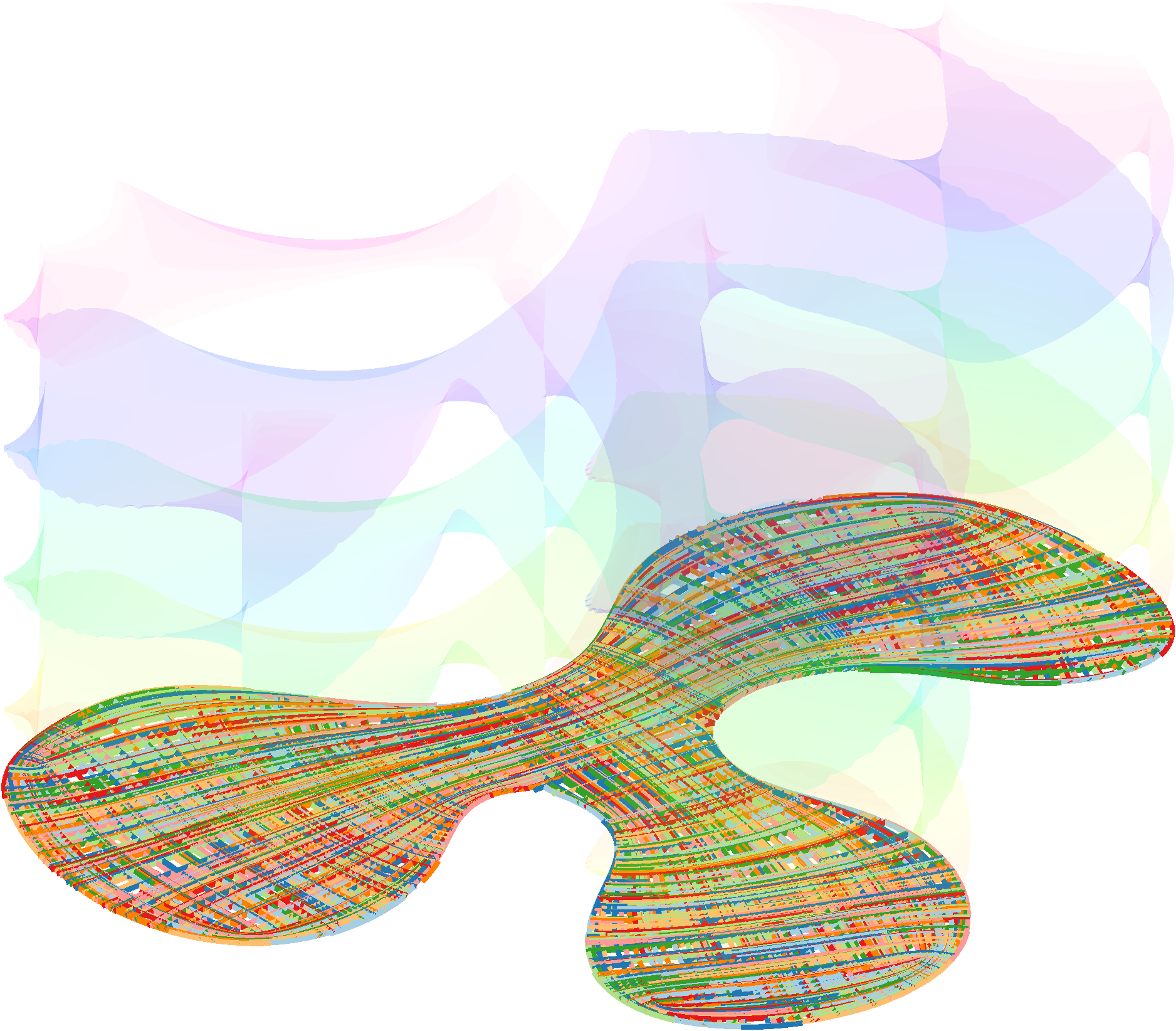}&%
\includegraphics[width=\imgwidth]{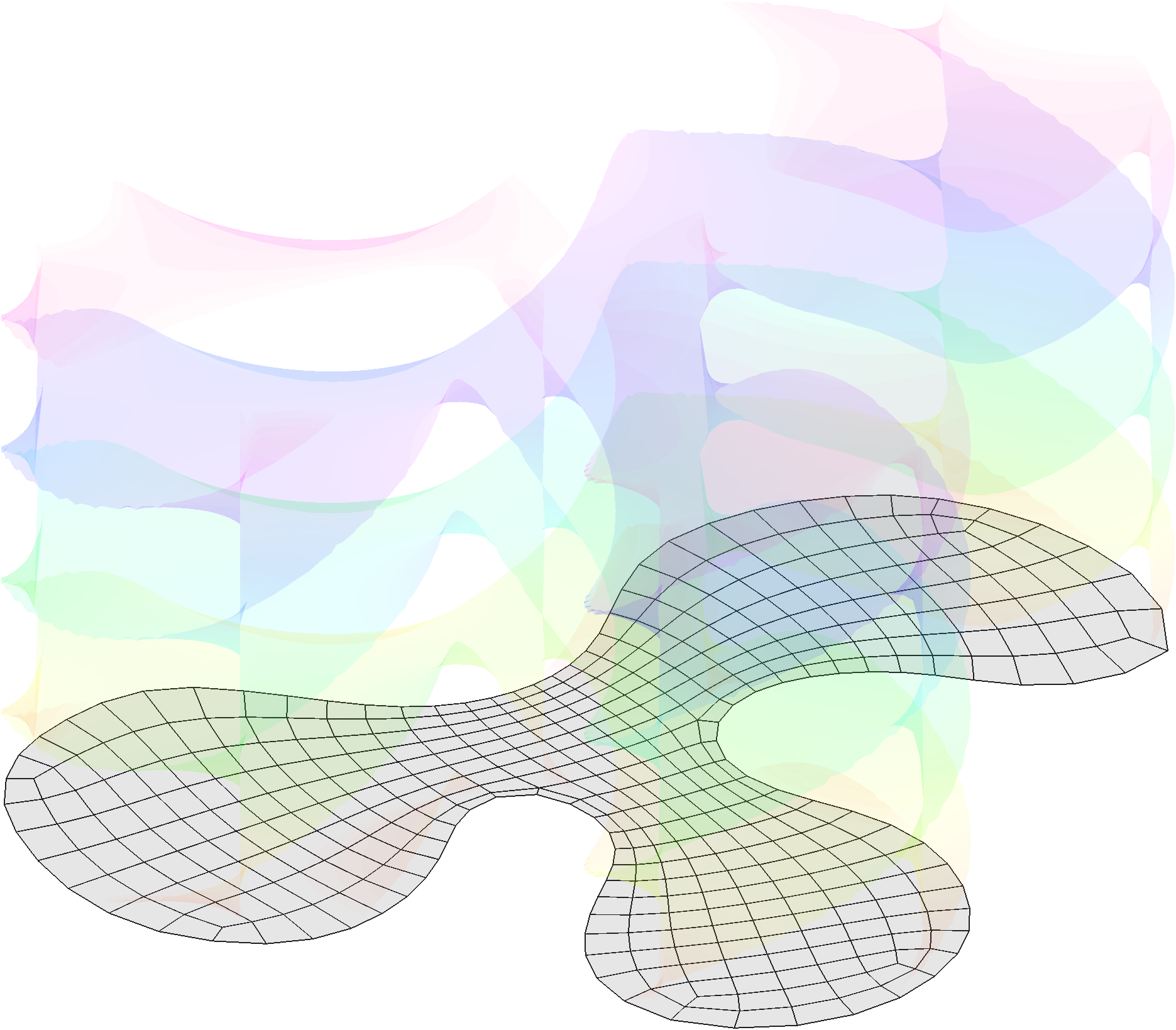}\\
{\footnotesize minimal current} & {\footnotesize extracted section} & {\footnotesize field values} & {\footnotesize integral curves} & {\footnotesize quad mesh}
\end{tabulary}
\caption{We compute directional fields by solving a convex relaxation of a minimal section problem in a circle bundle. After the field values are extracted from the computed section, they can be used for downstream applications such as quadrilateral meshing.}
\label{minsec:fig:teaser}

\end{teaserfigure}

\maketitle

\section{Introduction}
Directional fields (unit vector fields, line fields, cross fields, and their generalizations) are central to geometry processing pipelines. Unit vector fields and line fields are used in geodesic distance computations and texture synthesis. In field-based quadrilateral meshing, cross fields encode the spatially-varying orientations of quad elements and the placement of crucial singular points.

The presence of singularities is a robust feature of directional fields. Indeed, on closed surfaces, the Poincar\'e-Hopf index theorem guarantees that they will occur. In condensed matter physics and materials science, singularities are known as \emph{topological defects}, and their dynamics are crucial to understanding phenomena including plastic deformation and phase transitions.

Given the crucial role singularities play in the geometry of fields, optimizing the placement of singularities would seem a natural goal. But singularities pose a challenge for discrete variational methods. Existing approaches to field optimization largely avoid dealing with singularities, either requiring their placement to be specified in advance; modifying the space of field values so that singularities become zeros of a relaxed field; or ignoring the problem entirely, letting singularities slip through the cracks of a discretization. This is because singularities contradict a key assumption of these variational approaches: while traditional field optimization treats a field as a smooth function defined everywhere on its domain, singularities are precisely points where the field does not have a well-defined orientation.

In this work, we develop an approach to field optimization that
\begin{enumerate*}[label=(\roman*)]
	\item treats singularities as first-class citizens and explicitly represents their distribution as an optimization variable;
	\item optimizes a functional that is well-defined even in the presence of singularities; and
	\item solves a convex optimization problem to global optimality.
\end{enumerate*}
The main insight enabling our approach is to lift optimization of a field to optimization of its \emph{graph}, a surface in a fiber bundle sitting over the domain. From this point of view, singularities get lifted to boundary components of this surface. Optimization over the surface and its boundary can be treated as a \emph{minimal section} problem, related to the classical minimal surface problem. This problem admits a natural convex relaxation, replacing curves and surfaces by \emph{currents}. Over the space of currents, the minimal section problem becomes minimization of a convex norm subject to linear boundary constraints.

While this relaxation is convex and conceptually simple, making it practical requires representing currents on a circle bundle over a possibly curved surface. In general, such a bundle is not globally trivializable (expressible as a product space). Even locally, fibers cannot be ``glued together'' in a consistent way---this is the meaning of curvature. These considerations make it challenging to discretize such a bundle using finite elements. Instead, we exploit the fact that the fibers are circles and use Fourier decomposition in the vertical direction. It turns out that the relevant differential operators on the full bundle naturally decompose into covariant operators on the base surface, one for each frequency band. These operators can then be discretized with linear finite elements on the base, yielding a simple and fast \textsc{admm} algorithm whose most computationally expensive step parallelizes across frequencies.

In a particular limit under a fiber translation symmetry constraint, our optimization problem reduces to a sparse Poisson recovery problem like that posed in several previous works on computing cones for conformal mapping and quadrilateral meshing. Our work not only generalizes these methods, but also addresses the open problem of how to encourage cone-angle quantization in a convex manner.
Moreover, our explicit representation of the singularity measure makes it simple to impose constraints on singularity placement, such as restricting singularities to a subdomain, that would be difficult to enforce in previous frameworks.

Along with our convex relaxation, we propose an accompanying field extraction procedure that rounds a possibly diffuse current to a sharply-defined field. Like previous work in directional field design relying on convex relaxation, we offer no formal guarantees of exact recovery. But while previous work employs objective functionals that are ill-defined in the presence of singularities prior to relaxation, our objective is meaningful in both the unrelaxed and relaxed settings. We evaluate the quality of relaxed solutions in \Cref{minsec:sec:results} and leave to future work the problem of characterizing necessary conditions for exact recovery or bounding the rounding error.

In summary, using the tools of computational geometric measure theory, we provide a practical and efficient convex relaxation algorithm for optimization of directional fields with singularities.

\section{Related Work}

\subsection{Directional Fields and Singularities}

Directional fields are fields of collections of vectors tangent to a manifold. They include ordinary vector fields, unit vector fields, line fields, cross fields, and higher-multiplicity generalizations. They are central to many geometry processing algorithms, including to field-based quadrilateral and hexahedral meshing. \citet{vaxman_directional_2016} survey the construction and application of directional field design algorithms.

Methods for field design fall into two broad categories: connection methods, which require singularities to be placed in advance, and \emph{ab initio} methods, which place singularities automatically. Connection methods seek a \emph{trivial connection}, which is flat everywhere except at the specified singular points, around which its holonomies are specified \cite{craneTrivialConnectionsDiscrete2010,de2010trivial}. Connection methods have the advantage of only needing to solve a simple convex optimization problem. On the other hand, they leave the challenging task of optimizing singularity placement to the user. Recent work has generalized such methods to volumetric frame fields, though at the cost of convexity \cite{cormanSymmetricMovingFrames2019}.

In contrast, \emph{ab initio} field optimization methods treat directional fields as smooth functions (or sections) and represent singularities implicitly. However, unit directional fields with singularities have infinite Dirichlet energy, necessitating some kind of relaxation. Existing field optimization methods employ a relaxation in which field values are allowed to vary in their modulus. Scaling to zero at singularities renormalizes the energy at the cost of redefining the space of fields \cite{knoppel_globally_2013}. To try to preserve normalization, one can sacrifice convexity and adopt a penalty term of the Ginzburg-Landau type \cite{viertel_approach_2019}. Ginzburg-Landau--like approaches have recently been extended to three dimensions for the purpose of hexahedral meshing \cite{palmer_algebraic_2020,chemin}, where the penalty term becomes more complicated.

Thus existing methods either require singularities to be specified in advance or leave them implicit, whence they ``slip through the cracks'' of the computational mesh. In our approach, instead of relaxing the space in which our fields are valued, we relax optimization over fields to optimization over their graphs, viewed as currents. In doing so, we are able to make explicit the relationship between fields and singularities.

In explicitly modeling the distribution of singularities as a signed measure, this work bears some relation to a line of work viewing cross fields and their singularities from the perspective of conformal mapping. \citet{bunin_continuum_2008} proposes to view an infinitely fine quadrilateral mesh as a Riemannian metric whose curvature is flat except at a sparse set of \emph{cone singularities}, which coincide with the singularities of the mesh. As such, their total curvature is \emph{quantized} to multiples of $\pi/2$. In two dimensions, such a metric can be parametrized (relative to a fixed initial metric) by its log conformal factor, a scalar function related to the difference of curvatures by a Poisson equation (the Yamabe equation). Thus, the problem of computing a (continuous limit of a) quadrilateral mesh is reduced to a sparse, quantized inverse Poisson problem.

\citet{soliman_optimal_2018} treat the related problem of conformal mapping with arbitrary cone singularities (relaxing the quantization requirement) via convex optimization of a sparse inverse Poisson problem. As their work is targeted toward general conformal mapping, it does not produce quantized cones and thus cannot be used directly for quadrilateral meshing. In \Cref{minsec:subsec:symmetry}, we explain how a sparse inverse Poisson problem like that studied in \citet{soliman_optimal_2018} can be viewed as a limiting case of our minimal section problem.

In contrast to the convex method of \citet{soliman_optimal_2018}, if one wants to ensure cone quantization, one may resort to a variety of nonconvex optimization heuristics. \citet{myles_controlled-distortion_2013} use a greedy expansion procedure, progressively optimizing and rounding to re-impose quantization. \citet{farchi_integer-only_2018} select positions for quantized singularities via integer programming.

Finally, we note that computing a cross field often serves as a proxy for computing a ramified covering space encoding the topology of a mesh, as in \cite{bommes2009mixed,bommes2013integer} and the volumetric analogue \cite{nieser2011cubecover}. When our minimal current converges to a sharply defined surface, this surface serves as a concrete geometric instantiation of the desired covering space.

\subsection{Computational Geometric Measure Theory}
Geometric measure theory (\textsc{gmt}) was originally studied from a purely theoretical perspective as a tool for showing existence and regularity of minimal surfaces, among other problems. Several prior works have explored \textsc{gmt} as a source of convex relaxations in computational optimization. Most closely related to our work, \citet{wang_computing_2021} solve Plateau's minimal surface problem discretized on a regular grid by convex techniques, employing the alternating direction method of multipliers (\textsc{admm}) and discretizing differential operators via fast Fourier transforms.

The idea of lifting optimization problems over maps to optimization over graphs was a key insight in the work of \citet{mollenhoff_lifting_2019}, which proposes convex relaxations of mapping problems important in computer vision. In our work, we extend this idea to optimize over sections with singularities.
Lifting mapping problems to optimization over submanifolds of a bundle or product space has also been a paradigm in the functional maps literature \cite{rodolaFunctionalMapsRepresentation2019}. \citet{solomon_optimal_2019} interpolate between directional fields by interpolating a singularity measure \emph{horizontally} via optimal transport and smoothly interpolating vectors \emph{vertically} away from singularities. In our work, the boundary coupling between sections and their singularities unifies these horizontal and vertical modes of variation.

In medical imaging, currents have been used to model deformation and compute registration \cite{Charon:2014:FC, Vaillant:2005:SMVC, Glaunes:2004:DMD, Durrleman:2008:DDC, Durrleman:2009:SMC, Durrleman:2011:RAEVA}. In computer vision, the language of currents has been used to devise a representation for surfaces with boundary \cite{palmer_deepcurrents_2022} and to measure distances between surfaces sensitive to their orientations \cite{mollenhoff2019flat}. \citet{roweSparseStressStructures2023} take a computational \textsc{gmt} approach to topology optimization, relaxing optimization over stress-bearing structures to the space of varifolds.

\subsection{Theoretical Work on Minimal Sections}
Theoretical work has studied the regularity and topological properties of minimal sections in circle bundles over surfaces \cite{chacon_minimal_2011} and in more general fiber bundles \cite{johnson_regularity_1995,johnson_partial_2008}. This work provides strong guarantees on the regularity and topology of minimal sections, e.g., that they are continuous graphs almost everywhere. However, in these theoretical analyses the space of sections is usually constructed out of so-called integral currents. In practice, we can only optimize over the larger space of normal (bounded-mass) real currents, a convex relaxation. It would be beneficial to know if and when this convex relaxation is exact, i.e., if it has integral optima, and therefore if the regularity and topology guarantees extend to the solutions we can compute in practice. But to the authors' knowledge, the minimal section problem sits in a gap between the regime of known exact recovery and that of known counterexamples.

For the classical minimal surface problem in $\R^3$, it has been known for a long time that the relaxation to real currents is exact \cite{federer_real_1974}, i.e., that there exist integral minimizers. Recent work has illuminated this result with techniques from optimal transport \cite{brezis_plateau_2019} and shown that the integral minimizers occur precisely at extreme points of the convex set of minimizers, as would be expected from general intuition about convex relaxations.

The minimal section problem differs from the classical Plateau problem in that it features a free boundary component, resulting in coupled optimization over currents of differing dimension. A closely related problem that shares this feature is the \emph{flat norm decomposition} problem. For flat norm decomposition of integral currents, the story is mixed. Integral flat norm decompositions are known to achieve optimality for some integral $d$-currents in $\R^{d+1}$, but counterexamples are known for higher codimensions---as is true of the minimal mass problem in general codimension \cite{ibrahim_flat_2016}.

As we employ the language of currents to model topological defects in directional fields, we note that currents have also been employed as a theoretical tool in the study of elastic defects---specifically, dislocations and their interactions have been modeled using \emph{Cartesian currents} like those that arise in the minimal section problem \cite{scala_currents_2016,scala_constraint_2016,hudson_existence_2018,scala_variational_2019,scala_analytic_2020}.

\section{Intuition} \label{minsec:sec:intuition}

Recall that a \emph{directional field} assigns a collection of tangent vectors to each point on a surface. In a \emph{unit} directional field, the component vectors all have unit norm, thus falling on the unit circle in each tangent space. On a flat domain, all the tangent spaces (and accordingly their unit circles) can be identified, but on a curved surface these individual spaces form distinct fibers of the tangent bundle (resp.\ unit tangent bundle). The unit tangent bundle is a circle bundle, and line fields, cross fields, and so on also live in circle bundles, making circle bundles a natural setting for field optimization.

Formally, given a base surface $B$, a fiber bundle $\pi: E \to B$ is made up of \emph{fibers} $F \cong \pi^{-1}(x)$ in such a way that it locally looks like a product space, $\pi^{-1}(U) \cong U \times F$ over any small neighborhood $U$. A circle bundle is a fiber bundle with $F = \Sph^1$. The \emph{unit tangent bundle} $\Sph TB$ of $B$ is a circle bundle given as the locus of all unit vectors in the tangent bundle $TB$. To keep this article self-contained, we review the differential geometry of circle bundles in \Cref{minsec:sec:math-prelim}.

A smooth unit vector field on $B$ is a \emph{global section} of $\Sph TB$, and unit line fields, cross fields, etc., are sections of circle bundles that can be constructed from tensor powers of $TB$. Given a fiber bundle $E$, a section is a map $\sigma : B \to E$ that ``lifts'' each point in the base to a point in its fiber so that $\pi \circ \sigma = \id_B$. In other words, the graph of a smooth section $\Sigma = \gr \sigma$ intersects each fiber exactly once.

\begin{wrapfigure}{L}{0.3\columnwidth}
\centering
\includegraphics[width=0.29\columnwidth]{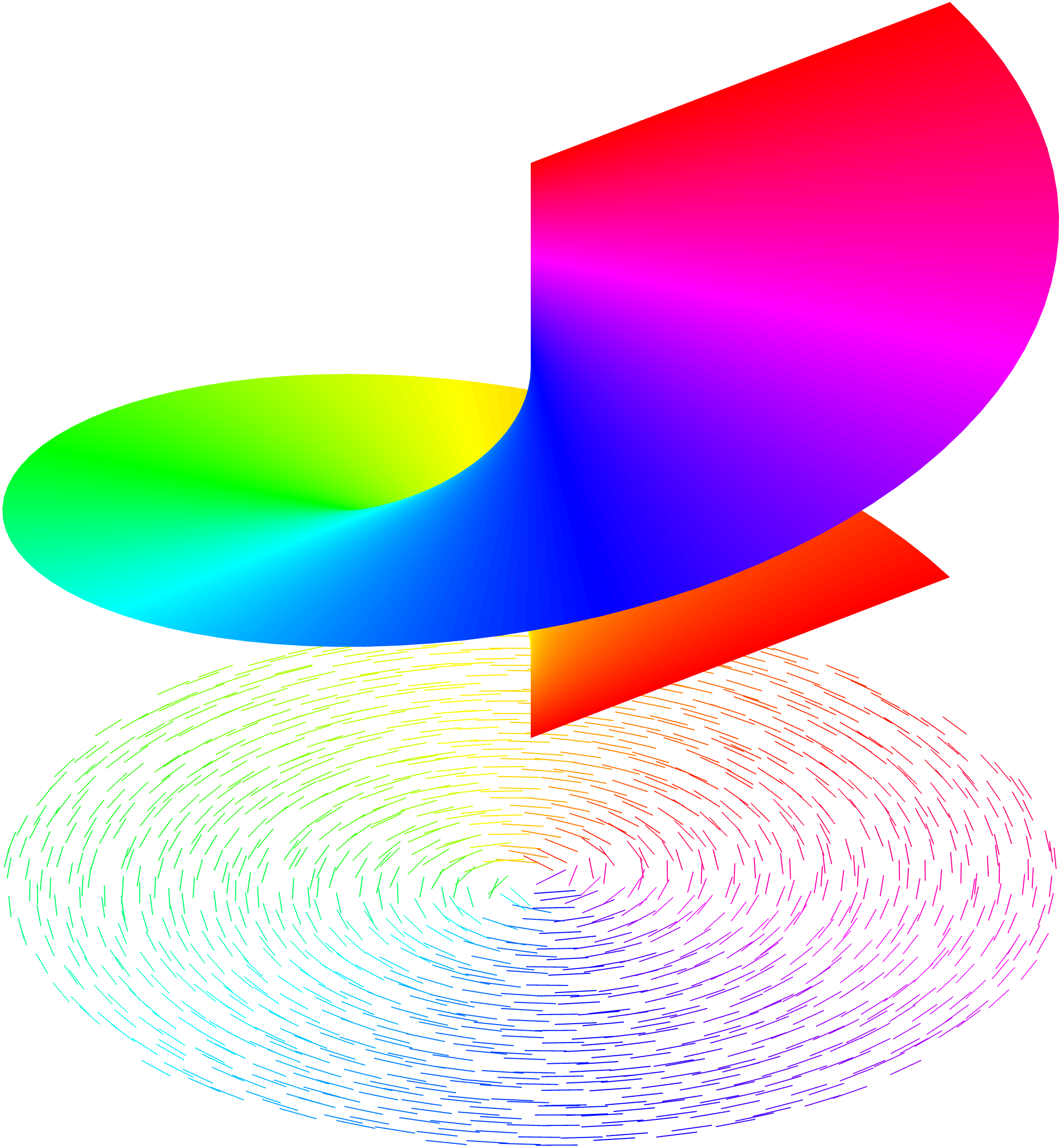}
\caption{The graph of a field around a singularity forms a helicoid.}
\label{minsec:fig:helicoid}
\end{wrapfigure}
Around a singularity, this picture breaks down. Any arbitrarily small loop encircling the singular point must lift to a loop that wraps around the fiber direction, encoding the holonomy of the singularity. As such, the section $\sigma$ cannot be extended continuously over the singular point. On the other hand, the surface $\Sigma$ \emph{can} be extended to subtend the entire circular fiber over the singular point. Thus, in the neighborhood of a singular point, $\Sigma$ is no longer strictly a graph, but rather looks like a helicoid, a periodic minimal surface with a vertical boundary component (\Cref{minsec:fig:helicoid}).

The core of our approach is to replace optimization of a field $\sigma$ with optimization of a \emph{current} or generalized surface, serving as a proxy for its graph $\Sigma$. Expressed in this language, the relationship between the field and its singularities will become a boundary condition. Working with currents will allow us to formulate an optimization problem over fields that remains well-defined in the presence of singularities. 

To begin with, suppose we had a smooth section $\sigma \in \Gamma(E)$ with no singularities. Then its \emph{graph}
\begin{equation}
	\Sigma = \gr \sigma \coloneqq \sigma(B) \subset E
\end{equation}
would be a smooth surface embedded in $E$.
The area of $\Sigma$ naturally generalizes common field optimization energies such as (covariant) Dirichlet energy. We state the result below and defer the proof to \Cref{minsec:subsec:bundle-metric}, along with relevant mathematical background.

\begin{proposition} \label{minsec:prop:area-func}
	The area of the graph $\Sigma$ can be evaluated as the following integral over the base surface:
	\begin{equation}
		\Area(\Sigma) = \int_B \sqrt{1 + r^2 |\Dext \sigma|^2} \; \dext A_B,
		\label{minsec:eq.graph-area}
	\end{equation}
	where $\Dext$ denotes covariant differentiation and $r$ is the fiber radius.
\end{proposition}

The fiber radius $r$ is a free parameter controlling the behavior of the area functional as a field energy. When the fibers are very long $(r \to \infty)$, the area of the section graph
\begin{equation}
	\Area(\Sigma) \approx r \int_B |\Dext\sigma| \; \dext A_B
\end{equation}
behaves like (covariant) total variation. On the other hand, as $r \to 0$,
\begin{equation}
	\Area(\Sigma) \approx \int_B \left[1 + \frac{r^2}{2}|\Dext\sigma|^2 \right] \; \dext A_B,
	\label{minsec:eq.graph-area-reduced}
\end{equation}
equal to the (covariant) Dirichlet energy up to additive and multiplicative constants. Similarly, the area functional behaves like Dirichlet energy in regions of slow variation and like total variation where the field varies rapidly (see \Cref{minsec:fig:schematic-energy}). In \Cref{minsec:app:scaling}, we examine further the scaling behavior of the area functional in the presence of singularities.
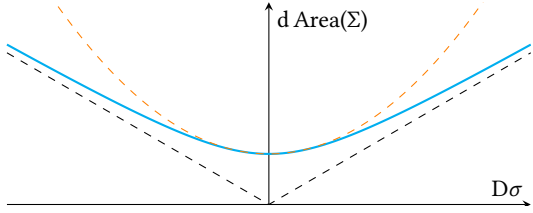
\begin{figure}
\centering
\tikzsetnextfilename{area-functional}
\begin{tikzpicture}
	\begin{axis}[
		width=\columnwidth, height=0.5\columnwidth,
        xmin=-1, xmax=1, 
        ymin=0, ymax=4,
        axis x line = center, 
        axis y line = center,
        xtick = \empty,
        ytick = \empty,
        xlabel = {$\Dext\sigma$},
        ylabel = {$\dext\Area(\Sigma)$},
        legend style = {nodes=right},
        legend pos = north east,
        clip mode = individual,
    ]
            \addplot[cyan, thick, samples=100, domain=-1:1]{sqrt(1+9*x^2)};
            \addplot[black, dashed, domain=-1:1]{abs(3*x)};
			\addplot[orange, dashed, samples=100, domain=-1:1]{1+(9/2)*x^2};
	\end{axis}
\end{tikzpicture}
\caption{The area functional acts like Dirichlet energy where $\sigma$ varies slowly and total variation where it varies rapidly, such as near singularities.}
\label{minsec:fig:schematic-energy}
\end{figure}

While Dirichlet energy is infinite in the presence of singularities, the area functional remains well-defined and finite. This suggests using area as a regularization of Dirichlet energy.
Moreover, the area functional extends naturally to a convex norm on the space of currents, known as the \emph{mass norm}. (See \Cref{minsec:subsec:gmt} for background material on currents and \textsc{gmt}.) In the next section we use this norm to formulate a convex relaxation of field optimization.

\section{Problem Formulation} \label{minsec:sec:problem}
\subsection{Minimal Sections}

We propose to compute fields with singularities by solving the following convex optimization problem, which we call the \emph{generalized minimal section} problem:
\begin{equation} \begin{aligned}
\min_{\Sigma, \Gamma} \quad &\Mass(\Sigma) + \lambda \Mass(\Gamma)\\
 \subj \quad &\partial \Sigma + \Gamma_0 + \Sph^1 \times \Gamma = 0\\
 &\pi_\# \Sigma = B \\
 &\Sigma \in \mathcal{D}_2(E) \\
 &\Gamma \in \mathcal{D}_0(B).
 \end{aligned} \tag{\textsc{gms}} \label{minsec:prob.sym}
\end{equation}
Here $\Gamma_0$ is a fixed boundary curve, viewed as a current, sitting over $\partial B$ and encoding Dirichlet boundary data. The $2$-current $\Sigma$ represents the graph of a section of the circle bundle $E$, wherever the section is defined---that is, away from its singularities, which are encoded in the measure $\Gamma$. The mass norm of $\Sigma$ acts like a field energy as detailed in \Cref{minsec:sec:intuition}, while the mass norm of $\Gamma$ serves as a regularizer, controlling the cost of singularities. The boundary constraint says that apart from the fixed Dirichlet data, the section may only have vertical boundary components sitting over the measure $\Gamma$. This is precisely the behavior of a section around a topological defect (\Cref{minsec:fig:helicoid}). Finally, the projection constraint ensures that $\Sigma$ behaves like a graph, its shadow covering the entire base manifold.

The problem \eqref{minsec:prob.sym} is a generalization of what appears in \cite{chacon_minimal_2011}, which leaves out the $\Mass(\Gamma)$ regularizer. It is also a convex relaxation: we optimize over general currents $\mathcal{D}_p$ rather than the integral currents $\mathcal{I}_p$ studied in theoretical work.

\subsection{Representation by Smooth Forms}
In this section, we take the abstract problem \eqref{minsec:prob.sym} and prepare it for discretization by re-expressing it in terms of differential forms. Rather than achieving this by taking the pre-dual, we rewrite the primal problem in terms of forms, which will allow us to take advantage of Hodge decomposition.

While the solutions we seek are singular currents supported on submanifolds of the full bundle $E$, we will represent our candidate solutions by smooth forms on $E$. On any ambient manifold $X$, there is an embedding $[\cdot] : \Omega^{n-p}(M) \to \mathcal{D}_p(X)$ given by the wedge pairing:
\begin{equation}
\langle[\alpha], \beta \rangle \coloneqq \int_X \alpha \wedge \beta.
\end{equation}
For a current $\Xi = [\xi]$, we will denote $\xi = \llbracket \Xi \rrbracket$. Where it is not confusing, we will elide the brackets and simply represent our candidate $p$-currents by $(n-p)$-forms.

With this in mind, we can translate our boundary constraint into the language of exterior calculus:
\begin{equation}
\begin{aligned}
	\langle \partial \Sigma, \alpha\rangle \coloneqq \langle \Sigma, \dext \alpha\rangle
	&= \int_E \Sigma \wedge \dext \alpha \\
	&= -\int_E \dext (\Sigma \wedge \alpha) + \int_E \dext \Sigma \wedge \alpha \\
	&= \langle\dext\Sigma, \alpha \rangle_E - \langle \Sigma, \alpha \rangle_{\partial E},
\end{aligned}
\end{equation}
where $\Sigma$ is now viewed as a $1$-form representing a $2$-current. Equating interior and boundary terms, we have
\begin{align}
	&\dext \Sigma = -\Sph^1 \times \Gamma \\
	&\Sigma \mid_{\partial E} = \Gamma_0
\end{align}
By definition,
\begin{equation}
\begin{aligned}
	\langle\Sph^1 \times \Gamma, \alpha_V \tau + \alpha_H \rangle
	&= \langle \Gamma, \langle \Sph^1, \alpha_V \tau\rangle \rangle \\
	&= \int_B \llbracket \Gamma \rrbracket \int_F \alpha_V \tau \\
	&= \int_E \pi^*\llbracket \Gamma \rrbracket \wedge(\alpha_V \tau + \alpha_H),
\end{aligned}
\end{equation}
where the last equality is by Fubini splitting and the fact that $\pi^*\llbracket \Gamma \rrbracket \wedge \alpha_H = 0$ since both are horizontal and there are no horizontal $3$-forms. Thus, writing everything in terms of forms, we have the exterior system:
\begin{alignat}{2}
	&\dext \Sigma = -\pi^*\Gamma \quad &&\in \Omega^2(E)\\
	&\Sigma \mid_{\partial E} = \Gamma_0 &&\in \Omega^1(\partial E).
\end{alignat}

The optimization problem \eqref{minsec:prob.sym} becomes:
\begin{equation}
\begin{aligned}
 \min \quad &\int_E \hodge |\Sigma|_g + \lambda \int_B \hodge |\Gamma| \\
 \subj \quad &\dext \Sigma = -\pi^* \Gamma \\
 &\Sigma \mid_{\partial E} = \Gamma_0 \\
 &\int_{\pi^{-1}(x)} \Sigma = 1 \quad \forall x \in B \\
 &\Sigma \in \Omega^1(E), \Gamma \in \Omega^2(B),
\end{aligned}
\end{equation}
with $g$ the Sasaki metric. We can simplify the constraints by para\-metr\-iz\-ing solutions using Hodge(-Morrey-Friedrichs) decomposition (\emph{cf.}\ \cite{Zhao:2019:3DH}). As $\int_{\pi^{-1}(x)} \tau = 2\pi$, we can write:
\begin{equation}
	\Sigma = \bar{\tau} + \dext f + \delta \beta_t, \label{minsec:eq:hodge-decomp}
\end{equation}
where $f \in C^\infty(E)$, $\bar{\tau} = \tau/2\pi$, and $\beta_t$ is a $2$-form tangential to the boundary. Note that this is not precisely a Hodge decomposition: though $\bar\tau$ will not in general be harmonic, it does have the correct cohomology. On a multiply-connected base surface, additional homological constraints may necessitate additional harmonic terms.

Taking the exterior derivative of \eqref{minsec:eq:hodge-decomp} and using that $\dext \tau = -\pi^*\kappa$, the pullback of the negative Gauss curvature $2$-form on the base, we obtain
\begin{equation}
	\dext \delta \beta_t = \dext \Sigma - \dext\bar{\tau} = \pi^* \left(-\Gamma + \frac{1}{2\pi}\kappa\right) \eqqcolon \pi^* (\bar\kappa - \Gamma). \label{minsec:eq:beta-constraint}
\end{equation}
Out of many possible choices for the free variable $\beta_t$, we select for convenience $\beta_t = -\pi^* (\hodge\phi)$, where
\begin{equation}
\hodge \Delta^0_B \phi = \hodge \delta \dext \phi = \dext \delta (\hodge\phi) = \Gamma - \bar\kappa
\end{equation}
and $\phi\mid_{\partial B} = 0$. This choice satisfies the constraint \eqref{minsec:eq:beta-constraint} because $\dext$ commutes with $\pi^*$, as does $\delta$ under our choice of metric on $E$ (see \Cref{minsec:prop:horiz-base}).

Using that $\delta \hodge \phi = \hodge^{-1}\dext \hodge \hodge \phi = - \hodge \dext \phi$, and replacing $\Gamma$ and $\bar\kappa$ by their Hodge duals on the base, we obtain an optimization problem of the form
\begin{equation}
\begin{aligned}
 \min \quad &\int_E \hodge |\Sigma|_g + \lambda \int_B \hodge |\Gamma| \\
 \subj \quad &\Sigma = \bar\tau + \dext f + \pi^*\hodge \dext \phi \\
 &\Delta \phi = \Gamma - \bar\kappa\\
 &\Sigma_V \ge 0\\
 &\Sigma \mid_{\partial E} = \Gamma_0 \\
 &\phi \mid_{\partial B}  = 0,
\end{aligned}
\tag{P}
\label{minsec:prob.primal}
\end{equation}
where we've added the additional positivity constraint $\Sigma_V \ge 0$ to prevent the section from ``folding over'' on itself.

\section{Optimization}
\label{minsec:sec:opt}

In this section, we introduce our \textsc{admm}-based optimization algorithm in the continuous setting. The analysis will result in a global step in which we solve a system of \textsc{pde}s, a local step that decomposes pointwise, and a dual update. The \textsc{pde}s will motivate our discretization in \Cref{minsec:sec.discrete}.

To derive the \textsc{admm} iterations, we introduce (scaled) dual variables $w$ and $z$ associated to the \textsc{pde} constraints in \eqref{minsec:prob.primal}, arriving at the following augmented Lagrangian in scaled form; we elide boundary conditions for the moment:
\begin{equation}
\begin{aligned}
	\mathcal{L}(\Sigma, \Gamma, f, \phi; w, z) &= \int_E \hodge |\Sigma|_g + \lambda \int_B \hodge |\Gamma| + \chi(\Sigma_V \ge 0) \\
	&+ \frac{\mu}{2}\|\Sigma - \bar\tau - \dext f - \pi^* \hodge\dext\phi + w\|^2 - \frac{\mu}{2}\|w\|^2 \\
&+ \frac{\nu}{2}\| \Gamma - \bar\kappa - \Delta \phi + z\|^2 - \frac{\nu}{2}\|z\|^2.
\end{aligned}
\label{minsec:eq:lagrangian}
\end{equation}
Here $\chi$ denotes the convex indicator function taking the value $0$ if the constraint is satisfied and $+\infty$ otherwise.

Then \textsc{admm} consists of the following steps, repeated until convergence \cite{boyd2011distributed}:
\begin{itemize}
	\item Global step: $f, \phi \gets \argmin_{f, \phi} \mathcal{L}$ 
	\item Local step: $\Sigma, \Gamma \gets \argmin_{\Sigma, \Gamma} \mathcal{L}$ 
	\item Dual update: take a gradient ascent step for the dual variables
	\begin{equation}
		\begin{aligned}
			w &\gets w + \Sigma - \bar\tau - \dext f - \pi^* \hodge \dext \phi \\
			z &\gets z + \Gamma + \bar\kappa - \Delta \phi
		\end{aligned}
	\end{equation}
\end{itemize}
In what follows, we examine the global and local steps in more detail.

\paragraph{Global step.} Grouping $f$ and $\phi$ together, the subproblem reads
\begin{equation}
\begin{aligned}
 \min_{f, \phi} \quad &\frac{\mu}{2}\|\Sigma - \bar{\tau} - \dext f - \pi^* \hodge\dext \phi + w\|^2 + \frac{\nu}{2}\| \Gamma - \bar{\kappa} - \Delta \phi + z\|^2 \\
 \subj \quad &\bar\tau  + \dext f + \pi^*\hodge\dext\phi \mid_{\partial E} = \Gamma_0 \\
 & \phi \mid_{\partial B} = 0.
\end{aligned}
\label{minsec:prob.fphi}
\end{equation}
Taking first variations with respect to $f$ and $\phi$ compactly contained within the interior of $E$, and using the orthogonality of Hodge decomposition, we obtain the following Euler-Lagrange equations:
\begin{align}
	&\Delta^0_E f = \delta_E^1 (w + \Sigma - \bar\tau) \label{minsec:eq.f-pde}  \\
	&(\mu\ell \Delta + \nu\Delta^2)\phi = \mu\ell \dext P_\tau (w + \Sigma - \bar\tau) + \nu \Delta (\Gamma + z - \bar\kappa) \label{minsec:eq.phi-pde} \\
	&\bar\tau  + \dext f + \pi^*\hodge\dext\phi \mid_{\partial E} = \Gamma_0 \label{minsec:eq.bdry1} \\
 & \phi \mid_{\partial B} = 0. \label{minsec:eq.bdry2}
\end{align}
Here we have used \Cref{minsec:prop:horiz-base} and the results of \Cref{minsec:subsec:forms-metric} to express the adjoint of $\pi^*$ in terms of the projection $P_\tau$ (\emph{cf.} \Cref{minsec:subsec:bundle-decomp}).
This is a system of linear elliptic partial differential equations in which $f$ and $\phi$ are coupled only through the boundary condition \eqref{minsec:eq.bdry1}. The $\phi$ equation \eqref{minsec:eq.phi-pde} is a fourth-order biharmonic equation like those studied in \cite{stein_natural_2018,stein2019mixed}, and it only involves quantities on the base $B$, making it suitable for discretization via ordinary finite elements. In contrast, \eqref{minsec:eq.f-pde} is a Poisson equation on the total space $E$. The curvature of the bundle $E$ frustrates the construction of finite elements; instead, we exploit the homogeneous fiber bundle structure in a hybrid spectral discretization in \Cref{minsec:subsec:bundle-fourier,minsec:subsec:global-step-discrete}.

\paragraph{Local step.} In the local step, we solve for $\Sigma$ and $\Gamma$ given the values of $f$ and $\phi$. These problems decouple as follows:
\begin{align}
	&\begin{alignedat}{2}&\argmin_{\Sigma}& &\int_E \left[\frac{\mu}{2}|\Sigma - \bar{\tau} - \dext f - \pi^* \hodge\dext \phi + w|^2 + |\Sigma| \right]\; (\hodge 1) \\
	&\subj \quad &&\Sigma_V \ge 0
	\end{alignedat}\\ \label{minsec:eq.local-step}
	&\argmin_\Gamma \quad \int_B \left[\frac{\nu}{2}|\Gamma - \bar{\kappa} - \Delta \phi + z|^2 + \lambda|\Gamma| \right]\; (\hodge 1).
\end{align}
Moreover, these subproblems depend only on the pointwise values of $\Sigma$ and $\Gamma$, not on their spatial derivatives. Thus, when discretizing, we decompose them into independent pointwise problems (whence the terminology \emph{local step}).

\section{Discretization}
\label{minsec:sec.discrete}
The optimization problem \eqref{minsec:prob.primal} poses a challenge for discretization because evaluating mass norms requires pointwise values of the forms $\Sigma$ and $\Gamma$, while the \textsc{pde} constraints are more naturally discretized in their weak form using finite elements. As \textsc{admm} decouples the \textsc{pde} from the local proximal step, we employ a compromise where the $f$ and $\phi$ fields computed in the global step are discretized with finite elements, while the local proximal steps for $\Gamma$ and $\Sigma$ are evaluated pointwise at edges and triangle corners, respectively. This placement of degrees of freedom attempts to minimize the amount of averaging between \textsc{admm} steps and thus promote convergence to sharply defined currents.

\subsection{Operator Decomposition}
\label{minsec:subsec:bundle-fourier}
\begin{figure}
\centering
\newcommand{\imgwidth}{0.23\columnwidth}
\begin{tabulary}{\columnwidth}{@{}CCCC@{}}
\includegraphics[width=\imgwidth]{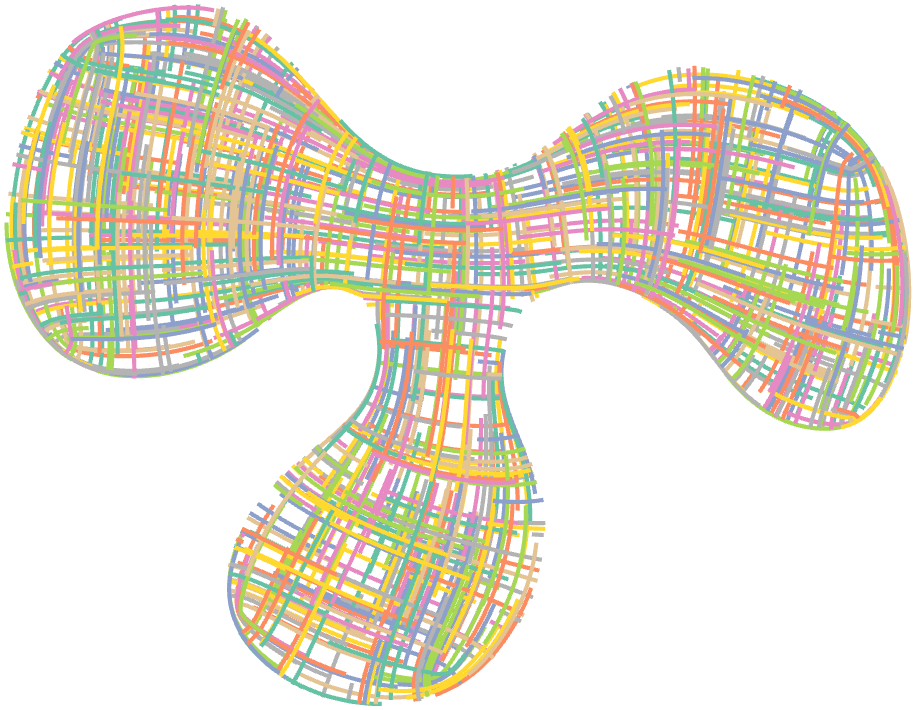} &
\includegraphics[width=\imgwidth]{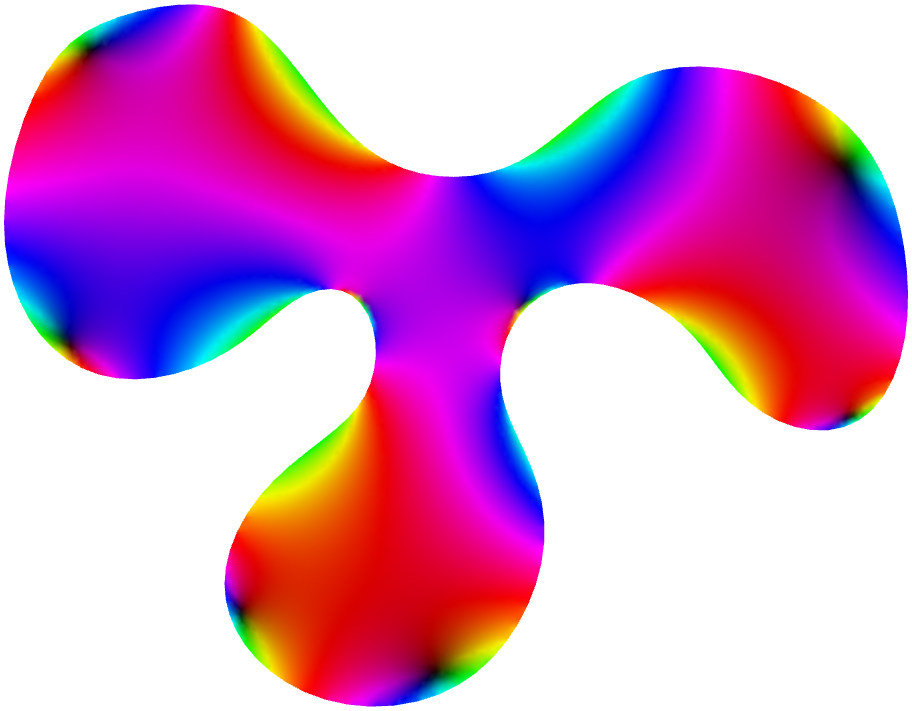} &
\includegraphics[width=\imgwidth]{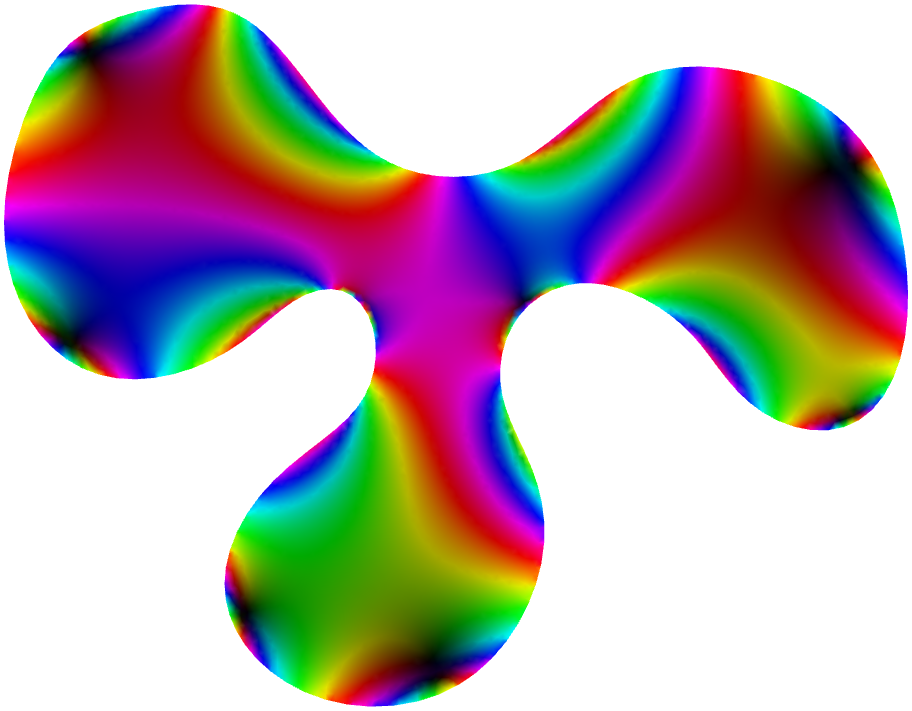} &
\includegraphics[width=\imgwidth]{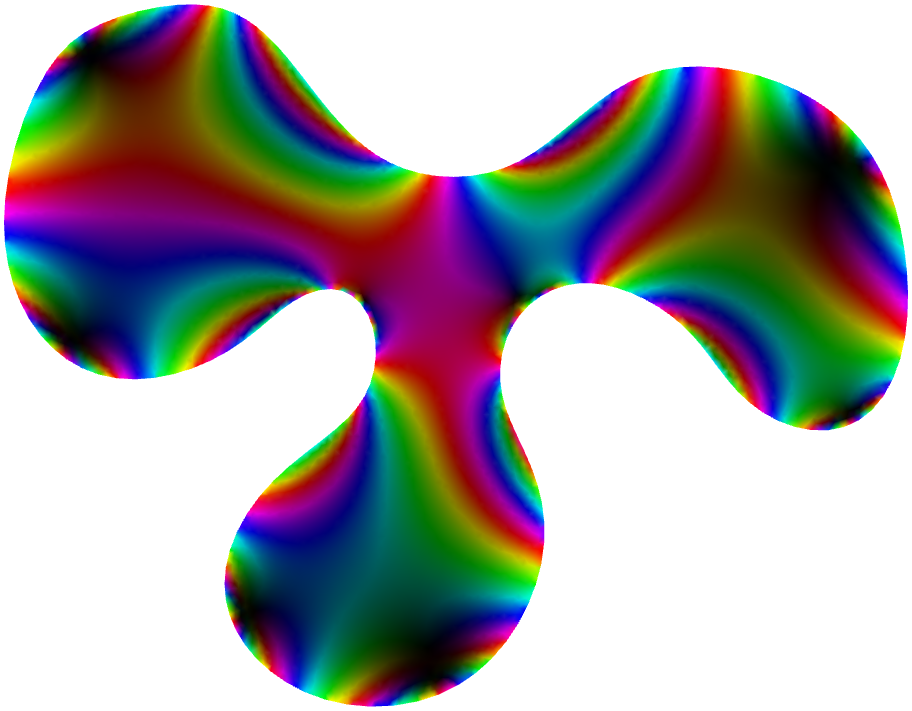} \\
& $k = 1$ & $k = 2$ & $k = 3$
\end{tabulary}
\caption{Vertical Fourier components of a minimal section corresponding to the cross field on the left drop to zero in the neighborhood of its singularities.}
\label{minsec:fig:schematic-fourier}
\end{figure}
A further challenge occurs in discretizing the differential operators on the bundle $E$. A typical discretization might rely on a mesh of the bundle. A natural choice would be to use prismatic elements---triangles of the base mesh extruded along the fiber direction. However, the cross-sections of such a mesh at fixed positions along the fiber direction would amount to nonsingular global sections of the bundle, an impossibility when $E$ is the unit tangent bundle of a closed surface of genus $g \ne 1$. Even in a local region, such a mesh cannot be constructed so that its vertical edges and horizontal triangles are orthogonal under the Sasaki metric. Indeed, the triangular faces of such a mesh would need to form horizontal sections, which can only exist, even locally, when the curvature is zero.

To address these problems, we take advantage of the homogeneous structure of the circle bundle, which allows us to represent forms and currents on the bundle by their vertical Fourier components (see \Cref{minsec:fig:schematic-fourier}). It turns out that in the Fourier basis, various exterior calculus operators on $E$ become block-diagonal, with blocks consisting of corresponding \emph{covariant} operators on the base. We then discretize these covariant operators using exact parallel transport on the base, avoiding meshing the bundle entirely.

As in \Cref{minsec:sec:intuition}, let $\pi: E_U \to U$ be a local trivialization. Choose coordinates $\{\mathbf{x}, \theta\}$ on $E_U$ such that $\vvf = \partial_\theta$, and $\theta \in [0, 2\pi)$ with the ends identified. 
Let $\xi$ be a form with the decomposition $\xi = \tau \wedge \alpha + \beta \in \Omega^p(E_U)$ (see \eqref{minsec:eq:form-decomp}). We can further decompose $\alpha$ and $\beta$ by taking their Fourier components in the vertical direction:
\begin{equation} \xi = \sum_{k} e^{ik\theta} [\tau \wedge \alpha_k + \beta_k], \label{minsec:eq.bundle-fourier} \end{equation}
where $\alpha_k$ and $\beta_k$ are horizontal homogeneous $(p-1)$- and $p$-forms, respectively.

Define the covariant exterior derivative at degree $p$ and frequency $k$ as follows:
\begin{equation} \Dext^p_k \coloneqq \dext^p + ik(\dext \theta - \tau)\wedge. \end{equation}
Note that $\dext \theta - \tau$ is horizontal homogeneous and thus corresponds to a form on the base. So $\Dext^p_k$ can be viewed as a linear differential operator on complex-valued forms on the base.
Then the exterior derivative on the total space decomposes as follows (see \Cref{minsec:app.op-decomp} for details):
\begin{equation} \dext^p \sum_k e^{ik\theta} \xi_k = \sum_k e^{ik\theta} \dext^p_k \xi_k, \end{equation}
where
\begin{equation} \begin{pmatrix} \Pi_H \\ \iota_{\vvf} \end{pmatrix} \dext^p_k 
\coloneqq \begin{pmatrix}
\Dext^p_k & \dext \tau \wedge \\
ik & -\Dext^{p-1}_k
\end{pmatrix} \begin{pmatrix} \Pi_H \\ \iota_{\vvf} \end{pmatrix}
\end{equation}
Similarly, the codifferential decomposes into blocks of the form
\begin{equation}
\delta_k^p = \begin{pmatrix}
	(-1)^p \hodge^{-1} \Dext_k \hodge & -r^{-2}ik \\
	r^2 \hodge^{-1}(\dext\tau\wedge)\hodge & (-1)^p \hodge^{-1}\Dext_k\hodge
\end{pmatrix},
\end{equation}
where the basis (projection) matrix is understood.
Finally, the Hodge Laplacian $\Delta^p_E$ breaks into blocks
\begin{equation}
\Delta_k^p = \begin{psmallmatrix}
	\Delta_{\Dext_k}^p + r^{-2} k^2 + r^2(\kappa\wedge)\hodge^{-1}(\kappa\wedge)\hodge & (-1)^p[\hodge^{-1}\Dext_k\hodge,(\kappa\wedge)] \\
	r^2[\Dext_k, \hodge^{-1}(\kappa\wedge)\hodge] & \Delta^{p-1}_{\Dext_k} + r^{-2}k^2 + r^2\hodge^{-1}(\kappa\wedge)\hodge(\kappa\wedge)
\end{psmallmatrix},
\end{equation}
where $[\cdot, \cdot]$ denotes the commutator of operators, and we recall that $\dext \tau = -\pi^* \kappa$.

In particular, scalar functions ($p=0$) have no vertical component. So only the top-left block is nonzero, and it simplifies to
\begin{equation}
	\Delta^0_k = \Delta^0_{\Dext_k} + r^{-2}k^2, \label{minsec:eq.scalar-laplacian}
\end{equation} 
that is, a connection Laplacian plus a vertical Laplacian at each frequency $k$.

\subsection{Global Step} \label{minsec:subsec:global-step-discrete}
To compute the global step of \textsc{admm}, we use the convenient splitting of frequencies described in \Cref{minsec:subsec:bundle-fourier}. This is the only step in which we work in vertical Fourier components.
\subsubsection{Exact part}
We now need to discretize each of the variables in our augmented Lagrangian so that the system of \textsc{pde} \eqref{minsec:eq.f-pde}--\eqref{minsec:eq.bdry2} is simple to solve numerically. \eqref{minsec:eq.f-pde} is the only equation involving fields on the full bundle, so we focus on it first. We discretize $f$ by first decomposing it into its frequency components along the fiber direction and then representing each component in a basis of ``covariant'' linear finite elements with complex coefficients.

On a triangle $T$, let $\psi_a$ denote the linear hat function taking the value $1$ at vertex $a$ and $0$ at the other vertices. Let $f^{(k)}_a$ be the complex value of $f^{(k)}$ at vertex $a$. Let $\rho_{a\to T}$ be the Levi-Civita parallel transport operator from the frame at vertex $a$ to the frame on triangle $T$, represented as a unit complex number (see \ref{minsec:app:op-pt} for details). The corresponding parallel transport operator at frequency $k$ is then given by $\rho_{a\to T}^{-k}$. Note the negative exponent as Fourier components are parallel transported with an opposite orientation to vectors on the base. On triangle $T$ and expressed with respect to the frame on $T$, the $k$th component of $f$ is given by
\begin{equation}
	f^{(k)}(x) = \sum_{a\sim T} \psi_a(x)\rho_{a\to T}^{-k} f_a^{(k)}. \label{minsec:eq.f-basis}
\end{equation}

Having chosen our degrees of freedom, we discretize the variational subproblem \eqref{minsec:prob.fphi} as follows.
For frequencies $k \ne 0$, there is no $\phi$ component, and we only have to solve for $f^{(k)}$. The $L^2$ norm on $E$ decomposes into one term per frequency, of the form
\begin{equation}
	\mathcal{L}(f^{(k)}) = \|\dext_k^0 f^{(k)} - \alpha^{(k)}\|^2 = \|\Dext_k^0 f^{(k)} - \alpha^{(k)}_H\|^2 + \|ik f^{(k)} - \alpha^{(k)}_V\|^2,
\end{equation}
where $\alpha^{(k)}$ is the $k$th frequency component of the $1$-form $\Sigma - \bar\tau + w$. Expanding, we have
\begin{equation}
\begin{aligned}
\mathcal{L}(f^{(k)}) = 2\pi r \sum_T \int_T  \Bigl( &|\Dext_k^0 f^{(k)}(x) - \alpha^{(k)}_H(x)|^2 \\
	&{}+ r^{-2}|ik f^{(k)}(x) - \alpha^{(k)}_V(x)|^2 \Bigr)\, \dext A .
\end{aligned}
\end{equation}
Each triangle is Levi-Civita-flat, so $\Dext_k$ reduces to the ordinary differential once $f$ is expressed in the frame of $T$. Taking $\sharp$s and letting $\begin{psmallmatrix}\mathbf{h} \\ v\end{psmallmatrix} = \alpha^\sharp = \begin{psmallmatrix}\alpha_H^\sharp \\ r^{-2} \alpha_V \end{psmallmatrix}$, we have
\begin{equation}
\begin{aligned}
	\mathcal{L}(f^{(k)}) &= 2\pi r \sum_T \Re \int_T \sum_{ab\sim T}  \Bigl[\\
	 &\quad \bar{\rho}^k_{a\to T}\bar{f}^{(k)}_a\rho^{-k}_{b \to T}f^{(k)}_b (\nabla\psi_a \cdot \nabla \psi_b + r^{-2} k^2 \psi_a \psi_b) \\
	&\quad {}- 2\bar\rho^{-k}_{a \to T}\bar{f}^{(k)}_a (\nabla \psi_a \cdot \mathbf{h}^{(k)}_{T,b} \psi_b - r^{-2} ik\psi_a \psi_b v^{(k)}_{T,b}) \Bigr]\, \dext A,
\end{aligned}
\end{equation}
where $\mathbf{h}_{T,b}$ and $v_{T,b}$ denote the values of the respective variables at the corner of triangle $T$ meeting vertex $b$.

Let $\dscm{P}_k$ be the covariant corner-vertex incidence matrix at frequency $k$. That is, for a triangle $T$ whose $i$th corner is incident on a vertex $a$, $\dscm{P}_k$ has an entry corresponding to the parallel transport operator from the frame at $a$ to the frame on $T$:
\begin{equation}
	(\dscm{P}_k)_{T_i,a} = \rho^{-k}_{a\to T}.
\end{equation}
Let $\dscm{G}$ be the triangle-wise gradient operator whose $2\times 1$ matrix blocks are
\begin{equation}
	(\dscm{G})_{T_j} = \nabla \psi_{T_j},
\end{equation}
where $\psi_{T_j} = \psi_a \mid_T$ is the linear restriction of the finite element associated to vertex $a$ incident on corner $j$ of triangle $T$. Finally, let $\dscm{A}$ be the triangle area matrix, $\dscm{U}$ be the face-corner incidence matrix, and $\dscm{M}$ be the corner-wise FEM mass matrix, whose entries are
\begin{equation}
	(\dscm{M})_{T_i,T_j} = \int_T \psi_{T_i} \psi_{T_j} \, \dext A.
\end{equation}

We can now rewrite the Lagrangian in terms of the vector of discrete degrees of freedom $\dsc{f}_k = (f^{(k)}_a)$,
\begin{dmath}
	\mathcal{L}(\dsc{f}_k) = 2\pi r \Re [\dsc{f}_k^\dag \dscm{P}_k^\dag \dscm{G}^\dag \dscm{A} (\dscm{G} \dscm{P}_k \dsc{f}_k - 2\dscm{U}\dsc{h}_k) + r^{-2} (k^2 \dsc{f}^\dag \dscm{P}_k^\dag \dscm{M} \dscm{P}_k \dsc{f} - 2ik \dsc{f}^\dag \dscm{P}_k^\dag \dscm{M} \dsc{v}_k) ]
\end{dmath}
Taking the first variation with respect to $\dsc{f}_k$, we obtain the following matrix equation, which serves as a discrete version of \eqref{minsec:eq.f-pde}:
\begin{equation}
(\dscm{P}_k^\dag \dscm{G}^\dag \dscm{A} \dscm{G} \dscm{P}_k + r^{-2} k^2 \dscm{P}_k^\dag \dscm{M} \dscm{P}_k) \dsc{f}_k = \dscm{P}_k^\dag \dscm{G}^\dag \dscm{A} \dscm{U} \dsc{h}_k + r^{-2} ik \dscm{P}_k^\dag \dscm{M} \dsc{v}_k. \label{minsec:eq-fpde-discrete}
\end{equation}
The system matrix on the left is a discrete version of the Laplacian block at frequency $k$ (\emph{cf.} \eqref{minsec:eq.scalar-laplacian}),
\begin{equation}
		\dscm{L}_k \coloneqq \dscm{P}_k^\dag \dscm{G}^\dag \dscm{A} \dscm{G} \dscm{P}_k + r^{-2} k^2 \dscm{P}_k^\dag \dscm{M} \dscm{P}_k. \label{minsec:eq:laplacian-freq-k-discr}
\end{equation}

\subsubsection{Coexact part; frequency zero}
The coexact part $\pi^*\hodge\dext \phi$ is horizontal homogeneous, making it orthogonal to $\dext f^{(k)}$ for $|k| > 0$ by orthogonality of the Fourier decomposition. In contrast, $f^{(0)}$ and $\phi$ are coupled via the boundary condition \eqref{minsec:eq.bdry1} and must be computed jointly via a Schur complement procedure.

The Lagrangian for the components at frequency zero is of the following form:
\begin{equation}
\begin{aligned}
	\mathcal{L}(f^{(0)}, \phi, \beta) &= \frac{\mu \ell}{2} \|\hodge \dext \phi + \dext f^{(0)} - \alpha_H^{(0)}\|^2 + \frac{\nu}{2}\|\Delta \phi - g\|^2 \\
	&\qquad{} + \int_{\partial B} \beta (\dext f^{(0)} + \hodge \dext \phi - g_0) \\
	&= \frac{\mu\ell}{2} \| J \nabla \phi + \nabla f^{(0)} - \mathbf{h}^{(0)}||^2 + \frac{\nu}{2}\|\Delta \phi - g\|^2 \\
	&\qquad{} + \int_{\partial B} \beta (\dext f^{(0)} + \hodge \dext \phi - g_0).
\end{aligned}
	\label{minsec:eq:freq0-lagrangian-continuous}
\end{equation}
Here $g = \Gamma - \bar\kappa + z$, $g_0 = \Gamma_0 - \bar\tau$, and $\beta \in \Omega^0(\partial B)$ is a Lagrange multiplier.
We have used that $(\hodge \dext \phi)^\sharp = J \nabla \phi$, where $J$ is rotation counterclockwise by $\pi/2$ in the tangent space of $B$. Note the factor of $\ell$ in the first term, arising from partial integration over the fibers.

We discretize $\phi$ in a basis of non-conforming or \emph{Crouzeix-Raviart} (C-R) elements. This is important for two reasons. First, as detailed in \cite{de_goes_vector_2016}, using C-R elements ensures that the discrete Hodge decomposition has the right number of degrees of freedom; if both the exact and coexact terms were represented with conforming elements, the Hodge decomposition would fail to span the space of piecewise-constant vector fields. Second, the discrete curl operator expressed in the non-conforming basis satisfies a discrete Stokes' theorem: the boundary integral of a vector field corresponds to the total discrete curl in the interior. This is important to ensure the sum of singularity indices agrees exactly with the total curvature of the boundary. The conforming finite elements lack an analogous discrete Stokes property.

The boundary condition \eqref{minsec:eq.bdry2} is enforced by explicitly setting the value of $\phi$ to zero at boundary edge midpoints.
Each interior edge $e$ gets a Crouzeix-Raviart basis element $\phi_e$, taking the value $1$ on $e$ and $-1$ on the pair of vertices opposite $e$. Let $\dscm{\hat{G}}$ be the C-R gradient matrix, and let $\dscm{J}$ be the matrix representing rotation counterclockwise by $\pi/2$ in the tangent space of each triangle. Let $\dscm{\hat{M}}$ be the C-R mass matrix. Let $\dscm{\hat{L}} = \dscm{\hat{G}^\dag A \hat{G}} = \dscm{\hat{G}}^\dag \dscm{J}^\dag \dscm{A} \dscm{J} \dscm{\hat{G}}$ be the symmetric C-R Laplacian. Then \eqref{minsec:eq:freq0-lagrangian-continuous} becomes
\begin{equation}
\begin{aligned}
	\mathcal{L}(\dsc{f}_0, \phi, \beta)
	&= \frac{\mu\ell}{2} (\dscm{J}\dscm{\hat{G}} \phi + \dscm{G}\dscm{P}_0\dsc{f}_0 - \dsc{h}_0)^\dag \dscm{A} (\dscm{J}\dscm{\hat{G}} \phi + \dscm{G}\dscm{P}_0 \dsc{f}_0 - \dsc{h}_0) \\
	&\qquad{} + \frac{\nu}{2} (\dscm{\hat{M}}^{-1} \dscm{\hat{L}} \phi - \dsc{g})^\dag \dscm{\hat{M}} (\dscm{\hat{M}}^{-1} \dscm{\hat{L}} \phi - \dsc{g}) \\
	&\qquad{} + \beta^\dag [\dscm{B} (\dscm{G} \dscm{P}_0 \dsc{f}_0 + \dscm{J}\dscm{\hat{G}} \phi) - \dsc{g}_0] \\
	&= \frac{\mu\ell}{2} \dsc{f}_0^\dag \dscm{L}_0 \dsc{f}_0  + \frac{1}{2}\phi^\dag (\mu\ell \dscm{\hat{L}} + \nu \dscm{\hat{L}\hat{M}}^{-1}\dscm{\hat{L}}) \phi \\
	&\qquad{} - \mu\ell \dsc{h}_0^\dag \dscm{A} \dscm{J} \dscm{\hat{G}} \phi - \mu\ell \dsc{h}_0^\dag \dscm{A G} \dscm{P}_0 \dsc{f}_0 - \nu \dsc{g}^\dag \dscm{\hat{L}} \phi  \\
	&\qquad{} + \beta^\dag [\dscm{B} (\dscm{G} \dscm{P}_0 \dsc{f}_0 + \dscm{J}\dscm{\hat{G}} \phi) - \dsc{g}_0].
\end{aligned}
\end{equation}
The resulting \textsc{kkt} system is
\begin{equation}
	\begin{psmallmatrix}
		\mu\ell \dscm{L}_0 & 0 & \dscm{P}_0^\dag \dscm{G}^\dag \dscm{B}^\dag \\
		0 & \mu\ell \dscm{\hat{L}} + \nu \dscm{\hat{L}} \dscm{\hat{M}}^{-1} \dscm{\hat{L}} & \dscm{\hat{G}}^\dag \dscm{J}^\dag \dscm{B}^\dag \\
		\dscm{B}\dscm{G}\dscm{P}_0 & \dscm{B}\dscm{J}\dscm{\hat{G}} & 0
	\end{psmallmatrix}
	\begin{psmallmatrix}
		\dsc{f}_0 \\ \phi \\ \beta
	\end{psmallmatrix}
	=
	\begin{psmallmatrix}
		\mu\ell \dscm{P}_0^\dag\dscm{G}^\dag \dscm{A} \dsc{h}_0 \\
		\mu\ell \dscm{\hat{G}}^\dag \dscm{J}^\dag \dscm{A} \dsc{h}_0 + \nu \dscm{\hat{L}} \dsc{g} \\
		\dsc{g}_0
	\end{psmallmatrix}.
\label{minsec:eq:pde-freq0-discrete}
\end{equation}
We solve this system by first solving the Schur complement system for $\beta$, and then solving for $\dsc{f}_0$ and $\phi$ independently. All matrices involved are positive semidefinite and can be Cholesky factored in advance. The Schur complement matrix is typically dense, but small enough (of order the number of vertices in $\partial B$) that it can be factored efficiently.

\subsubsection{Boundary Conditions and Constants}

\paragraph{Curvature.} As $\bar\kappa$ is coupled to $\phi$ and $\Gamma$ through a Poisson equation, we discretize all three in the Crouzeix-Raviart basis. To transport the Gauss curvature $\kappa$ to edge centers, the integrated Gauss curvature at each vertex is first computed as the angle defect. Then it is converted to pointwise curvature by dividing by vertex area, and the pointwise curvature on each edge is defined to be the mean of the curvatures at its endpoints. This is similar to the approach in \cite{ben-chenDiscreteKillingVector2010}; note, however, that we take an unweighted mean to avoid issues with negative cotangent weights on non-Delaunay meshes:
\begin{equation}
	\kappa_v = \frac{1}{A_v}\sum_{T \sim v} (2\pi - \theta_{T_v}) \qquad \kappa_{vw} = \frac{1}{2}(\kappa_v + \kappa_w).
\end{equation}

\paragraph{Vertical form.} The vertical form $\bar\tau$ only appears in the expressions $\alpha = \Sigma - \bar\tau + w$ and $g_0 = \Gamma_0 - \bar\tau$---that is, on the right hand sides of the linear systems \eqref{minsec:eq-fpde-discrete} and \eqref{minsec:eq:pde-freq0-discrete}. Thus, we represent $\bar\tau$ in the spatial domain, and transform the combined expressions $\alpha$ and $g_0$ into vertical Fourier components at the beginning of the global step. $\bar\tau$ is particularly simple, having a constant value of $(\bar\tau_H = 0, \bar\tau_V = 1/(2\pi))$ at each point in the bundle.

\paragraph{$f^{(k)}$ Boundary Conditions.} For frequencies $|k| > 0$, the boundary condition \eqref{minsec:eq.bdry1} becomes
\begin{equation}
\begin{pmatrix}
\Dext^p_k \\ ik
\end{pmatrix}
f^{(k)} = \dext f^{(k)} = (\Gamma_0 - \bar\tau)^{(k)} = (\Gamma_0)^{(k)},
\end{equation}
which can be solved explicitly for the value of $f^{(k)}$ in terms of the Fourier components of $\Gamma_0$:
\begin{equation}
	f^{(k)} = \frac{(\Gamma_0)^{(k)}_V}{ik}. \label{minsec:eq.bdry1-explicit}
\end{equation}
$\Gamma_0$ is a curve ($1$-current) sitting over the boundary $\partial B$. Its corresponding $1$-form, when restricted to a single fiber, should have a vertical component whose integral over a segment of the fiber is exactly the number of times that segment intersects $\Gamma_0$ viewed as a curve. That is, $(\Gamma_0)_V \mid_{\pi^{-1}(x)}$ is a Dirac delta measure $\delta_{\gamma_0(x)}$, where $\gamma_0(x)$ represents the boundary value at $x$. The Fourier series of a Dirac delta is of the form
\begin{equation}
	\delta_{\gamma_0}(\theta) = \sum_k e^{ik(\theta - \gamma_0)}.
\end{equation}
Na\"ively truncating this series yields an approximation with ringing artifacts that can dip into the negative. Instead, inspired by \cite{bandeira_non-unique_2020}, we adopt the \emph{Fej\'er kernel} approximant,
\begin{equation}
	\delta_{\gamma_0}(\theta) \approx \sum_{|k| \le K} \left(1 - \frac{|k|}{K}\right) e^{ik(\theta - \gamma_0)}, \label{minsec:eq:fejer}
\end{equation}
which is nonnegative. This allows us to impose the convex constraint $\Sigma_V \ge 0$ in the local step of \textsc{admm}.

To impose the boundary condition \eqref{minsec:eq.bdry1-explicit} discretely, the entries $(\dsc{f}_k)_{\dsc{B}}$ at boundary vertices are fixed, corresponding rows of $\dscm{L}_k$ are removed, and $(\dscm{L}_k)_{\dsc{IB}} (\dsc{f}_k)_{\dsc{B}}$ is subtracted from the right-hand side of \eqref{minsec:eq-fpde-discrete} before solving for the interior components $(\dsc{f}_k)_{\dsc{I}}$.

\subsection{Local subproblem}
In the local step, we solve for $\Sigma$ and $\Gamma$ given the values of $f$ and $\phi$ computed in the previous step. The objective function in the subproblem \eqref{minsec:eq.local-step} decomposes into pointwise components. As such, we conduct the local step of \textsc{admm} in the spatial domain; $f$ and $\phi$ are transformed back into the spatial domain prior to this step.
To match $\phi$, we discretize $\Gamma$ with values on edge centers. For $\Sigma$, we assign one value per triangle corner per increment along the fiber direction. The choice of corners rather than vertices is designed to minimize numerical averaging artifacts and thus encourage convergence of the current to a sharply defined surface.

As such, we solve for the value of $\Sigma$ at each combination of triangle corner and vertical increment separately. This amounts to evaluating a proximal operator for the unsquared $L_2$ norm of a vector, also known as a \emph{shrinkage} operator.
At a corner $j$, we solve
\begin{equation}
	\begin{alignedat}{2}&\argmin_{\Sigma_j}& &\frac{\mu}{2}|\Sigma_j - \hat\Sigma_j|_g^2 + |\Sigma_j|_g \\
	&\subj \quad &&(\Sigma_j)_V \ge 0,
	\end{alignedat}
\end{equation}
where $\hat\Sigma_j$ is the value of $\bar{\tau} + \dext f + \pi^* \hodge\dext \phi - w$ interpolated to corner $j$. The solution is given by the explicit formula
\begin{equation}
	\Sigma_j = \left(1 - \frac{1}{\mu|\bar{\Sigma}_j|_g}\right)^+ \bar{\Sigma}_j, \qquad \bar{\Sigma}_j = \begin{pmatrix}
		(\hat\Sigma_j)_H \\
		((\hat\Sigma_j)_V)^+
	\end{pmatrix},
\end{equation}
where $y^+ \coloneqq \max\{0, y\}$. Similarly, at each interior edge center $e$, we compute
\begin{equation}
	\Gamma_e = \left(1 - \frac{\lambda}{\nu |\hat\Gamma_e|}\right)^+ \hat\Gamma_e,
\end{equation}
where $\hat\Gamma = \dscm{\hat{M}}^{-1} \dscm{\hat{L}} \phi + \bar{\kappa} - \dsc{z}$.

\subsection{Convergence Criterion}
Because our augmented Lagrangian \eqref{minsec:eq:lagrangian} has two least-squares augmentations associated to two penalty parameters $\mu$ and $\nu$, we compute two pairs of primal and dual residuals to determine convergence. At iteration $t$, these have the form
\begin{equation}
	\begin{alignedat}{4}
		&R_{P}^\mu &&\coloneqq \| \Sigma^{(t)} - \hat\Sigma^{(t)} \|_g \qquad &&R_{P}^\nu &&\coloneqq \|\Gamma^{(t)} - \hat\Gamma^{(t)}\| \\
		&R_{D}^\mu &&\coloneqq \| \Sigma^{(t-1)} - \Sigma^{(t)} \|_g \qquad &&R_{D}^\nu &&\coloneqq \|\Gamma^{(t-1)} - \Gamma^{(t)}\|.
	\end{alignedat}
\end{equation}
We deem convergence to have occurred when all four residuals fall below a threshold $\epsilon$, equal to $\num{5e-4}$ in our experiments. We also use the residuals to adaptively update the penalty parameters $\mu$ and $\nu$, using a standard adaptive scheme \cite[eq.\ (3.13)]{boyd2011distributed}. Our complete algorithm is shown in \Cref{minsec:alg:admm}.

\begin{algorithm}
	\caption{\textsc{admm}-based algorithm for computing minimal sections.} \label{minsec:alg:admm}
	\DontPrintSemicolon
	\SetKwInOut{Input}{input}
	\SetKwInOut{Result}{result}
	\Input{boundary values $\Gamma_0$, curvature $\kappa$}
	\Result{$\Sigma$, $\Gamma$, $\dsc{f}$}
	\BlankLine
	\Repeat{$\min\left\{R^\mu_{P}, R^\mu_{D}, R^\nu_{P}, R^\nu_{D}\right\} < \epsilon$}{
		\Begin(\tcc*[h]{global step}){
			solve \eqref{minsec:eq-fpde-discrete} for $\dsc{f}_k$\;
			solve \eqref{minsec:eq:pde-freq0-discrete} for $\dsc{f}_0$ and $\phi$\;
			$(\hat\Sigma_k)_H \gets \dscm{G} \dscm{P}_k \dsc{f}_k$\;
			$(\hat\Sigma_k)_V \gets ik \dscm{P}_k \dsc{f}_k$\;
			$(\hat\Sigma_0)_H \gets \bar{\tau} + \dscm{G}\dscm{P}_0 \dsc{f}_0 + \dscm{J} \dscm{\hat{G}} \phi$\;
			$\hat\Gamma \gets \dscm{\hat{M}}^{-1} \dscm{\hat{L}} \phi + \bar{\kappa} - \dsc{z}$\;
		}
		\Begin(\tcc*[h]{local step}){
			\For{triangle corner $j$}{
				$\bar{\Sigma}_j \gets \begin{psmallmatrix}
					(\hat\Sigma_j)_H \\
					((\hat\Sigma_j)_V)^+
				\end{psmallmatrix}, \qquad \Sigma_j \gets \left(1 - \frac{1}{\mu|\bar{\Sigma}_j|_g}\right)^+ \bar\Sigma_j$\;
			}
			\For{interior edge $e$}{
				$\Gamma_e \gets \left(1 - \frac{\lambda}{\nu |\hat\Gamma_e|}\right)^+ \hat\Gamma_e$\;
			}
		}
		\Begin(\tcc*[h]{dual update}){
			$w \gets w + \Sigma - \bar\tau - \dext f - \pi^* \hodge \dext \phi$\;
			$z \gets z + \Gamma + \bar\kappa - \Delta \phi$\;
		}
	}
\end{algorithm}

\subsection{Extracting the Field} \label{minsec:subsec:extraction}
If $\Sigma$ were a perfectly concentrated current, $\Sigma_V$ would be a delta measure on each fiber, and $\dsc{f}$ would be a sawtooth wave over each fiber. Thus, using the Fej\'er approximant \eqref{minsec:eq:fejer}, we would have
\begin{equation}
	\dsc{f}_{k,v} = \frac{1}{ik}\left( 1 - \frac{|k|}{K} \right) e^{-ik\sigma_v},
\end{equation}
where $\sigma_v$ is the angle of the field at vertex $v$. Thus, we can extract the field values from the Fourier coefficients of $f$:
\begin{equation}
	\dsc{z}_v = e^{i\sigma_v} = -i \frac{\dsc{f}_{-1,v}}{|\dsc{f}_{-1,v}|}.
\end{equation}

\section{Results} \label{minsec:sec:results}
We visualize the sections we compute in several different ways. For a flat domain $B$, the bundle $E$ is simply a Cartesian product $E \times \Sph^1$. In this setting, we can embed the whole bundle and directly plot computed currents as periodic surfaces in 3D. To render a current $\Sigma$, we plot the pointwise norm $|\Sigma|_g$ and use the same norm for alpha blending (\Cref{minsec:fig:validation}, left).

Alternatively, we can plot $|\Sigma|_g$ in polar coordinates in the tangent space of each triangle of the base. Such a plot resembles a plot of the extracted directional field, but it also visualizes the degree to which the mass of the current is concentrated or spread out over each fiber of $E$. To make this concentration visible, we normalize the maximum value on each fiber $\max_{\pi^{-1}(x)}|\Sigma_x|_g$ to $1$ for this visualization (\Cref{minsec:fig:validation}, middle).

Finally, we can also extract the directional field $z$ and plot either sampled values (\Cref{minsec:fig:validation}, right) or integral curves (\Cref{minsec:fig:degree}).

\paragraph{Validating \textsc{admm}.} In \Cref{minsec:fig:validation}, we validate our \textsc{admm}-based algorithm against a standard second-order interior-point solver, \textsc{mosek} \cite{mosek_opt}. The problem \eqref{minsec:prob.primal} is discretized with our operator decomposition and passed to \textsc{mosek} via the \textsc{cvx} modeling framework \cite{cvx,gb08}. The solutions are visually identical.
\begin{figure}
\newcommand{\imgheight}{0.30\columnwidth}
\centering
\includegraphics[height=\imgheight]{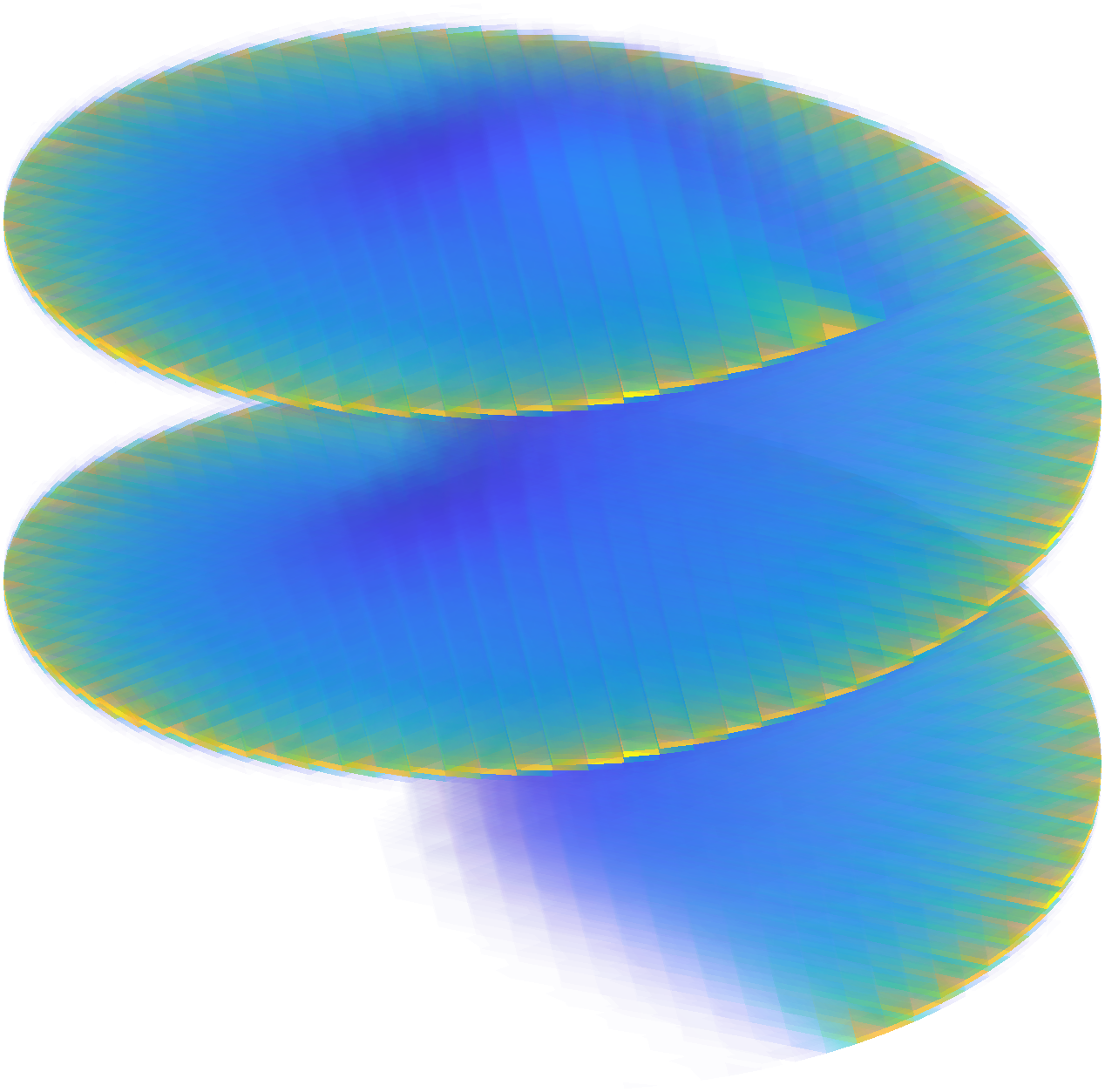}\hfill%
\includegraphics[height=\imgheight]{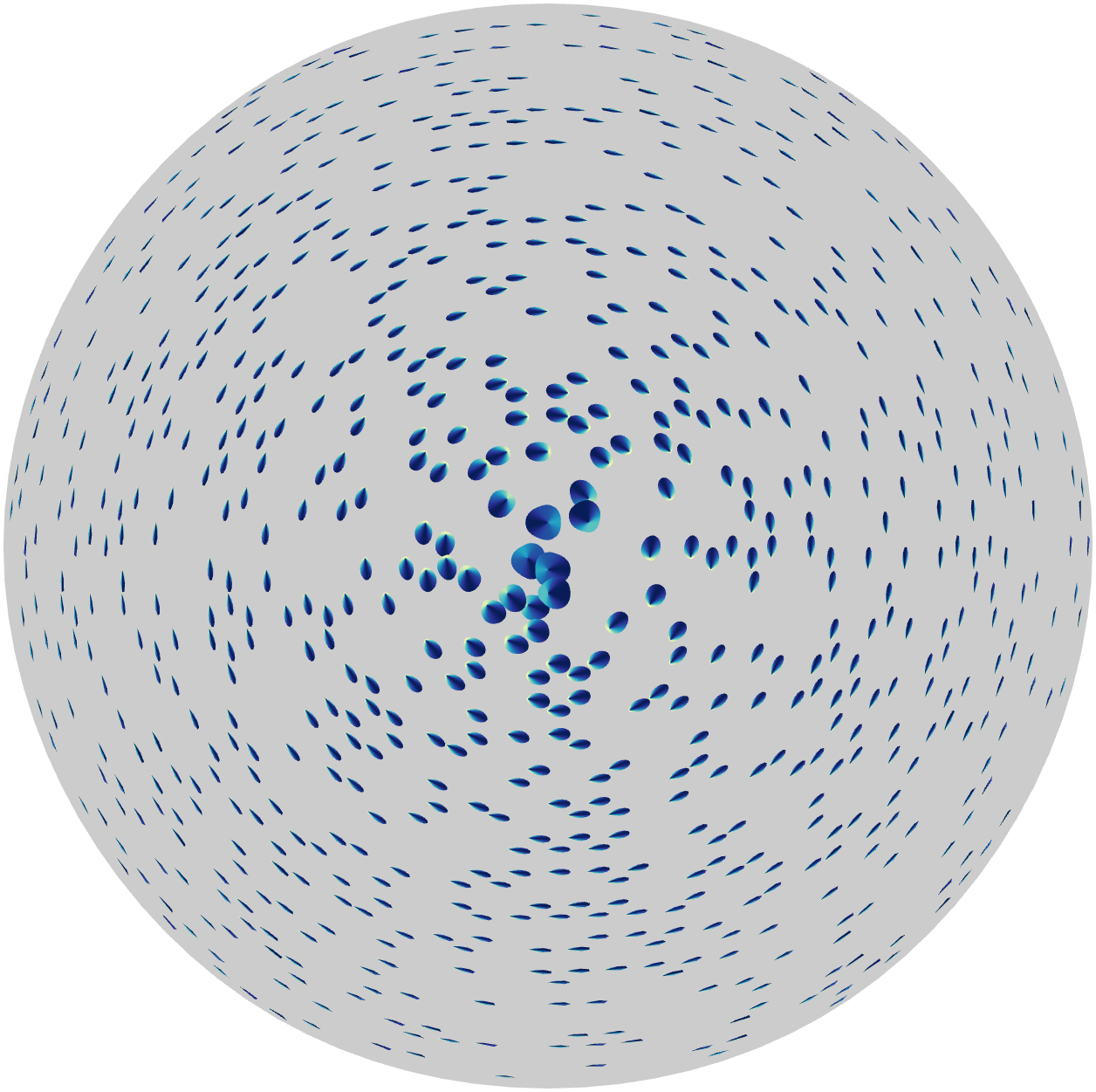}\hfill%
\includegraphics[height=\imgheight]{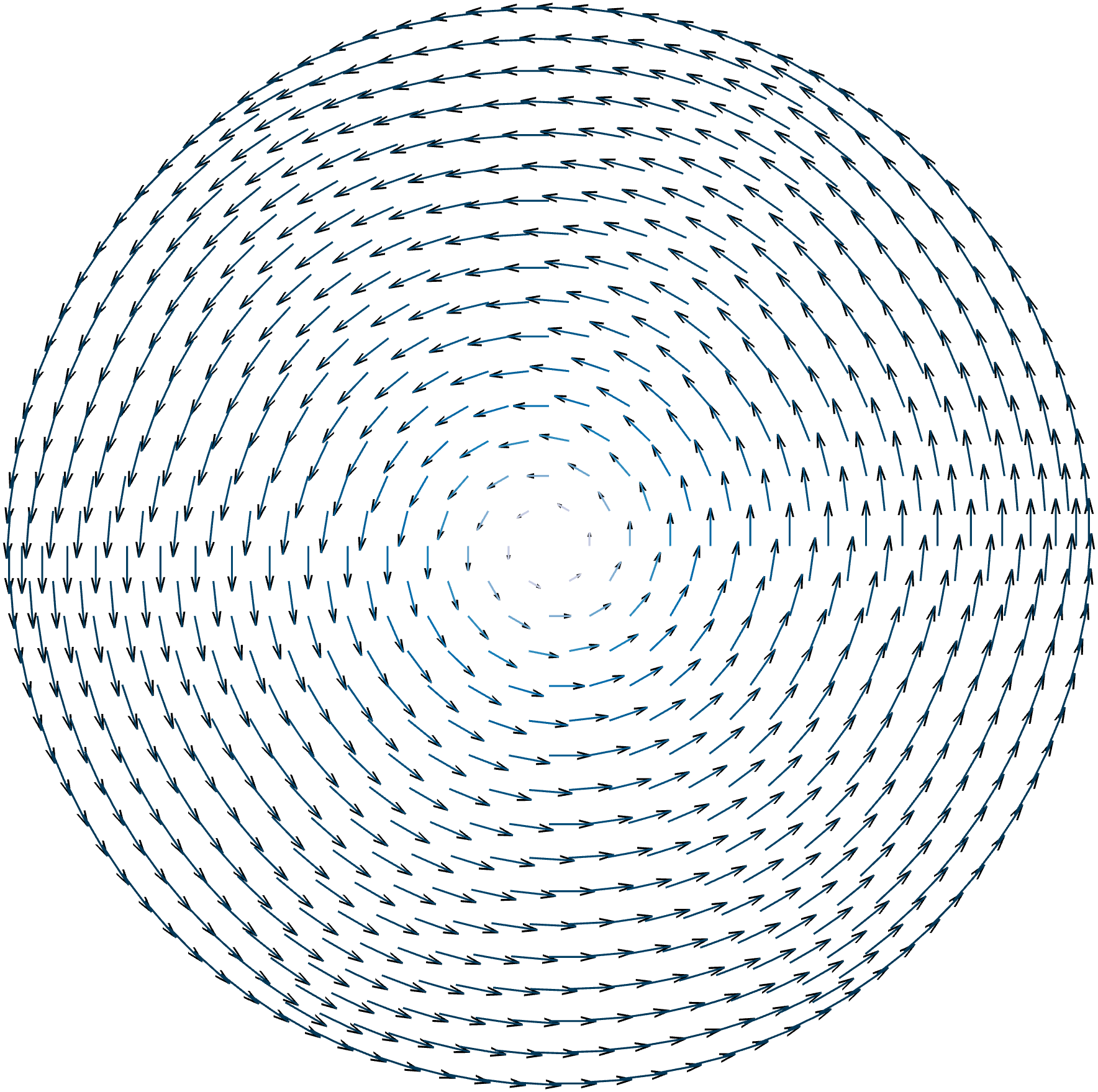}\\
\includegraphics[height=\imgheight]{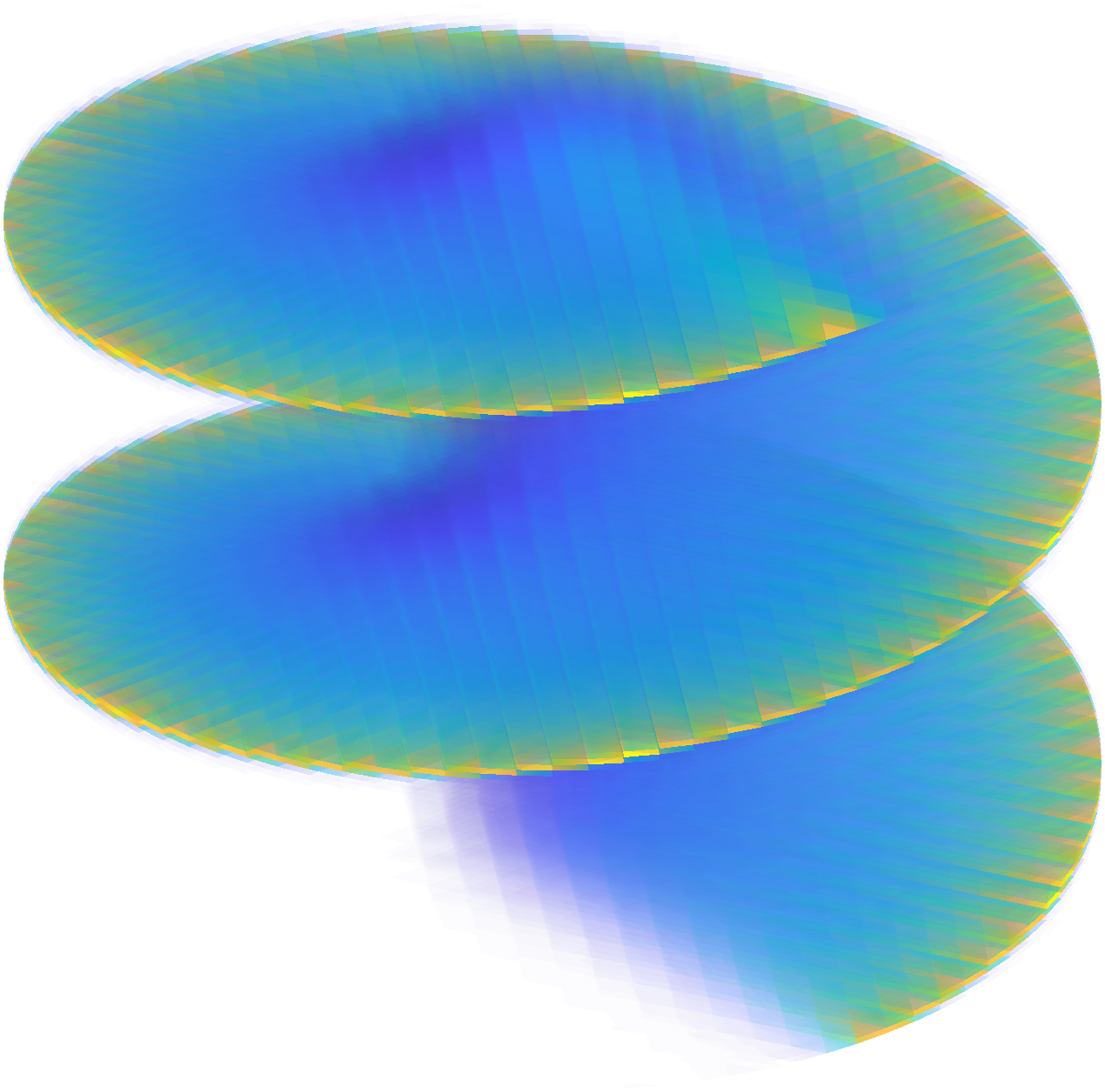}\hfill%
\includegraphics[height=\imgheight]{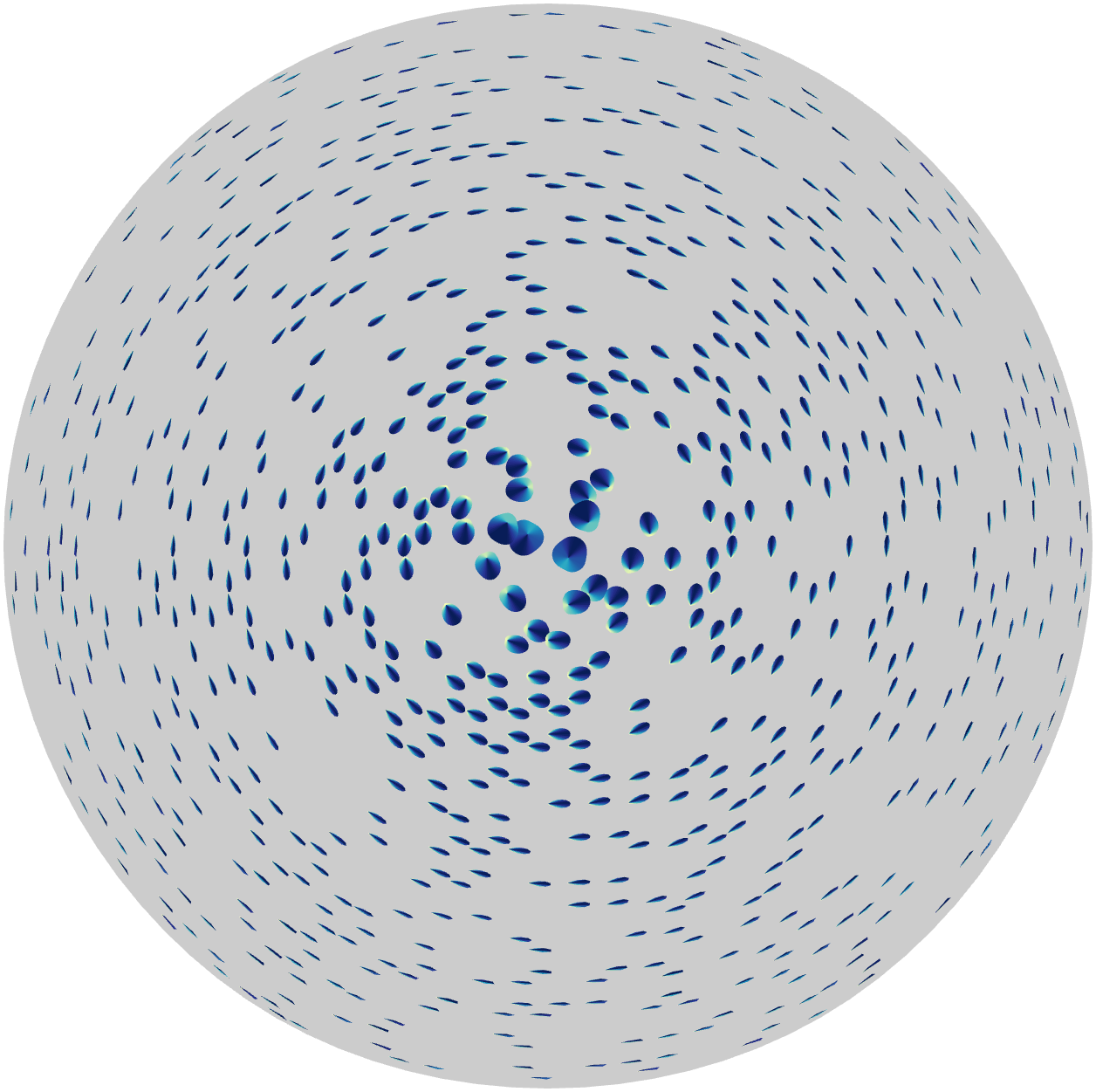}\hfill%
\includegraphics[height=\imgheight]{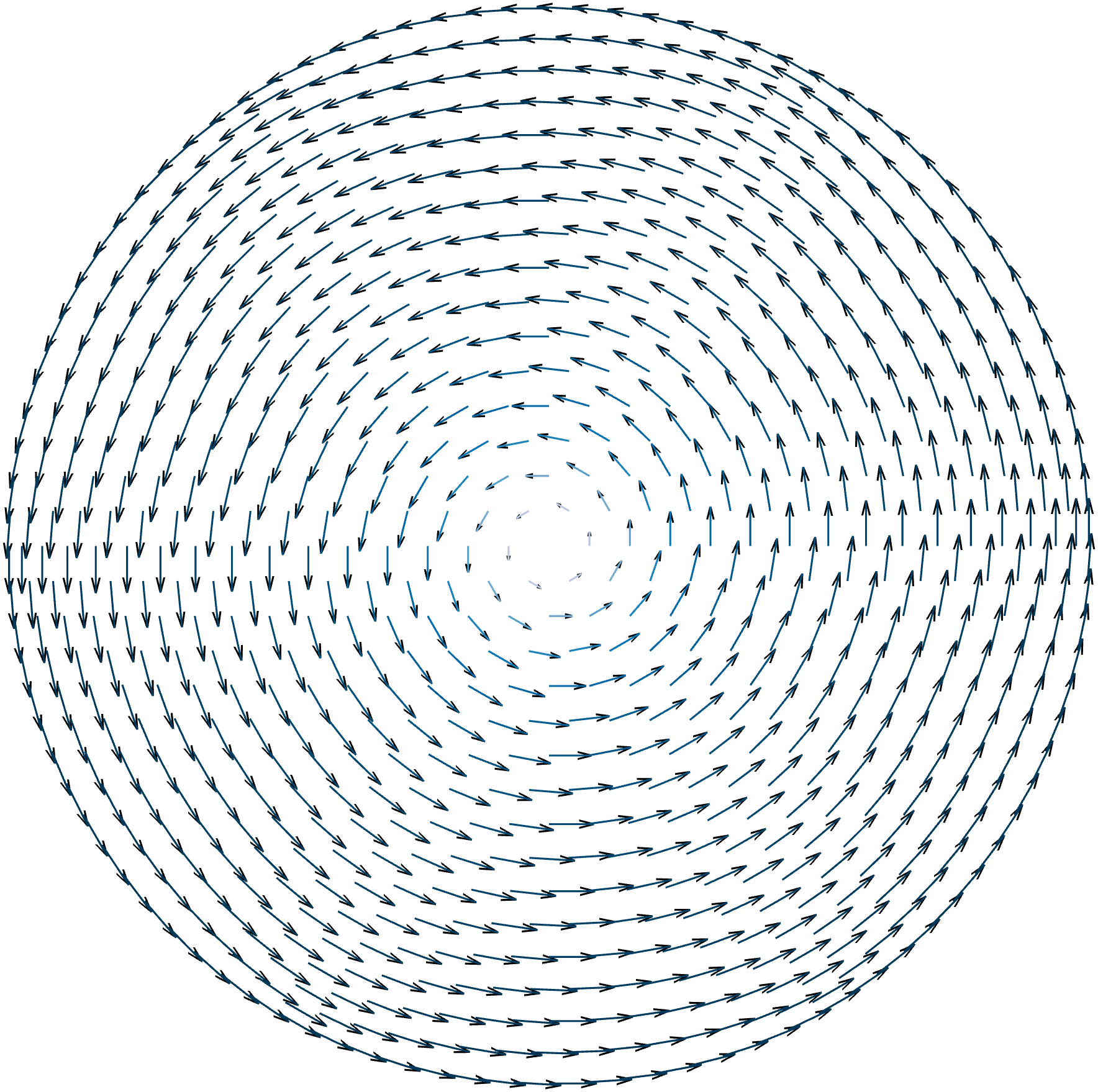}\\
\caption{A minimal section computed by \textsc{admm} (bottom) agrees with that computed by \textsc{mosek} (top) on a small example.}
\label{minsec:fig:validation}
\end{figure}

In \Cref{minsec:fig:convergence}, we examine the convergence of \Cref{minsec:alg:admm} on a simple domain. The logarithmic plot displays the hallmark linear convergence of \textsc{admm}. Meanwhile, the current converges quickly to a near-optimal solution encoding the correct topology, as shown in the volume plots.
\begin{figure}
\centering
\pgfplotstableread[col sep=comma]{figures/convergence/disk/admm_data.csv}{\admmData}
\tikzsetnextfilename{convergence-plot}
\begin{tikzpicture}
\begin{loglogaxis}[
        mark size = 0.5pt,
        width = \columnwidth,
        height = 0.6\columnwidth,
        enlarge x limits = 0,
        grid = none,
        xlabel = {Iteration $j$},
        ylabel = {Residual},
        xlabel near ticks,
        ylabel near ticks,
        legend pos = north east,
        legend cell align = left,
        legend style = {font=\footnotesize, row sep=0.1pt},
        every tick label/.append style = {font=\tiny}]
    \addplot+ [mark = none, thick] table [x = iter, y = pRes1] {\admmData};
    \addlegendentry{primal residual $R_P^\mu$};
    
    \addplot+ [mark = none, thick] table [x = iter, y = dRes1] {\admmData};
    \addlegendentry{dual residual $R_D^\mu$};
    
    \addplot [black, dashed] table [x = iter, y expr={10/\thisrowno{0}}] {\admmData}
    	node [above right, pos=0.3] {$\propto j^{-1}$};
\end{loglogaxis}
\end{tikzpicture}\\%
\newcommand{\imgwidth}{0.1\columnwidth}%
\begin{tabulary}{0.95\columnwidth}{@{}CCCCCCCC@{}}
	\includegraphics[width=\imgwidth]{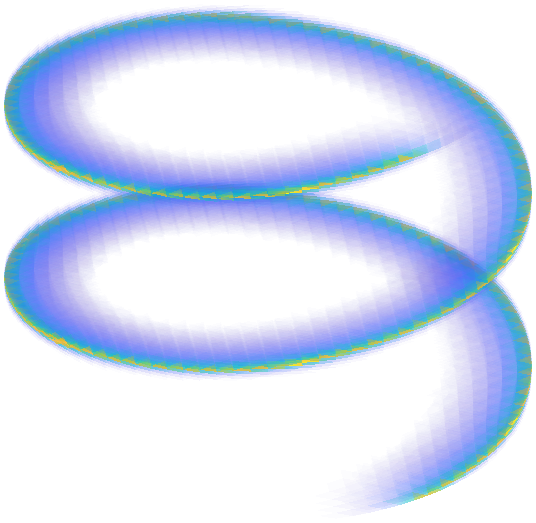} &
	\includegraphics[width=\imgwidth]{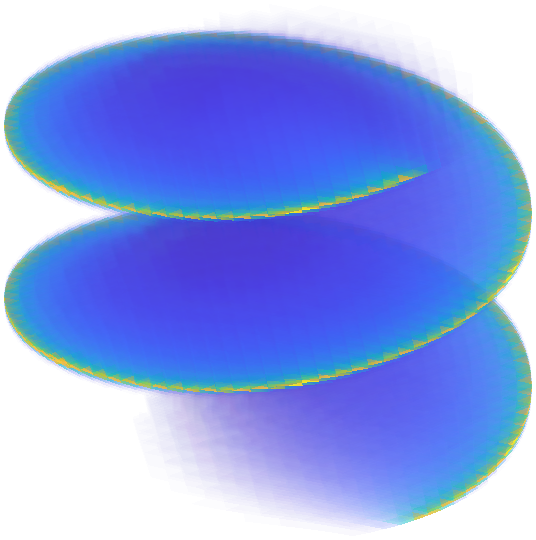} &
	\includegraphics[width=\imgwidth]{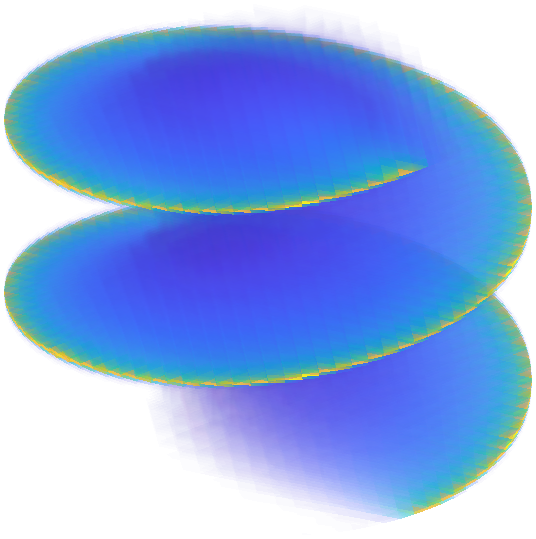} &
	\includegraphics[width=\imgwidth]{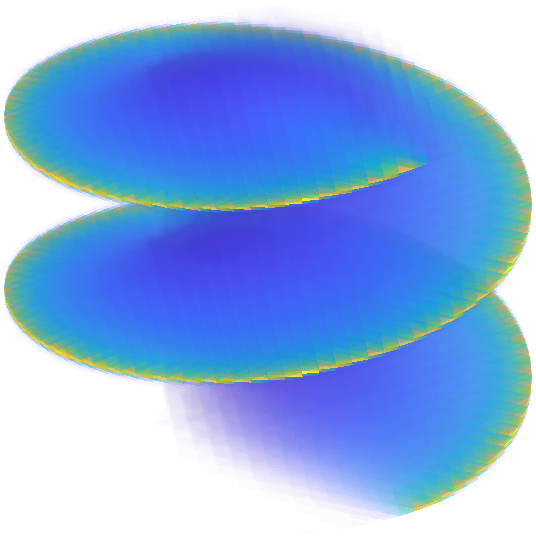} &
	\includegraphics[width=\imgwidth]{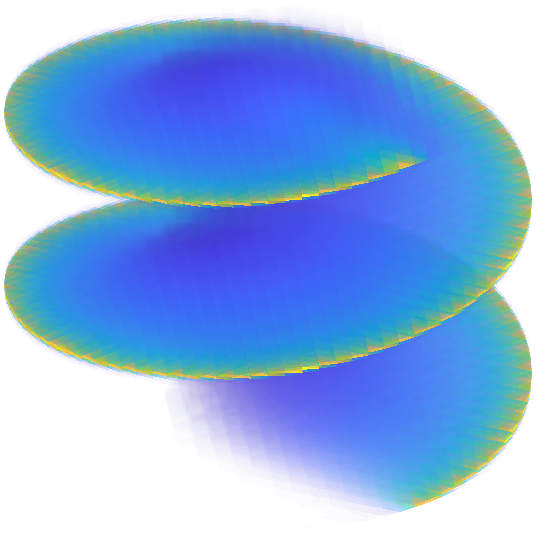} &
	\includegraphics[width=\imgwidth]{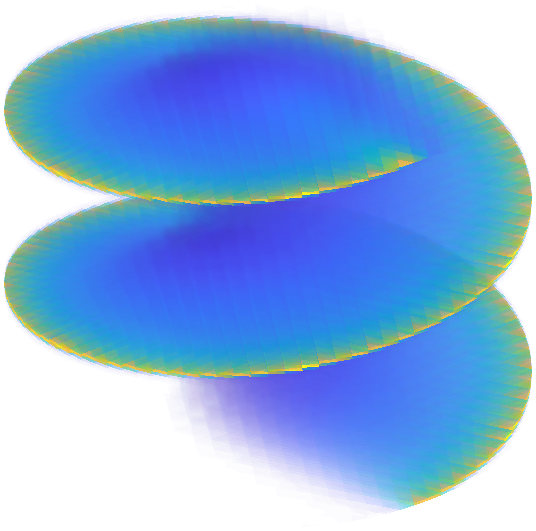} &
	\includegraphics[width=\imgwidth]{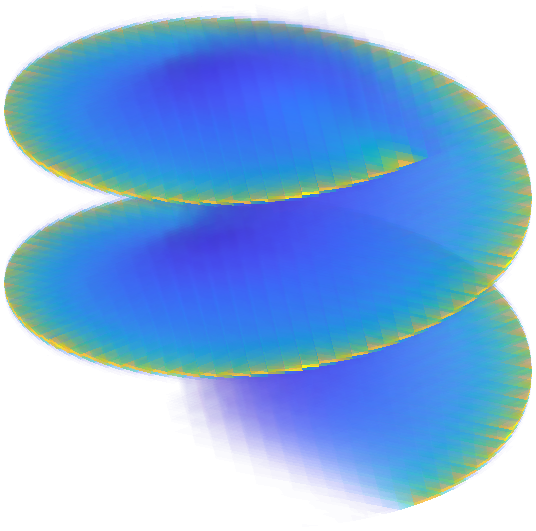} &
	\includegraphics[width=\imgwidth]{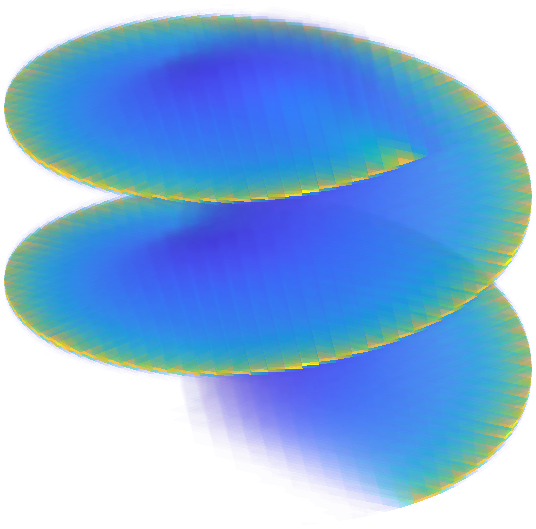} \\
	{$j=1$} & {$11$} & {$21$} & {$31$} & {$41$} & {$51$} & {$61$} & {$71$}
\end{tabulary}
\caption{\textsc{admm} quickly reaches a solution, and residuals converge better than linearly.}
\label{minsec:fig:convergence}    
\end{figure}
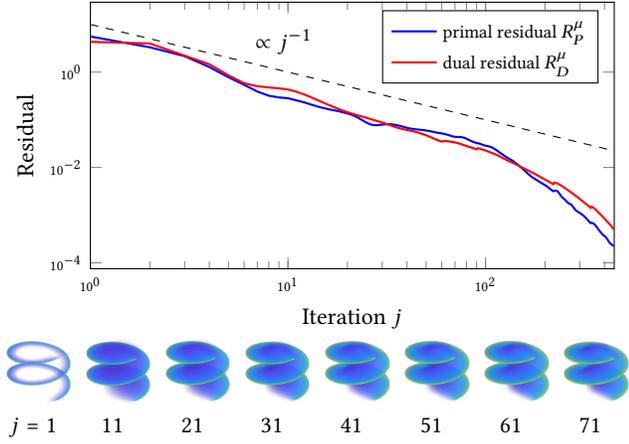

\paragraph{Degree.} In addition to unit vector fields, we can also compute directional fields of higher degree: line fields, cross fields, etc. This amounts to working with a circle bundle of higher degree, i.e., a tensor power of the unit tangent bundle. To construct a bundle of degree $d$ requires a few simple changes:
\begin{enumerate*}[label=(\roman*)]
	\item all parallel transport operators $\rho$ must be replaced by their $d$th powers;
	\item the curvature $\kappa$ is multiplied by $d$; and
	\item the boundary angle values $\gamma_0$ get multiplied by $d$ (modulo $2\pi$).
\end{enumerate*} \Cref{minsec:fig:degree} visualizes minimal sections of increasing degree on the same base surface.
\begin{figure}
\newcommand{\imgwidth}{0.2\columnwidth}
\centering
\begin{tabulary}{0.95\columnwidth}{@{}CCCC@{}}
\includegraphics[width=\imgwidth]{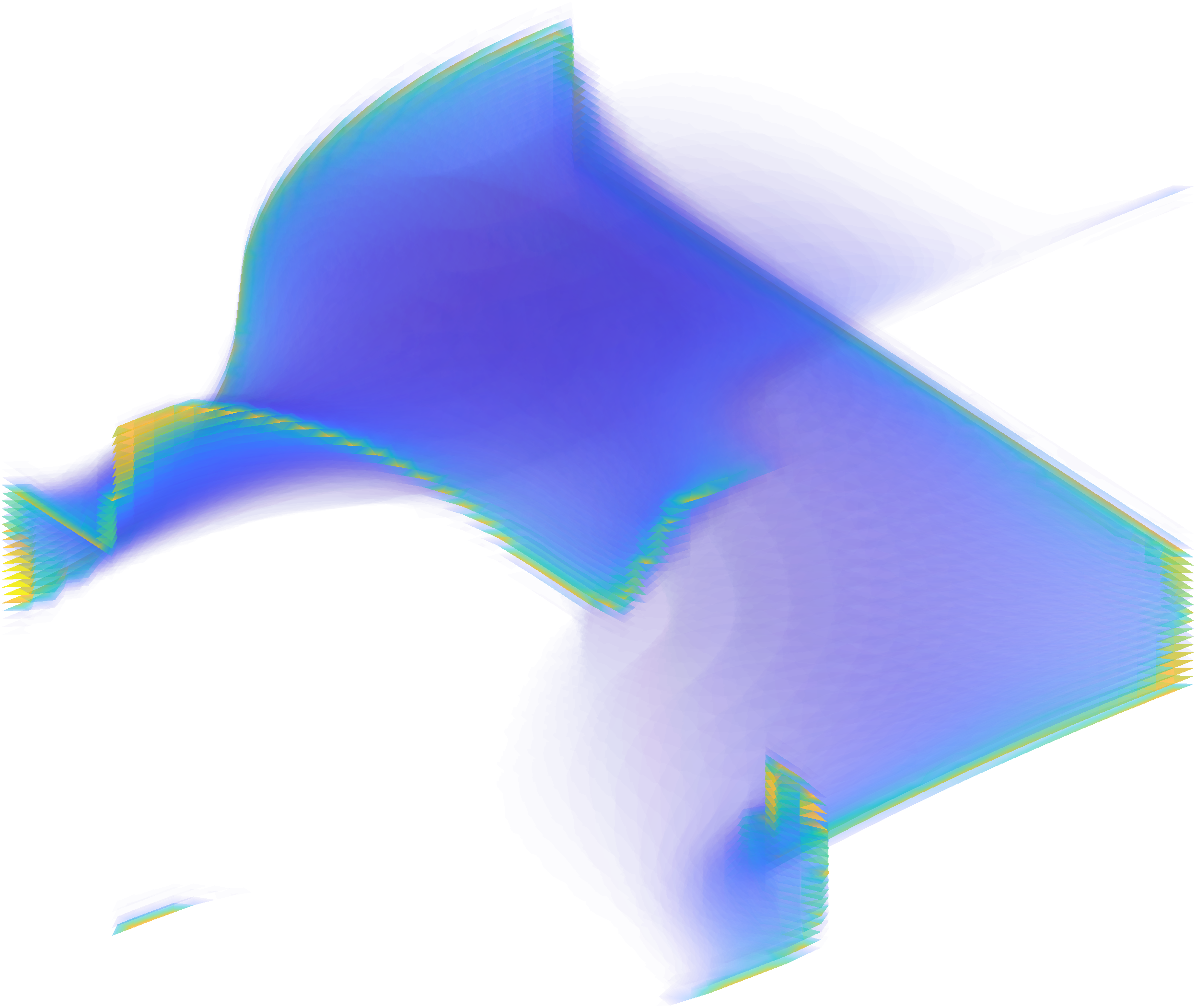} &%
\includegraphics[width=\imgwidth]{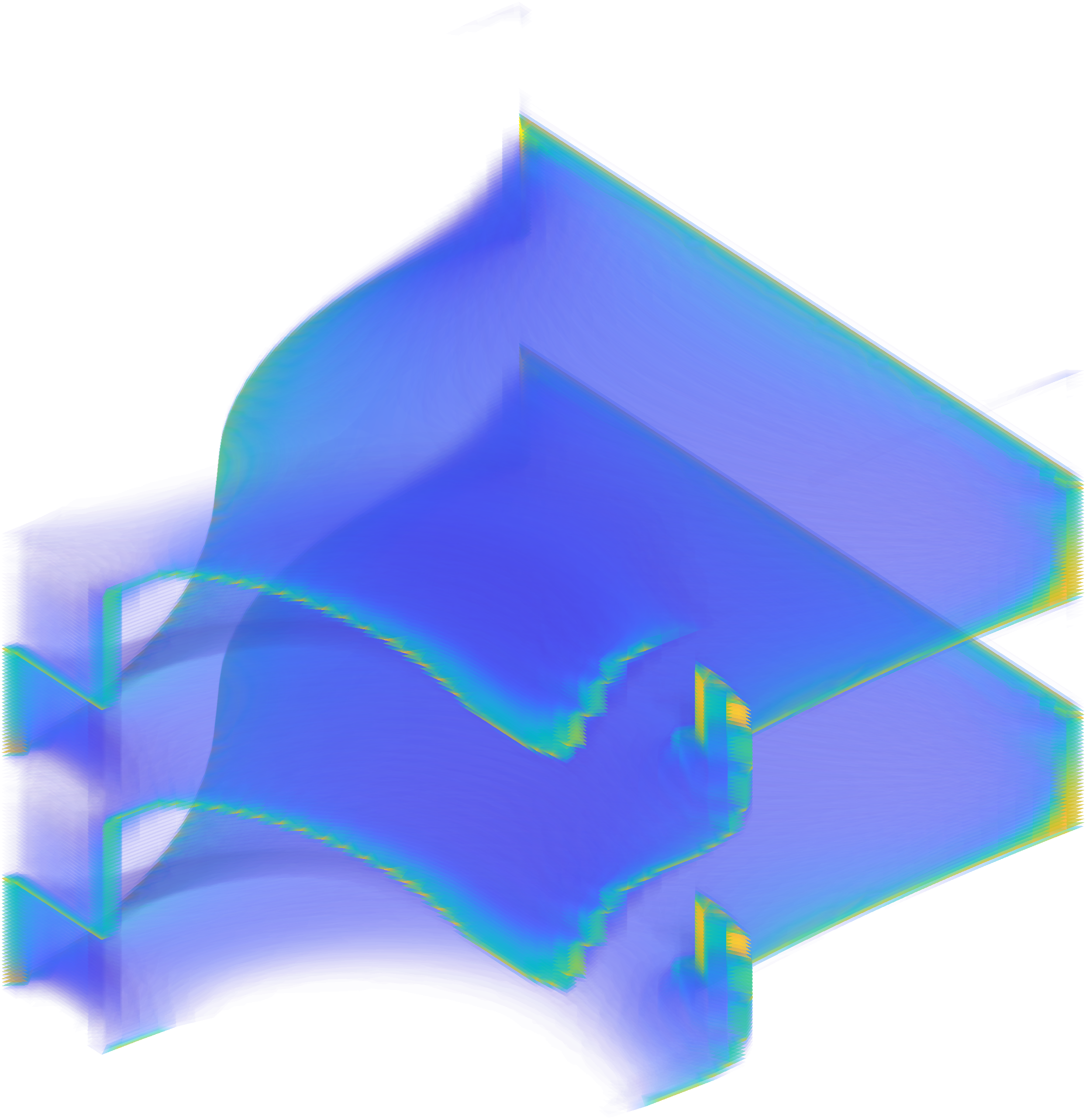} &%
\includegraphics[width=\imgwidth]{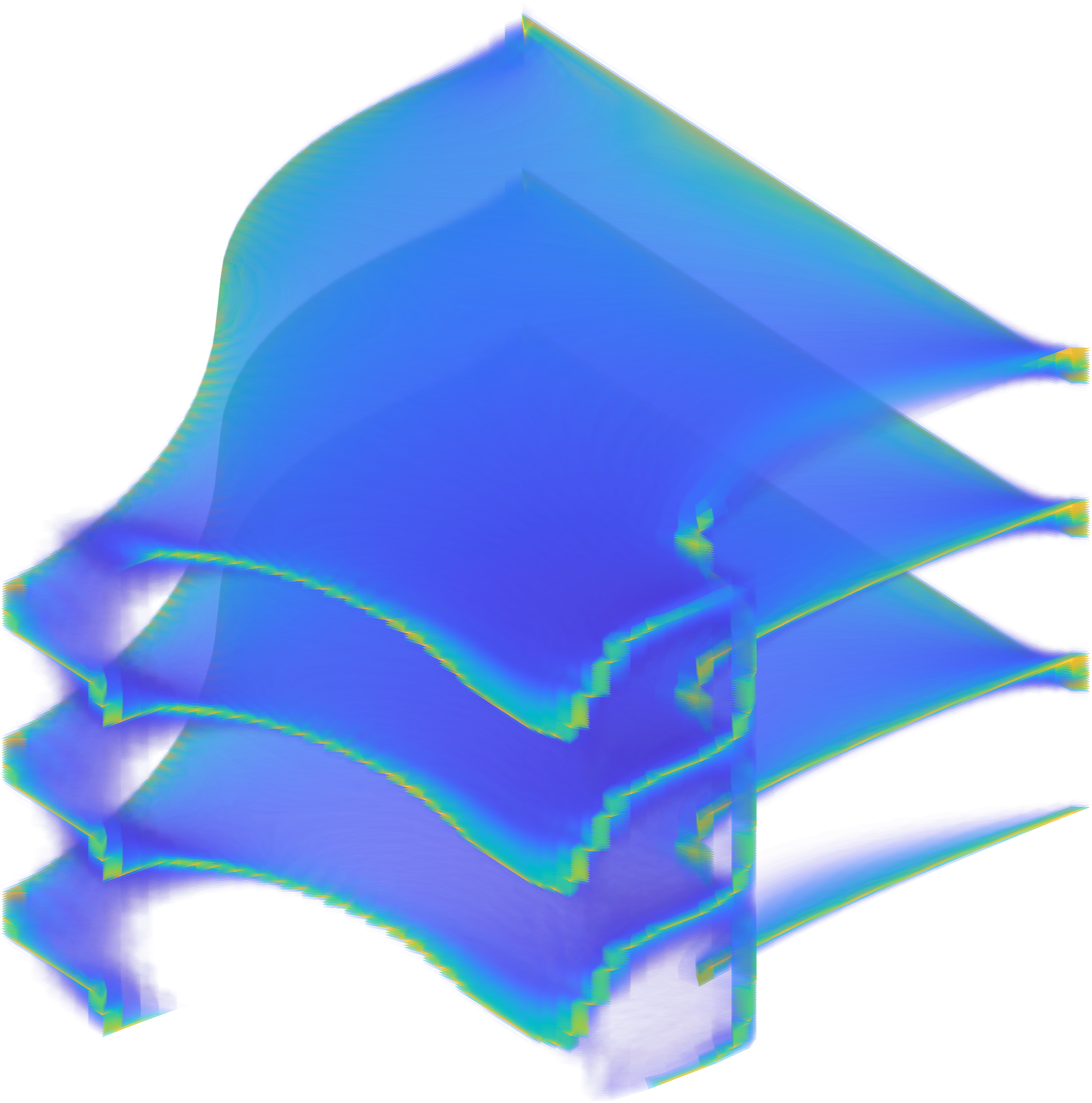} &%
\includegraphics[width=\imgwidth]{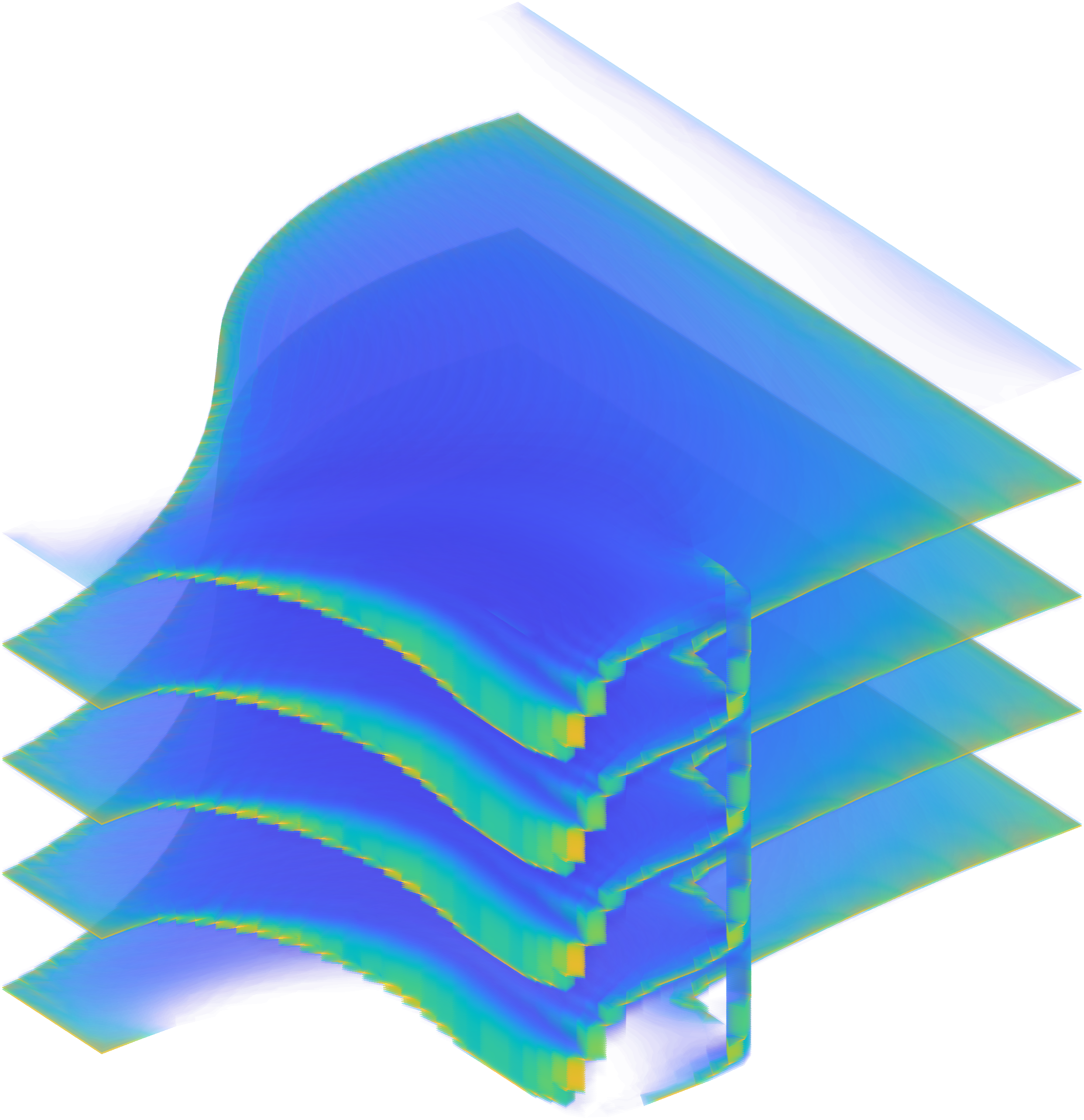}\\%
\includegraphics[width=\imgwidth]{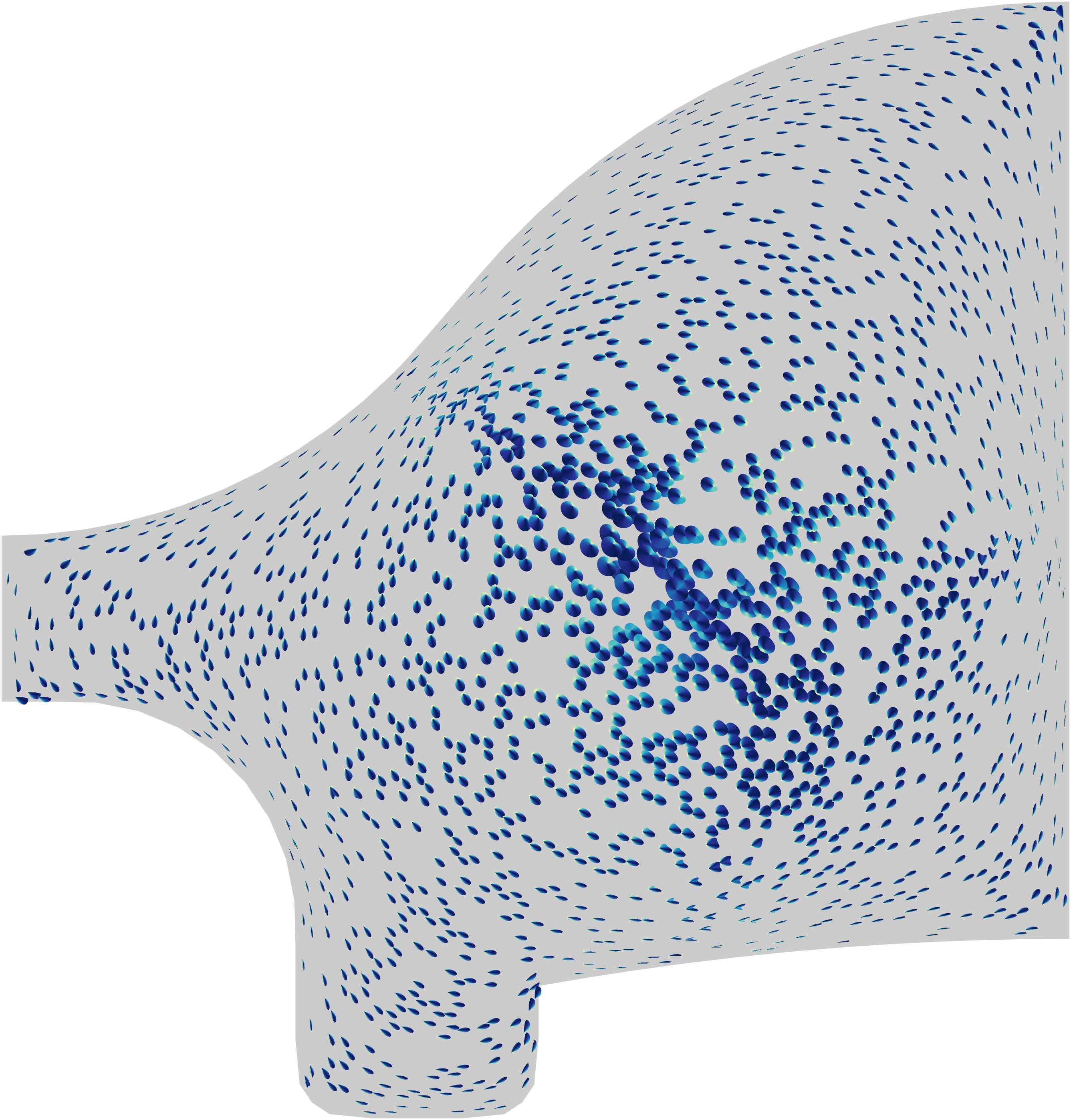} &%
\includegraphics[width=\imgwidth]{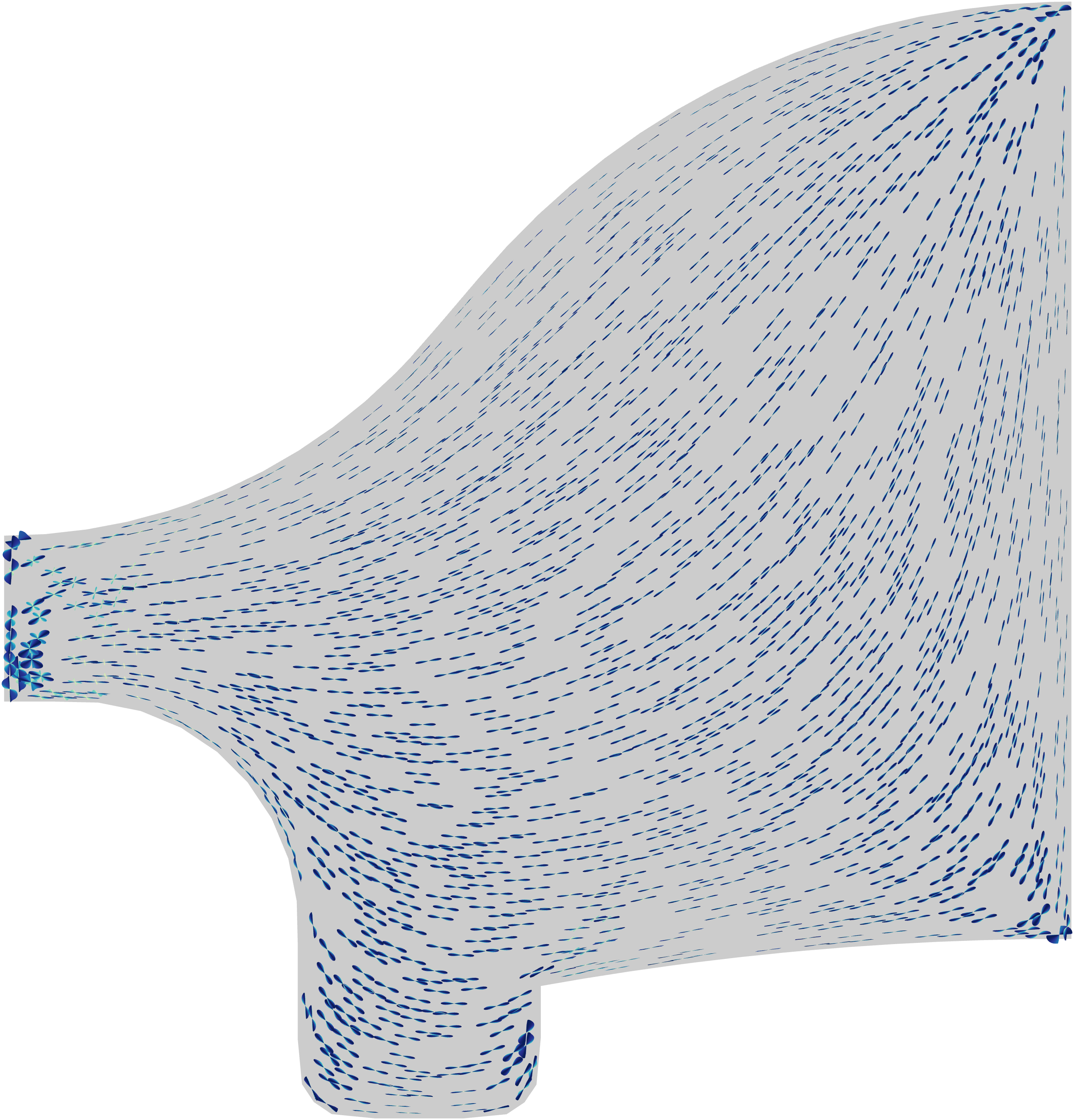} &%
\includegraphics[width=\imgwidth]{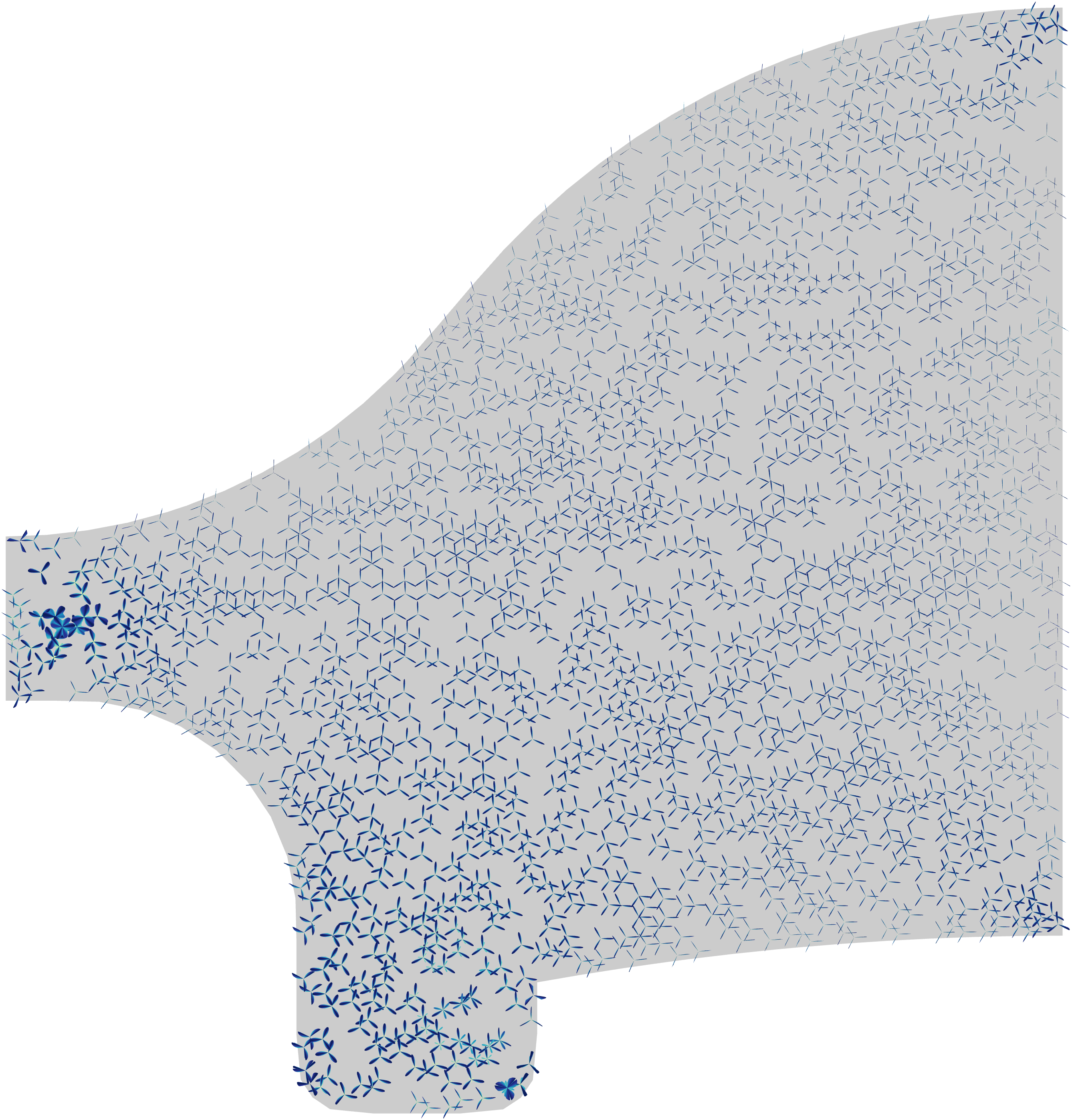} &%
\includegraphics[width=\imgwidth]{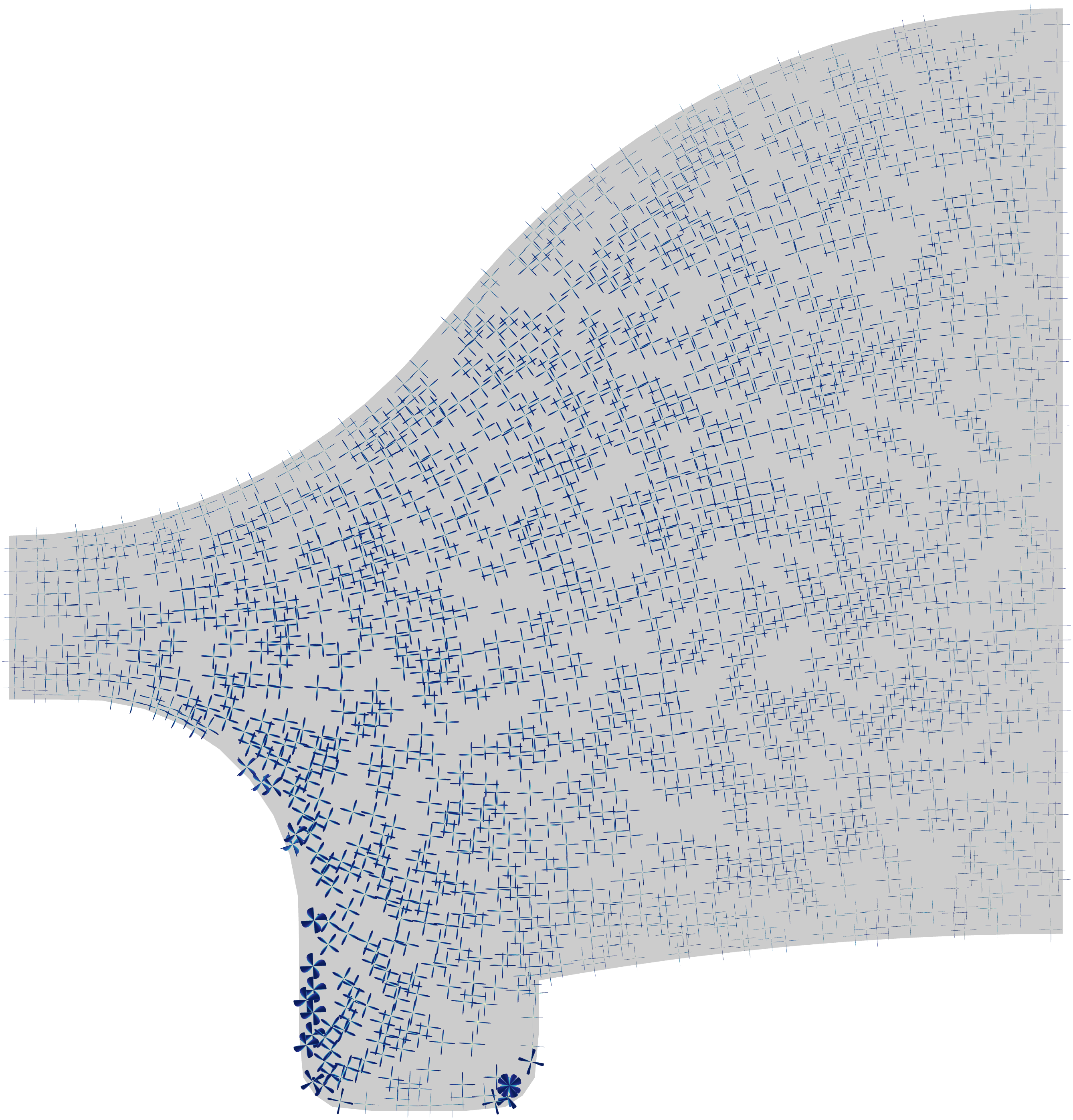}\\%
\includegraphics[width=\imgwidth]{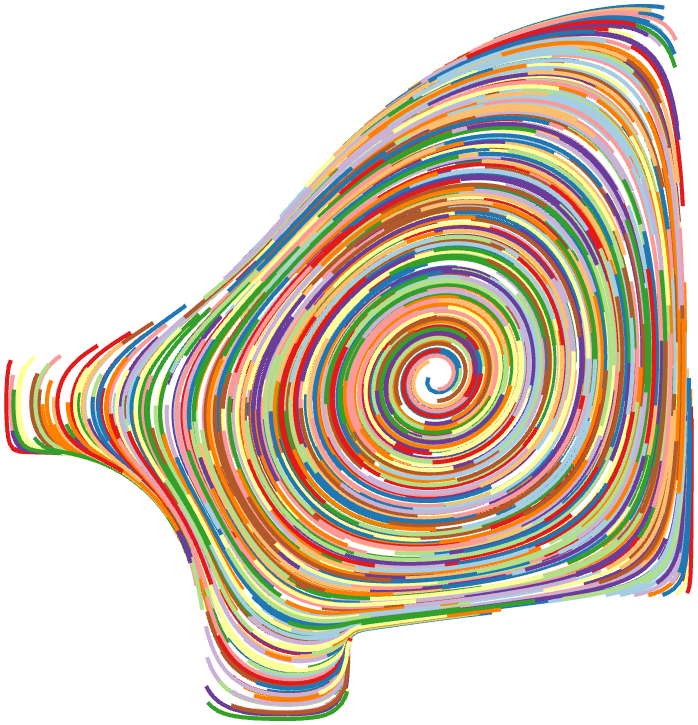} &%
\includegraphics[width=\imgwidth]{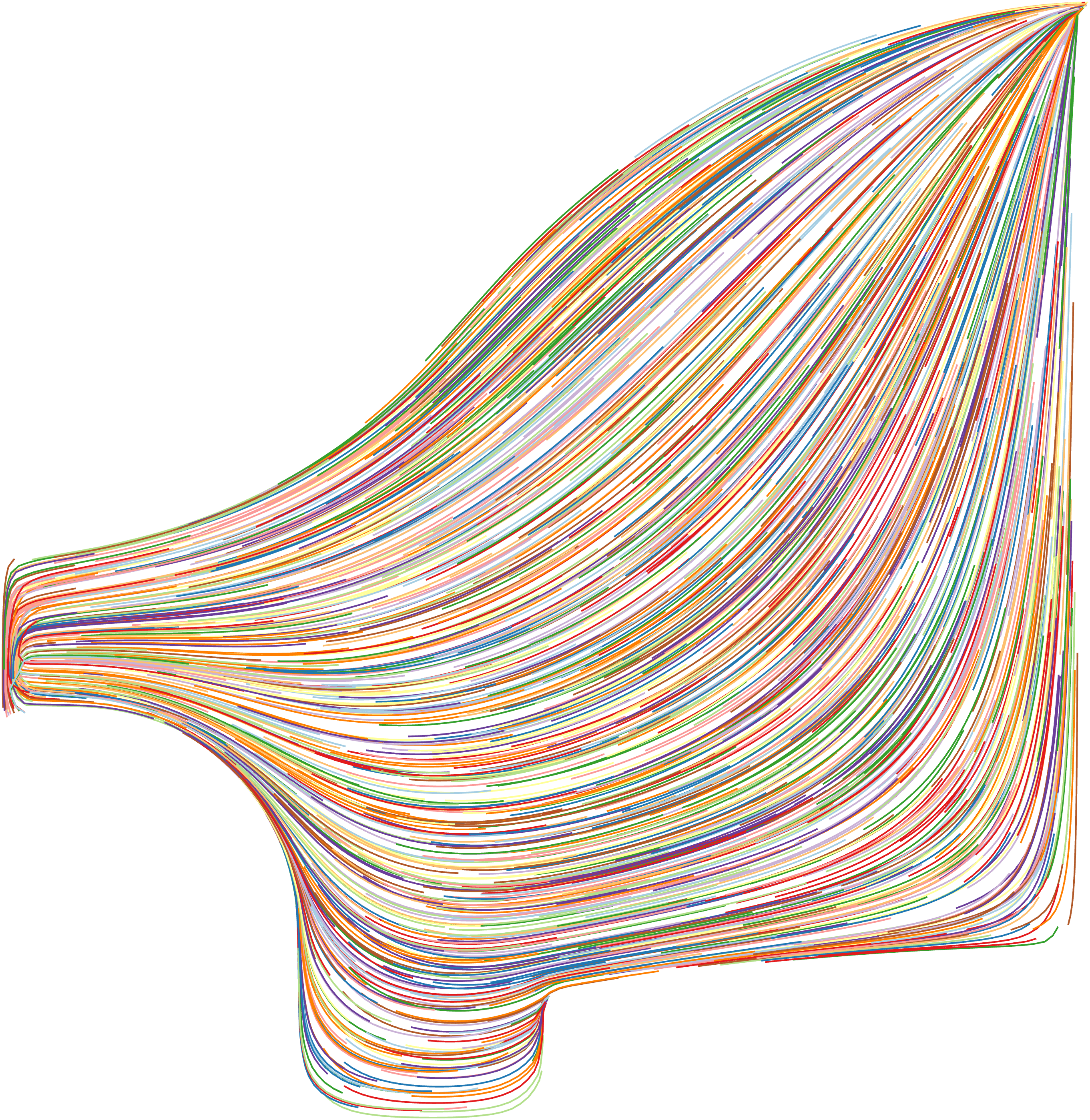} &%
\includegraphics[width=\imgwidth]{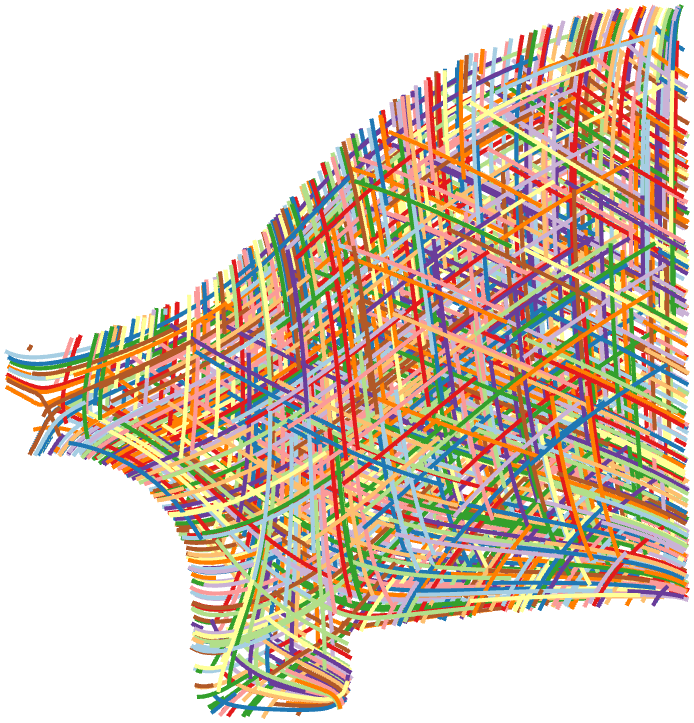} &%
\includegraphics[width=\imgwidth]{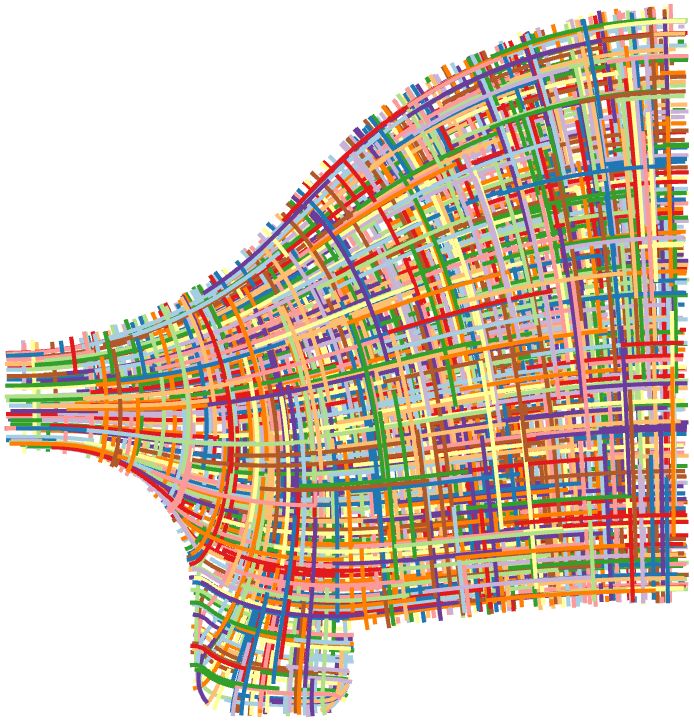}\\%
$d=1$ & $d=2$ & $d=3$ & $d=4$
\end{tabulary}
\caption{Minimal sections on bundles of increasing degree on the same base domain.}
\label{minsec:fig:degree}
\end{figure}

\paragraph{Parameters.} There are three tunable parameters in our method: \begin{enumerate*}[label=(\roman*)] \item the number of discrete segments $N$ of the fiber, which also sets the number of Fourier frequency components $K$; \item the regularization parameter $\lambda$, which controls the importance of sparsity of the singularities; and \item the fiber radius $r$, which controls the smoothness of solutions. \end{enumerate*} We find that using $N = 64$ (frequencies up to $K = \pm 31$) is suitable for most examples.

As shown in \Cref{minsec:fig:lambda}, increasing $\lambda$ penalizes singularities more relative to field energy, yielding a cross field with more curvature.
\begin{figure}
	\newcommand{\imgwidth}{0.45\columnwidth}
	\centering
	\begin{tabulary}{0.95\columnwidth}{@{}CC@{}}
		\includegraphics[width=\imgwidth]{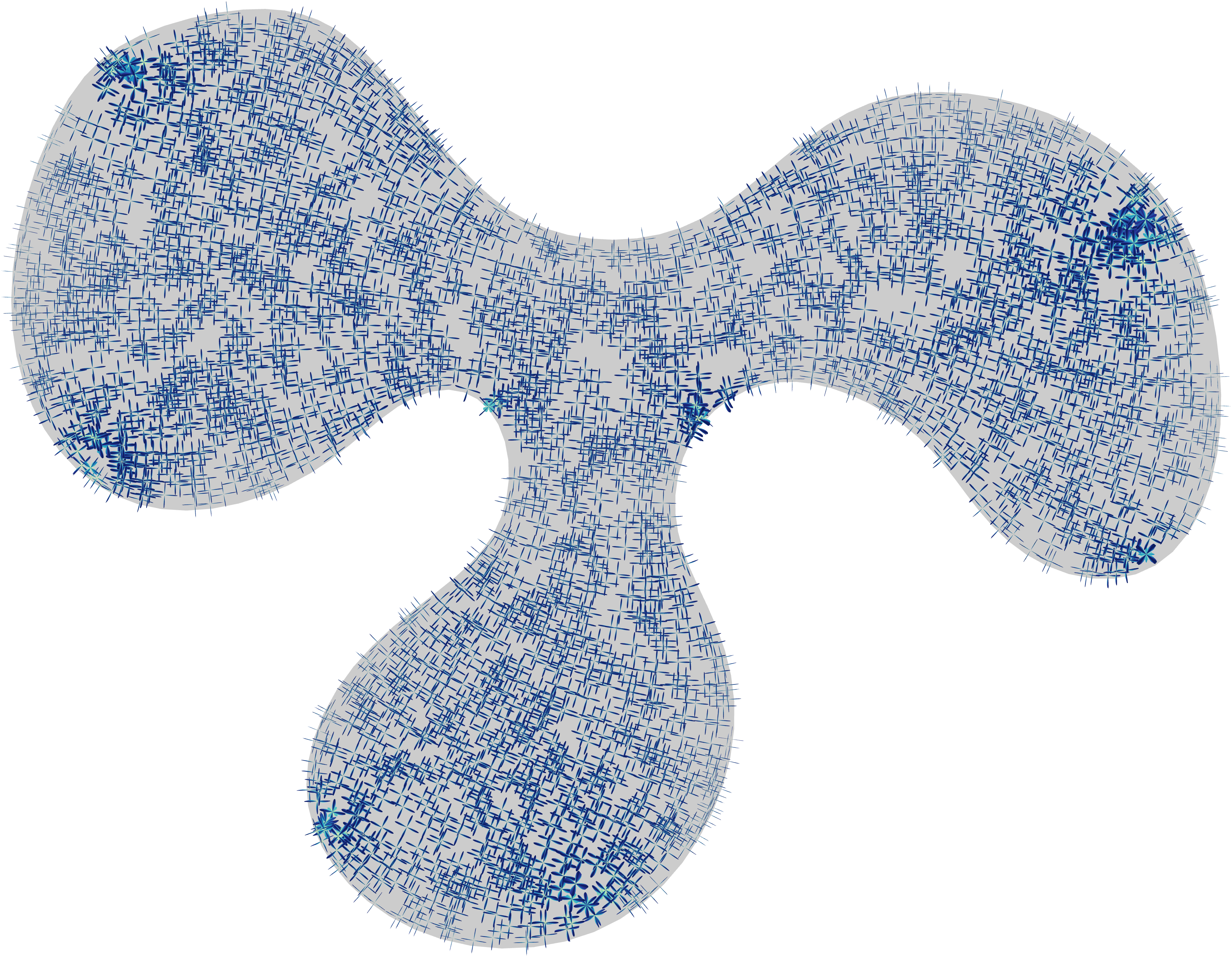} &
		\includegraphics[width=\imgwidth]{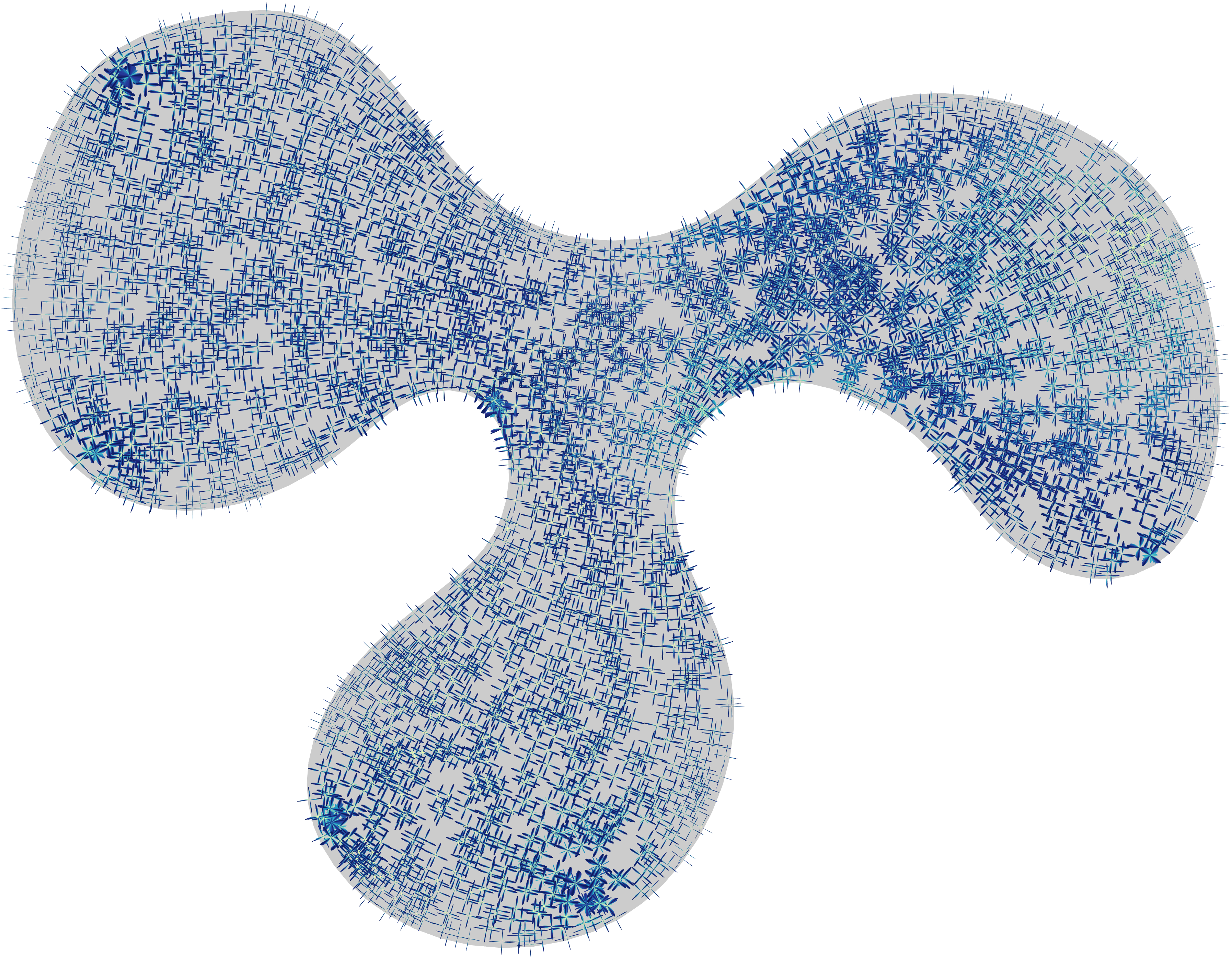} \\
		\includegraphics[width=\imgwidth]{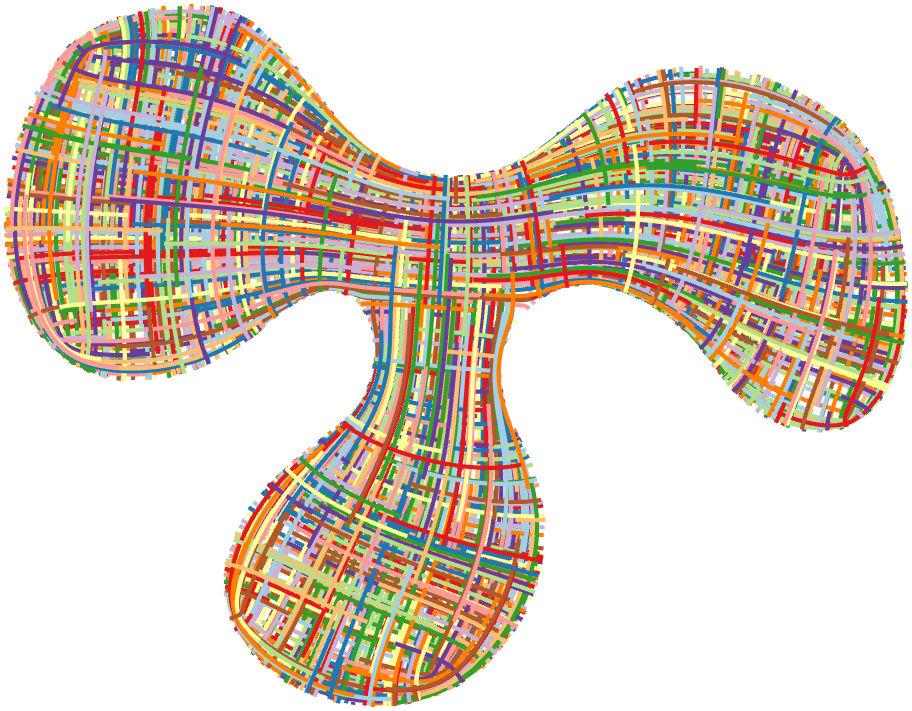} &
		\includegraphics[width=\imgwidth]{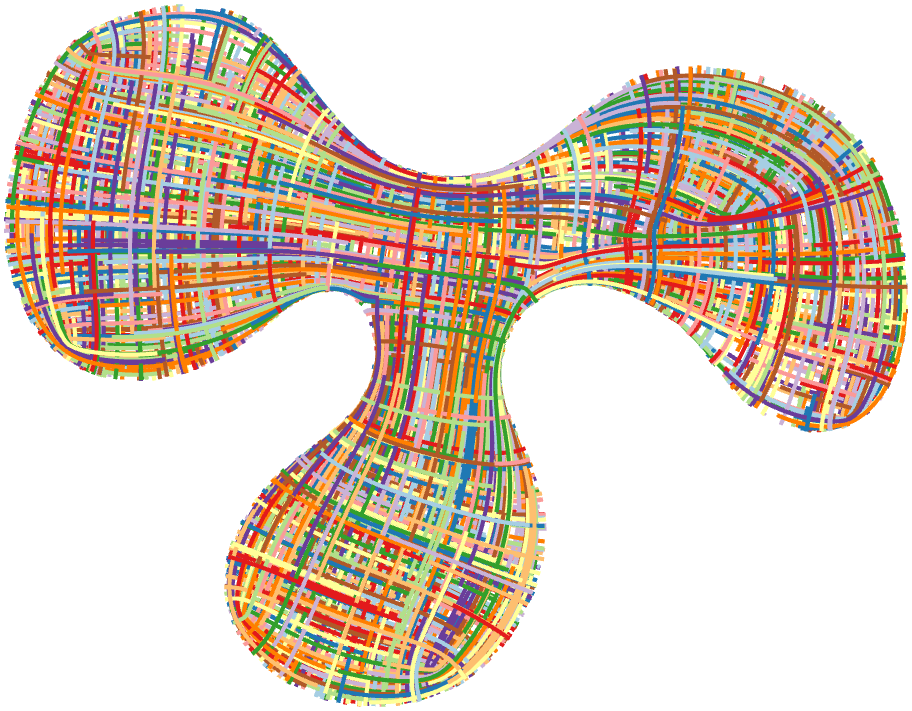} \\
		$\lambda = 2\pi$ & $\lambda = 10$
	\end{tabulary}
	\caption{With base area $100$ and $r = 1$, changing $\lambda$ trades off the position and number of singularities against overall smoothness of the field.}
	\label{minsec:fig:lambda}
\end{figure}
Meanwhile, the fiber radius $r$ controls whether the mass norm behaves more like Dirichlet energy or like total variation. It is defined relative to the area of the base, as uniformly scaling the bundle should only uniformly scale areas of sections. As shown in \Cref{minsec:fig:fiber-length}, changing the value of $r$ changes the character of the solution, from nearly piecewise constant to more smoothly varying.
\begin{figure}
	\newcommand{\imgwidth}{0.2\columnwidth}
	\centering
	\begin{tabulary}{0.95\columnwidth}{@{}CC|CC@{}}
	\includegraphics[width=\imgwidth]{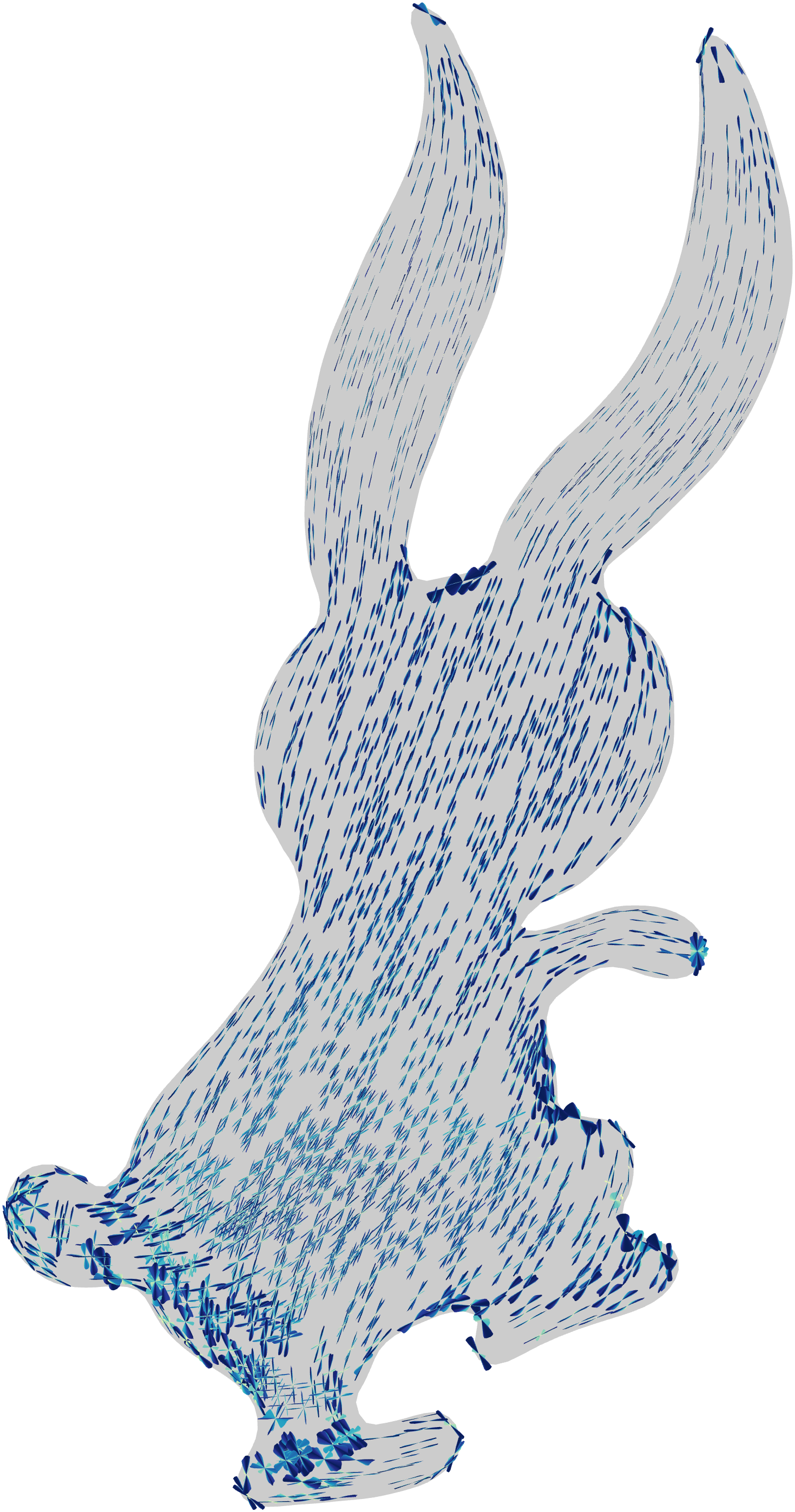} &
	\includegraphics[width=\imgwidth]{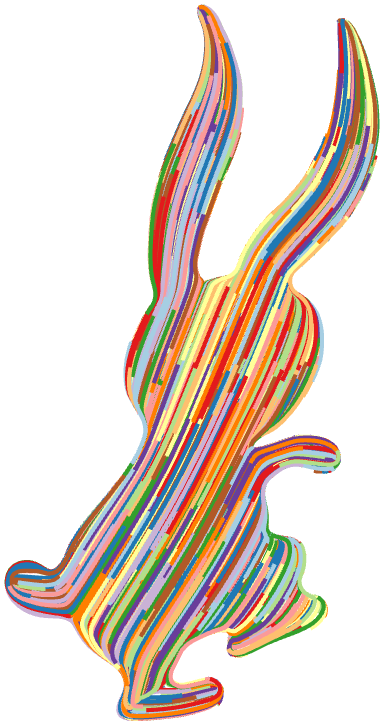} &
	\includegraphics[width=\imgwidth]{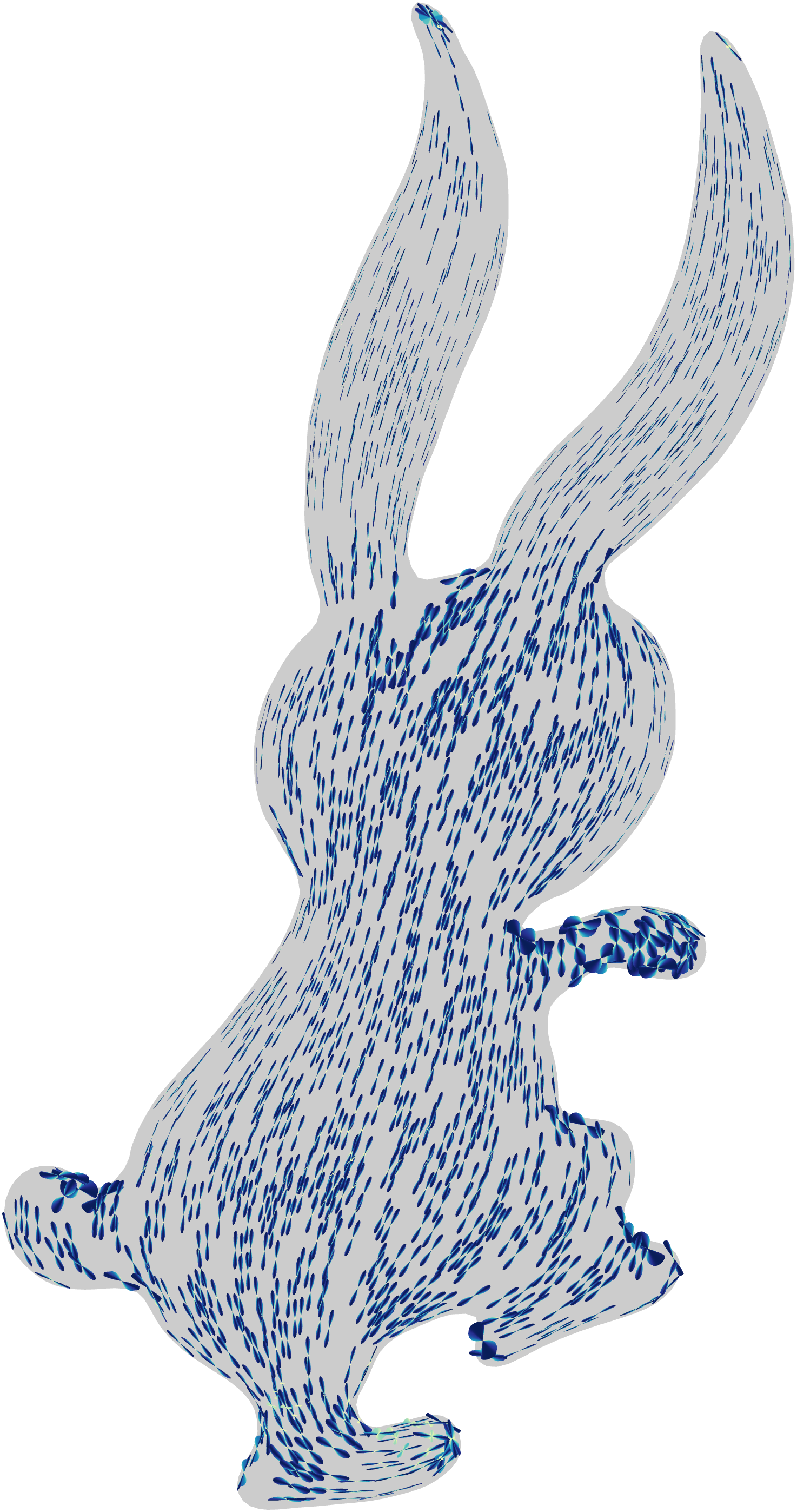} &
	\includegraphics[width=\imgwidth]{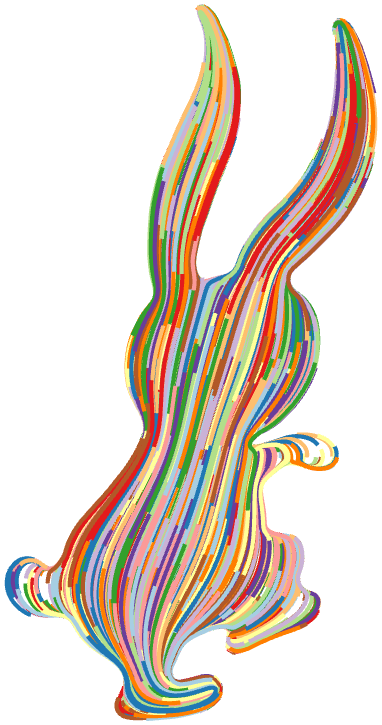} \\
	\multicolumn{2}{c}{$r = 1$} & \multicolumn{2}{c}{$r = 0.1$}
	\end{tabulary}
	\caption{With longer fibers, the field energy acts more like total variation, yielding a line field that is nearly constant in the interior with few areas of sharp curvature near the boundary (left). Reducing the fiber length leads to a field that is smoother overall (right). Both fields are computed on a base surface of area $1$, with $\lambda = 0.1$.}
	\label{minsec:fig:fiber-length}
\end{figure}
In \Cref{minsec:fig:relaxation}, we see that $r$ also controls how concentrated the current is toward a sharp surface. For a disk of fixed area $1$ and $\lambda = 1$, we plot the computed current for three values of the fiber radius. It appears much more concentrated when the fiber is long than when it is short. While the vertical concentration may be an artifact of displaying all three plots with a fixed visual fiber length, the core radius of the singularity also scales inversely with $r$.
\begin{figure}
\centering
\newcommand{\imgwidth}{0.28\columnwidth}
\begin{tabulary}{0.97\columnwidth}{@{}CCC@{}}
\includegraphics[width=\imgwidth]{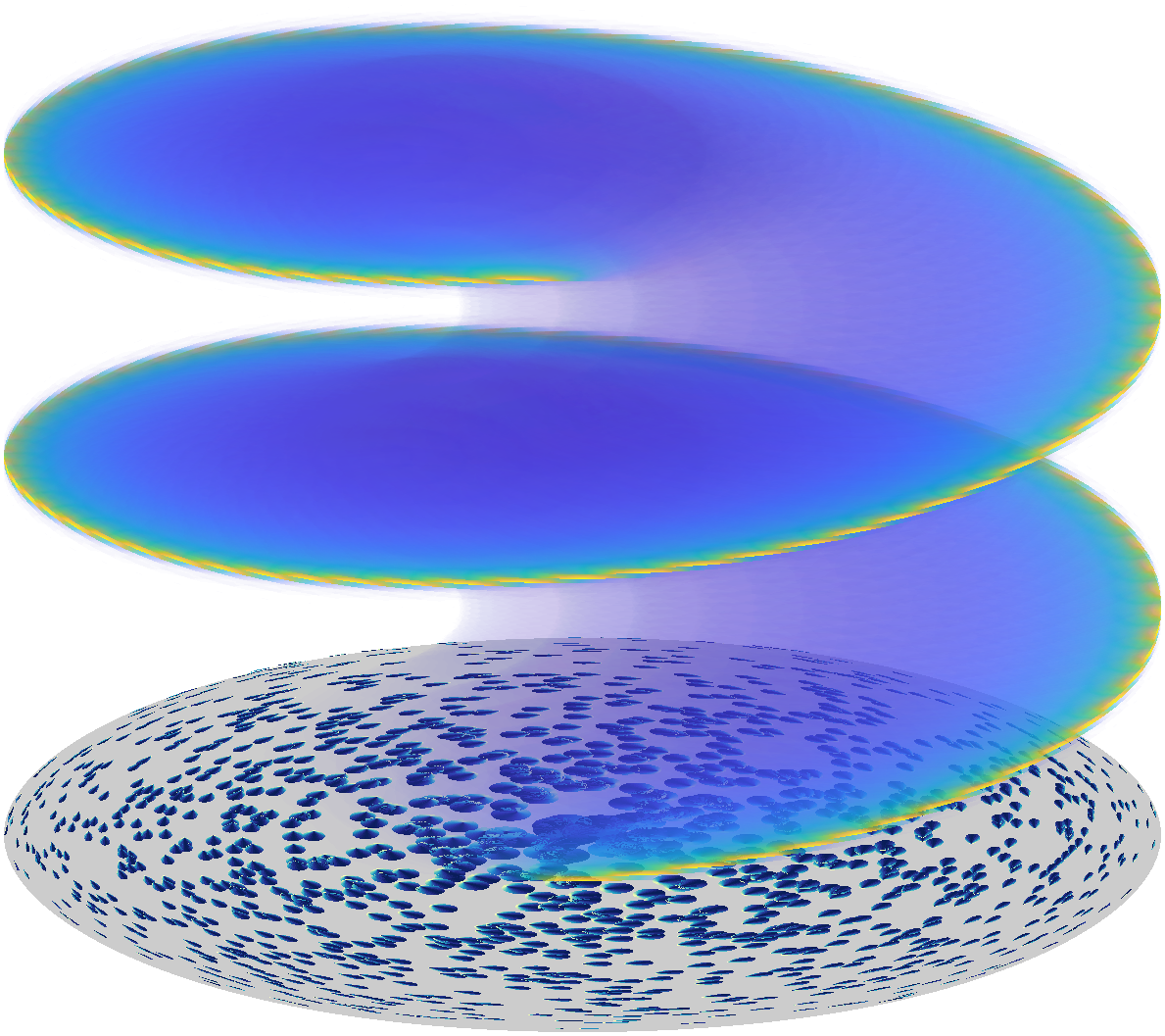} &
\includegraphics[width=\imgwidth]{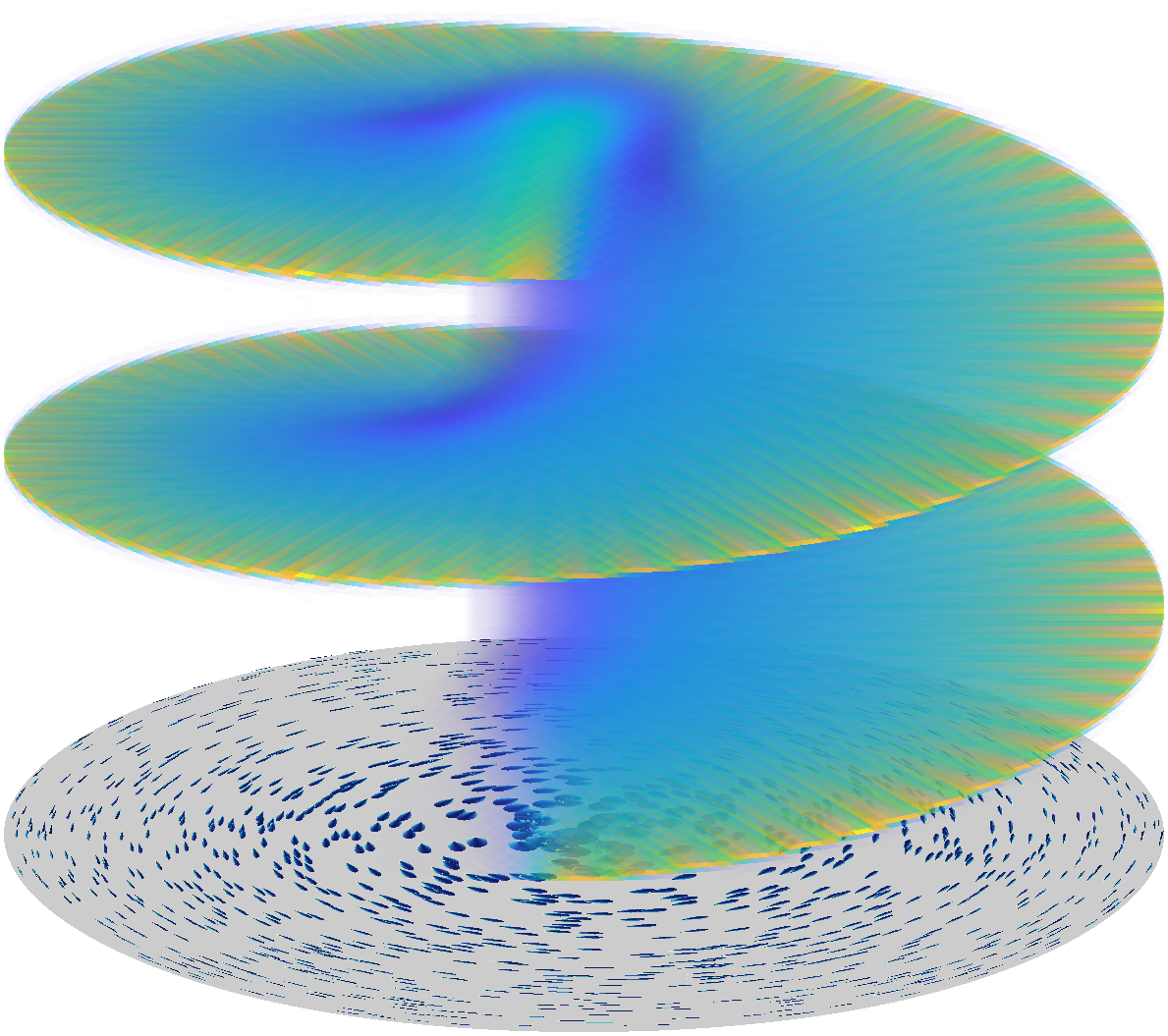} &
\includegraphics[width=\imgwidth]{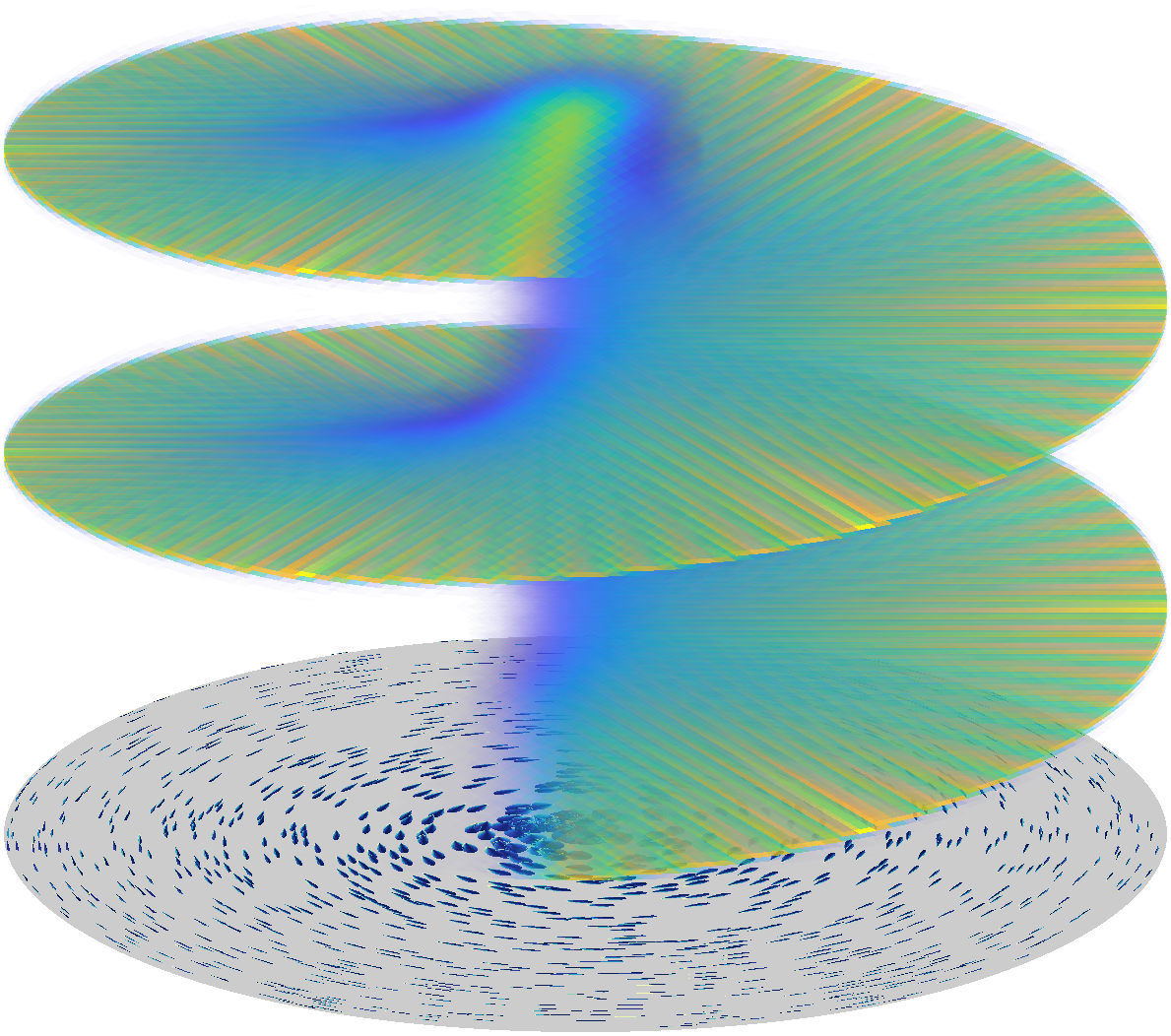} \\
$r = 0.1$ & $r = 1$ & $r = 10$	
\end{tabulary}
\caption{Increasing fiber length $r$ results in more-concentrated singularities.}
\label{minsec:fig:relaxation}
\end{figure}

\paragraph{Curved base.} \Cref{minsec:fig:curved} shows cross fields computed as minimal sections on curved patches with boundary. The fields change continuously as the curvature of the base surface varies, except in degenerate cases such as the rotationally-symmetric disk (top left).
\begin{figure}
\centering
\newcommand{\imgwidth}{0.3\columnwidth}
\begin{tabulary}{0.97\columnwidth}{@{}CCC@{}}
\includegraphics[width=\imgwidth]{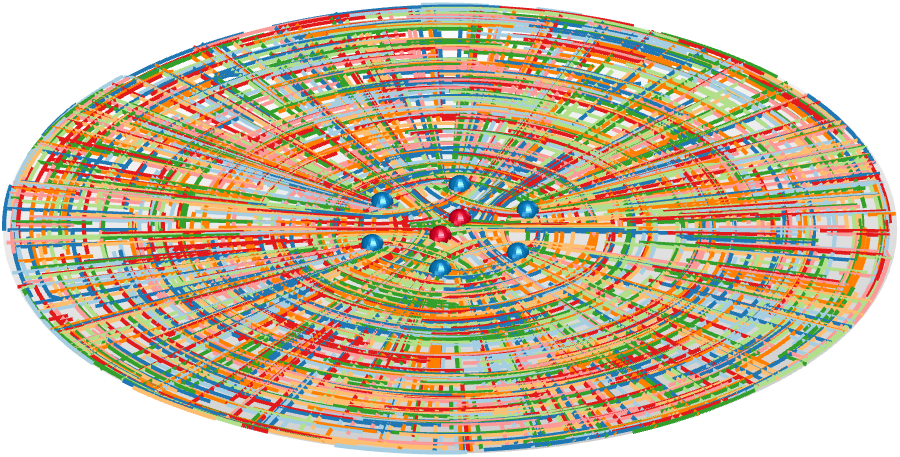} &
\includegraphics[width=\imgwidth]{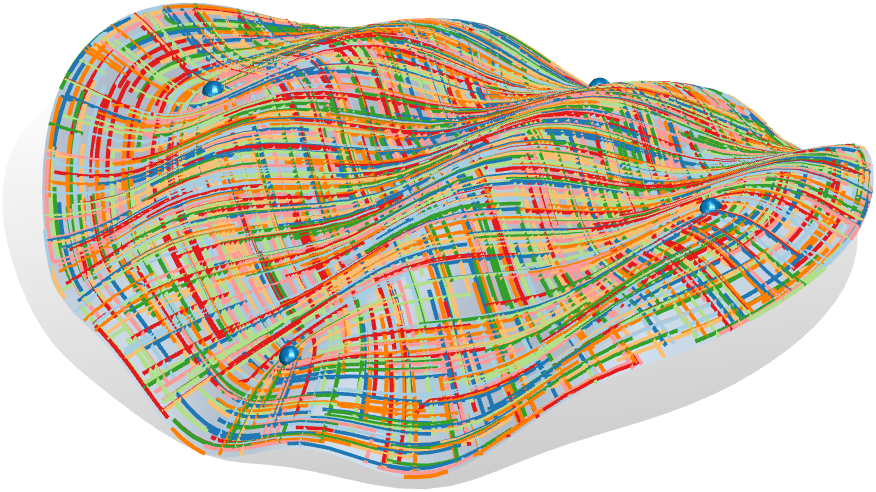} &
\includegraphics[width=\imgwidth]{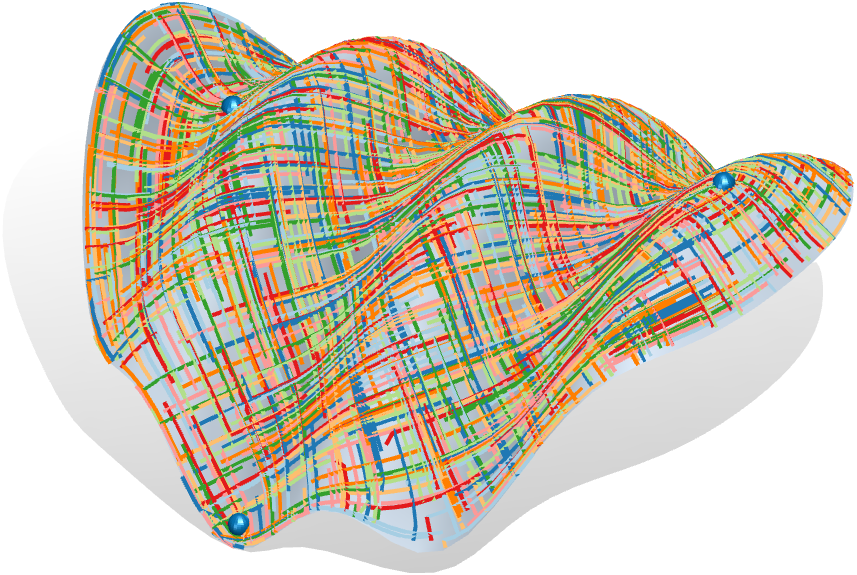} \\
\includegraphics[width=\imgwidth]{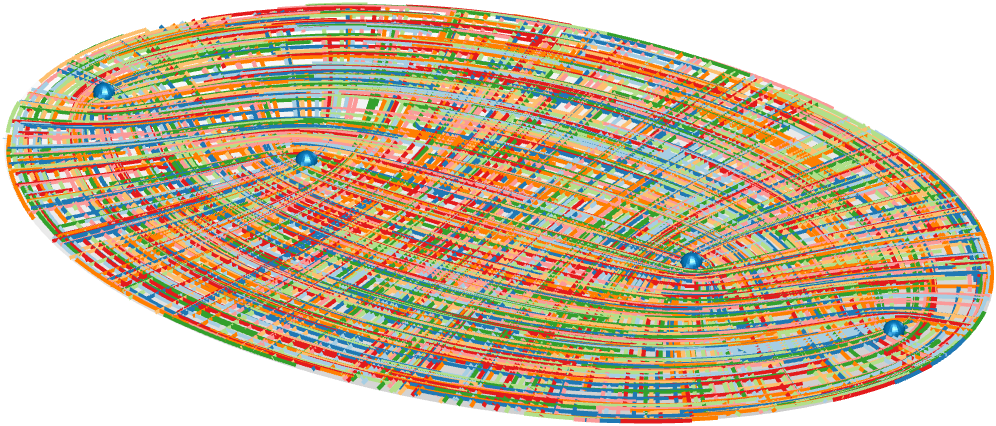} &
\includegraphics[width=\imgwidth]{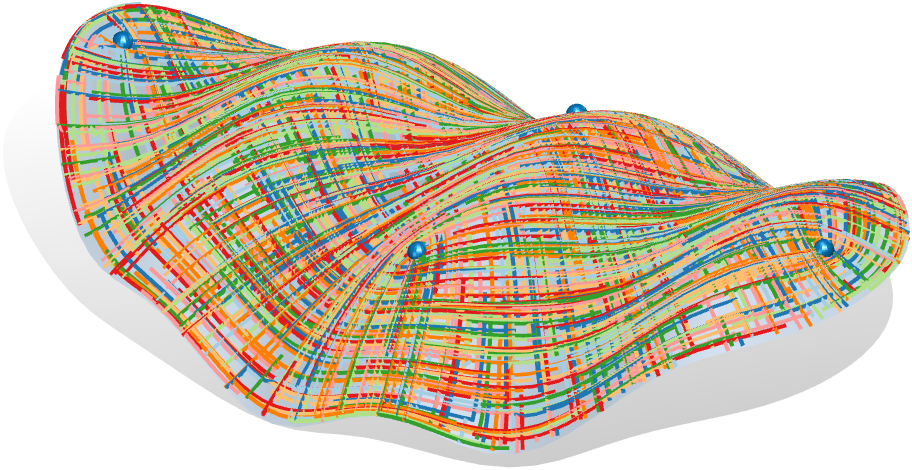} &
\includegraphics[width=\imgwidth]{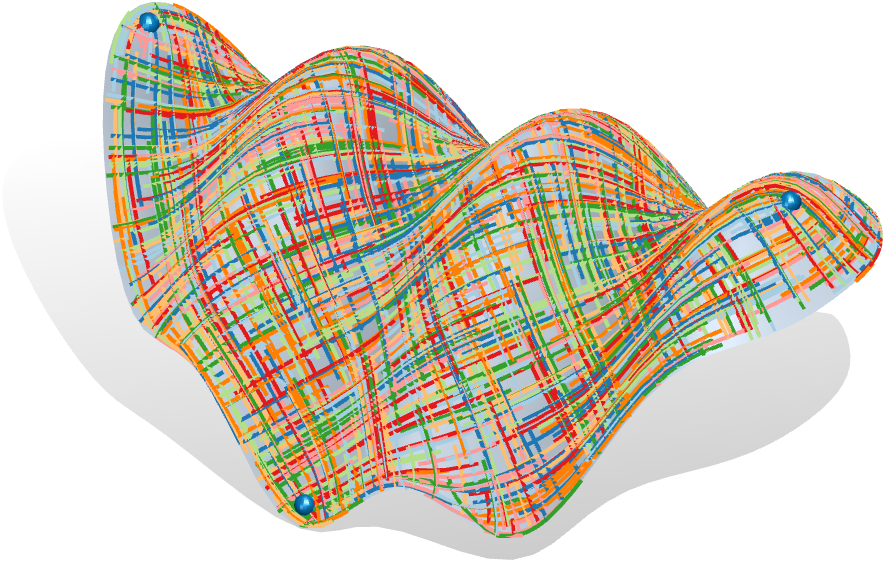}
\end{tabulary}
\caption{Minimal section relaxation yields cross fields aligned to the curvature variation of the base surface. Adding curvature variation breaks the rotational symmetry of the disk. For the ellipsoid, there is no such symmetry to begin with, so the singularities behave continuously as curvature is added to the base surface.}
\label{minsec:fig:curved}
\end{figure}

\paragraph{Masking Singularities.} Our approach makes it simple to offer user control of singularity placement. A singularity mask may be enforced by simply zeroing out the values of $\Gamma$ at masked vertices in the local step of \textsc{admm}. Alternatively, a ``soft mask'' may be implemented by promoting the regularization parameter $\lambda$ to a spatially-varying scalar field. As shown in \Cref{minsec:fig:hand-mask}, in both cases the resulting field is reconfigured to accommodate the restriction on its singularities.

\begin{figure}
\centering
\newcommand{\imgwidth}{0.3\columnwidth}
\begin{tabulary}{\columnwidth}{@{}CCC@{}}
Unmasked & Mask & Masked \\
\includegraphics[width=\imgwidth,align=c]{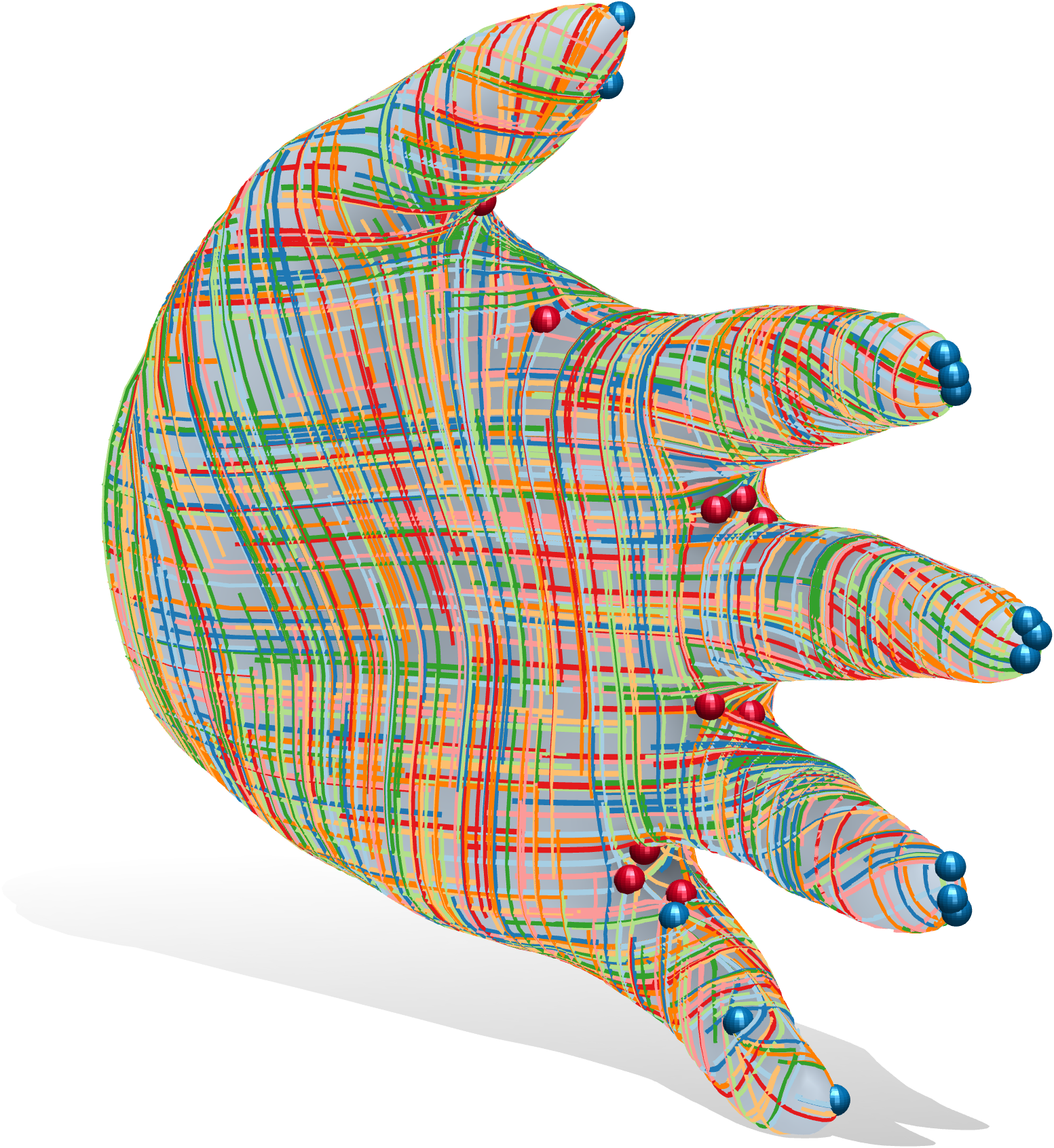} &
\includegraphics[width=\imgwidth,align=c]{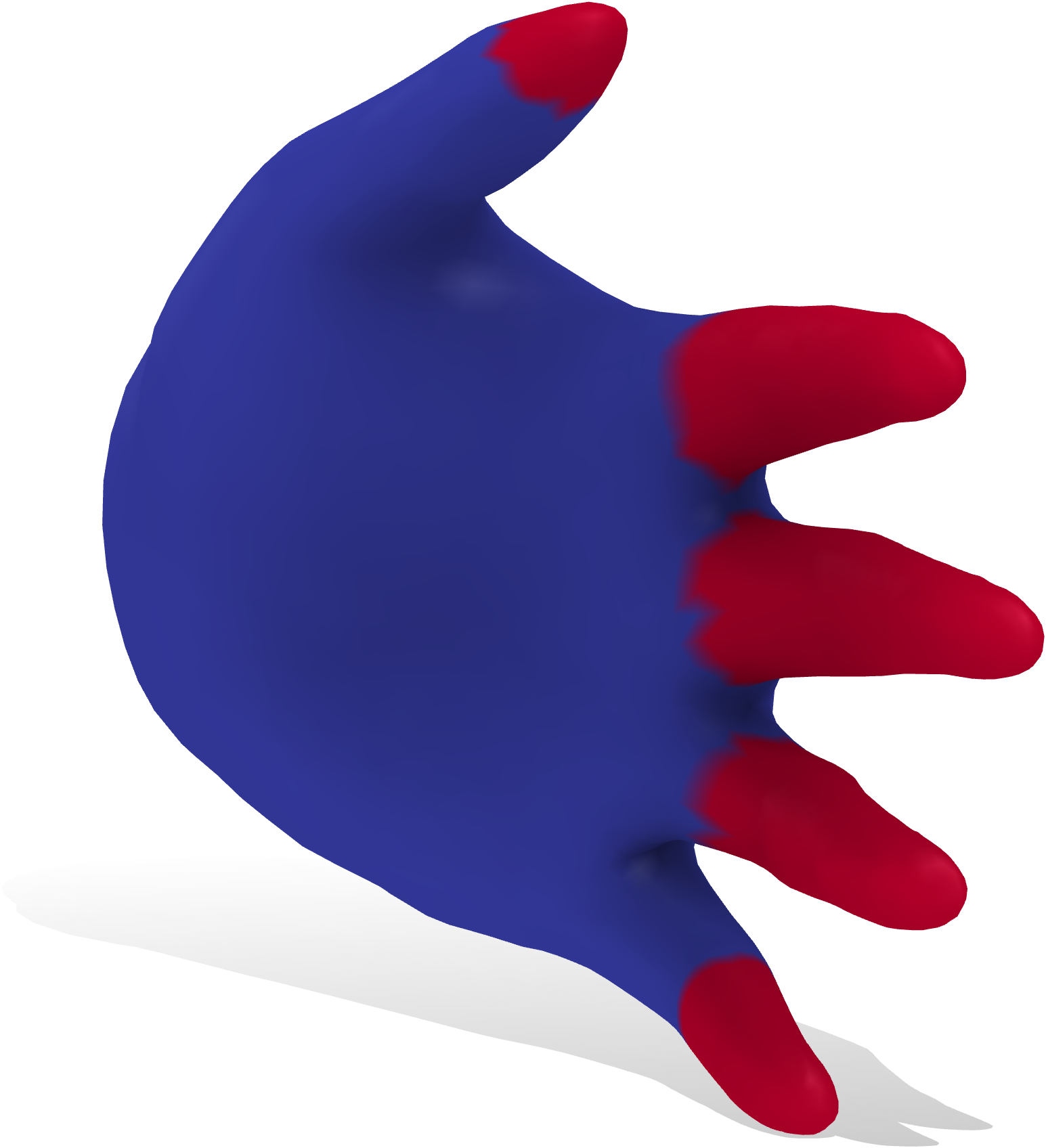} &
\includegraphics[width=\imgwidth,align=c]{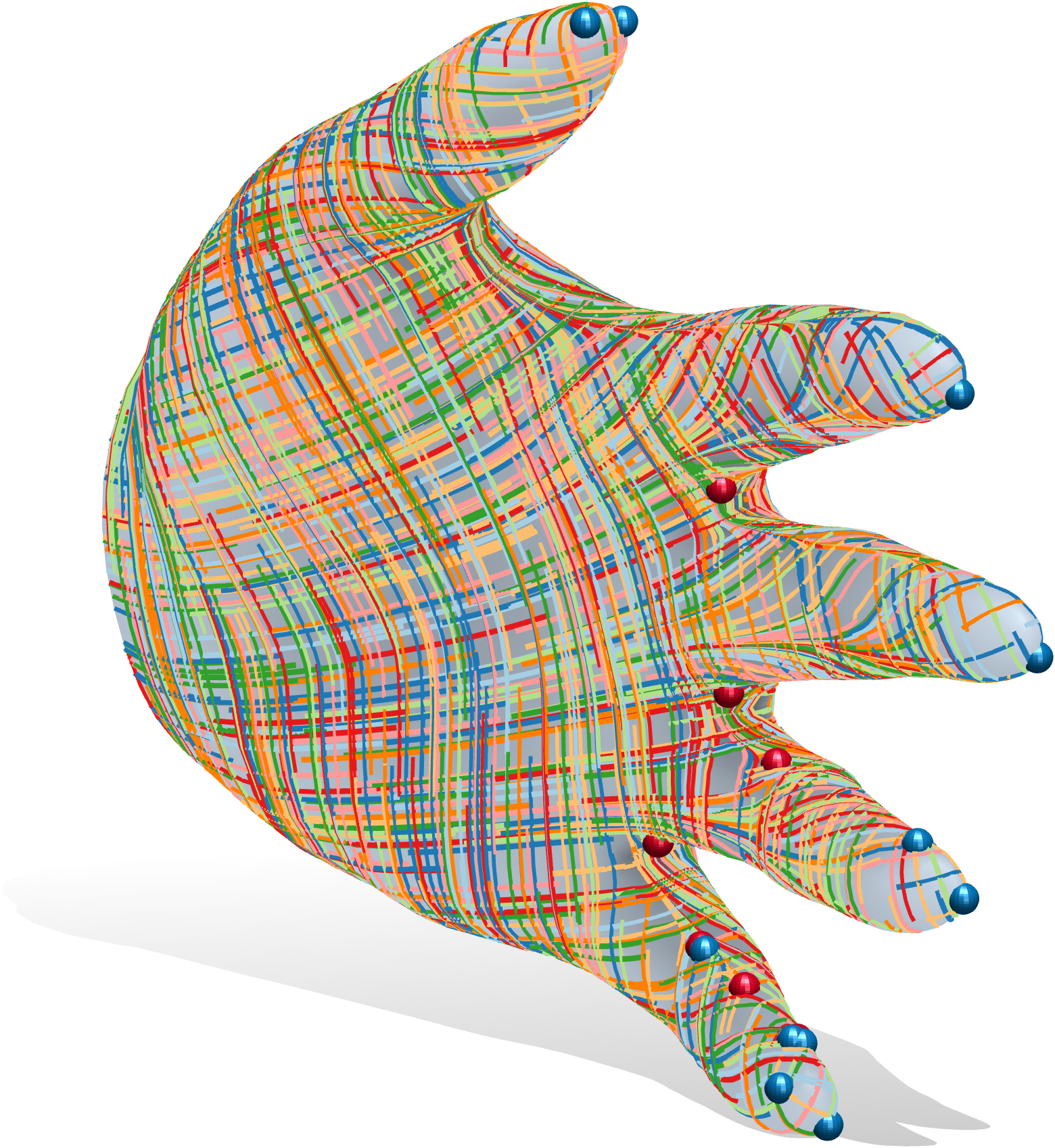} \\
\includegraphics[width=\imgwidth,align=c]{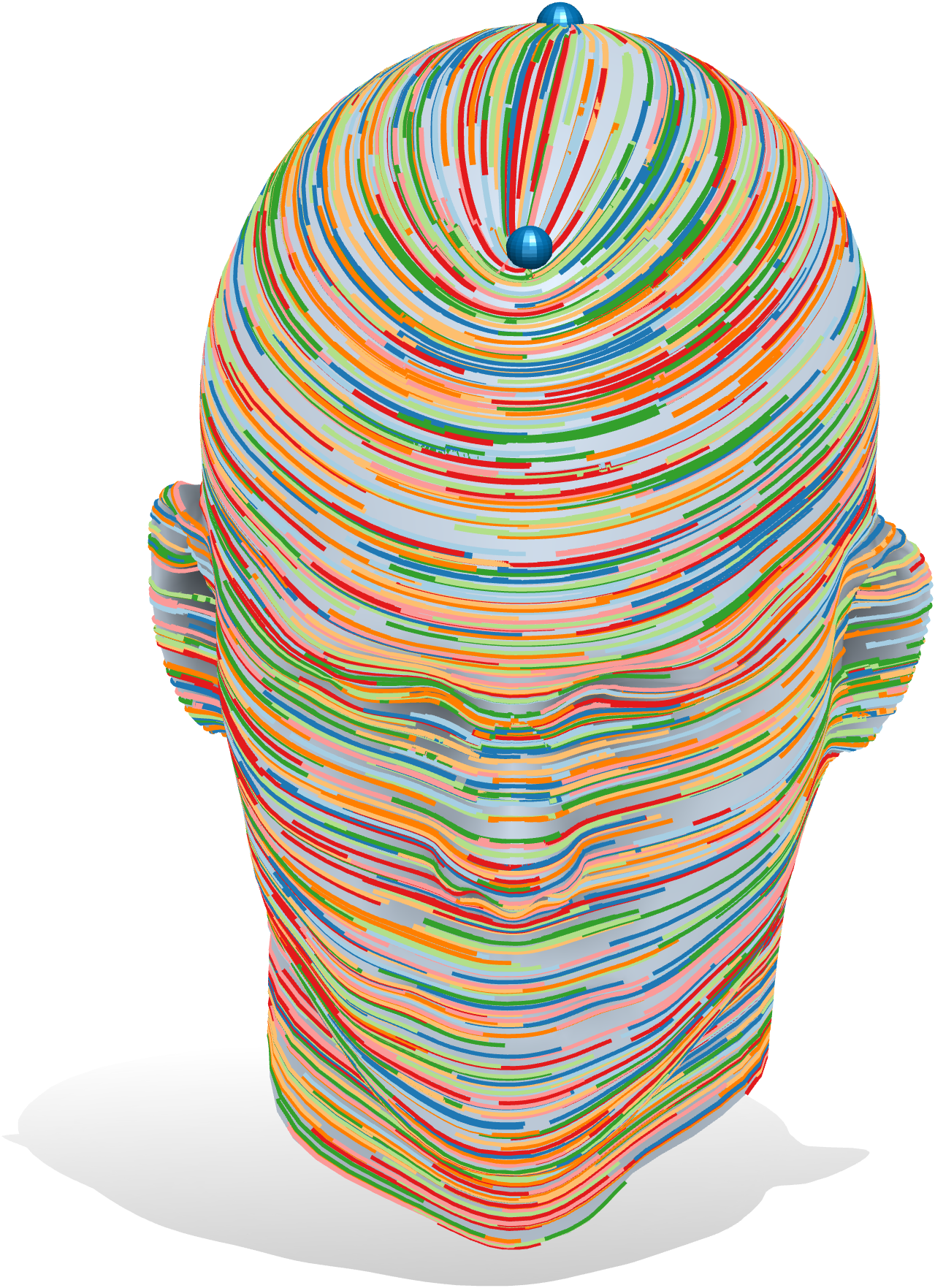} &
\includegraphics[width=\imgwidth,align=c]{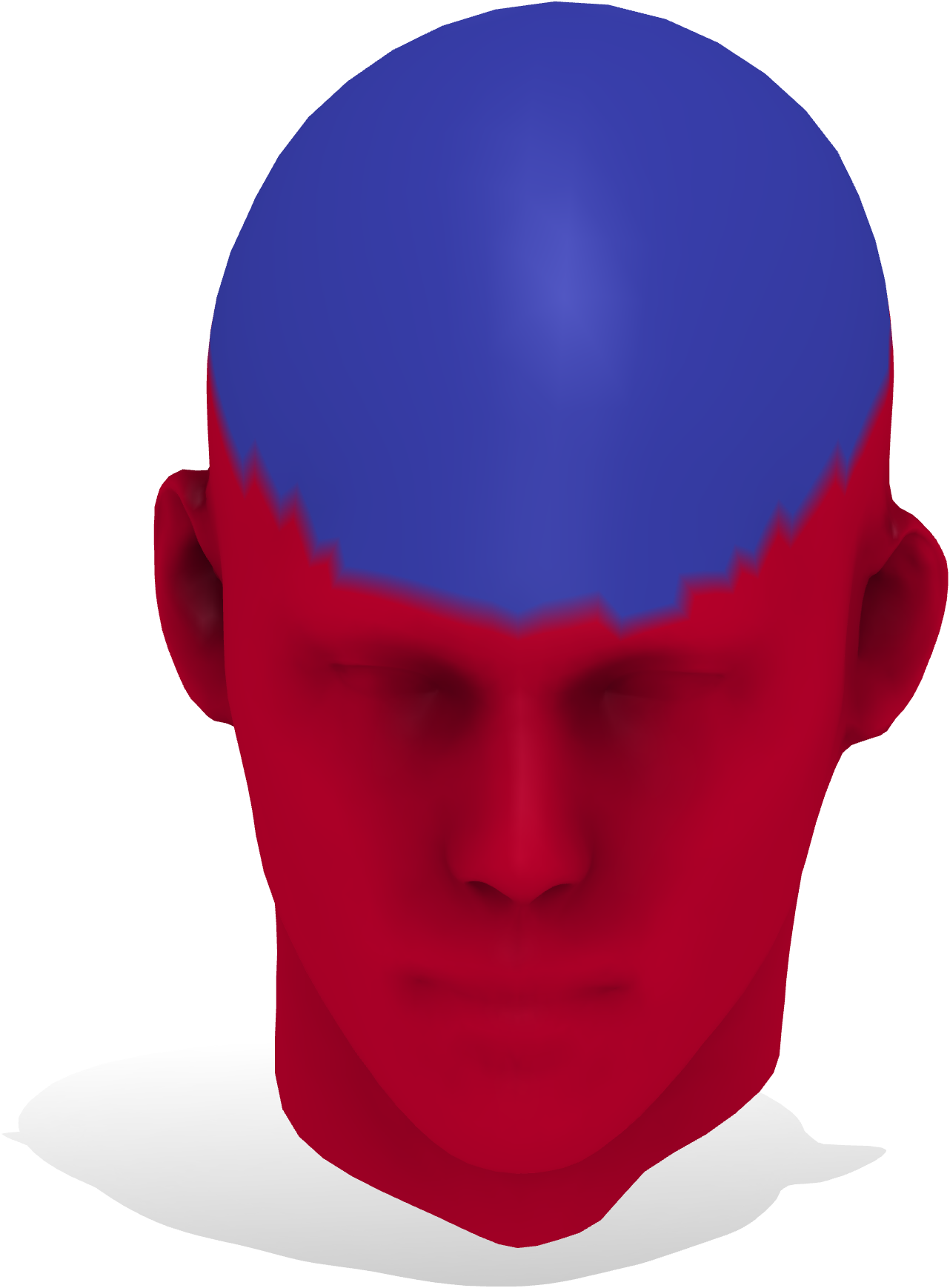} &
\includegraphics[width=\imgwidth,align=c]{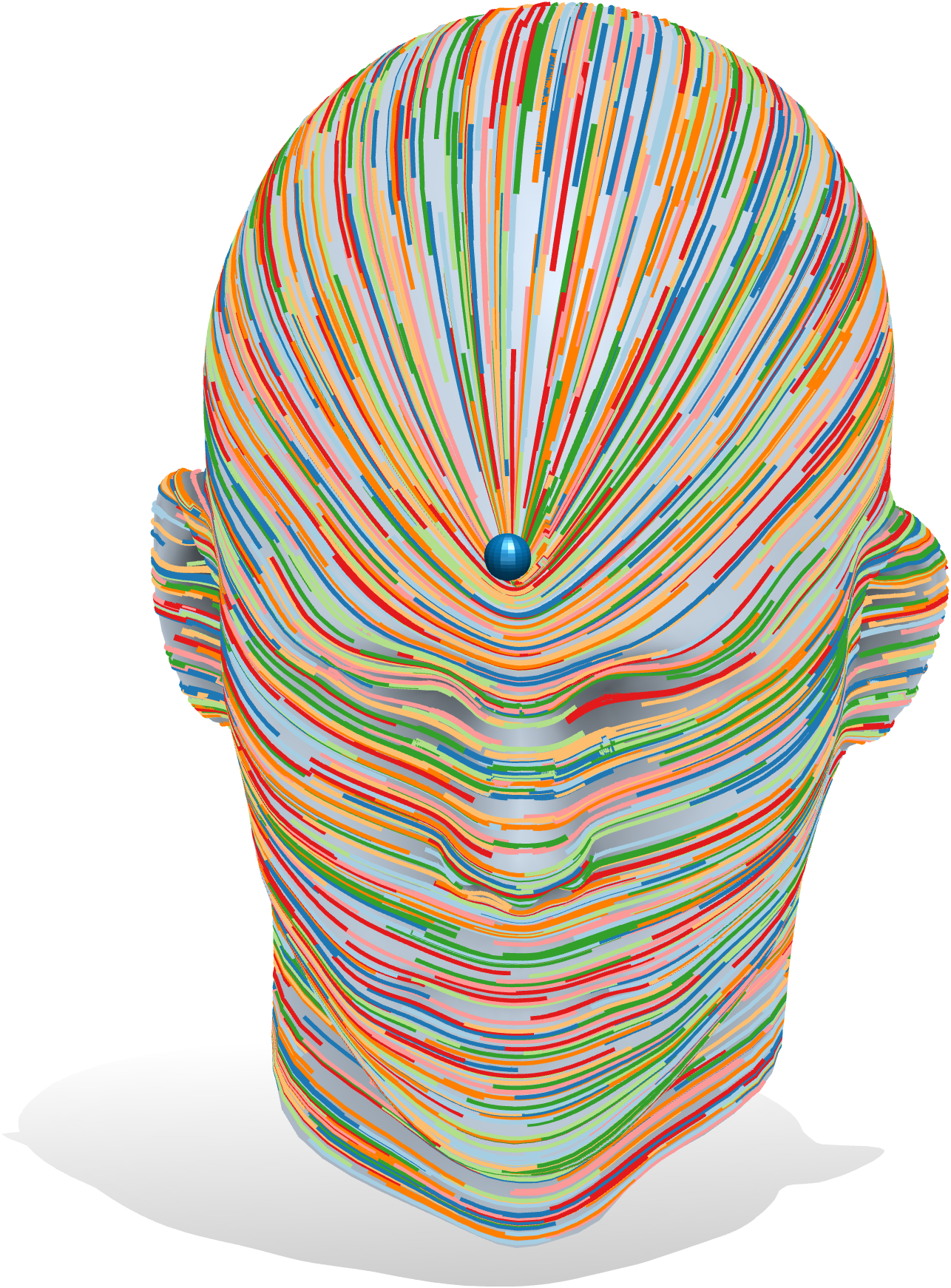} \\ \midrule
\includegraphics[width=\imgwidth,align=c]{figures/mask/head_unmasked} &
\includegraphics[width=\imgwidth,align=c]{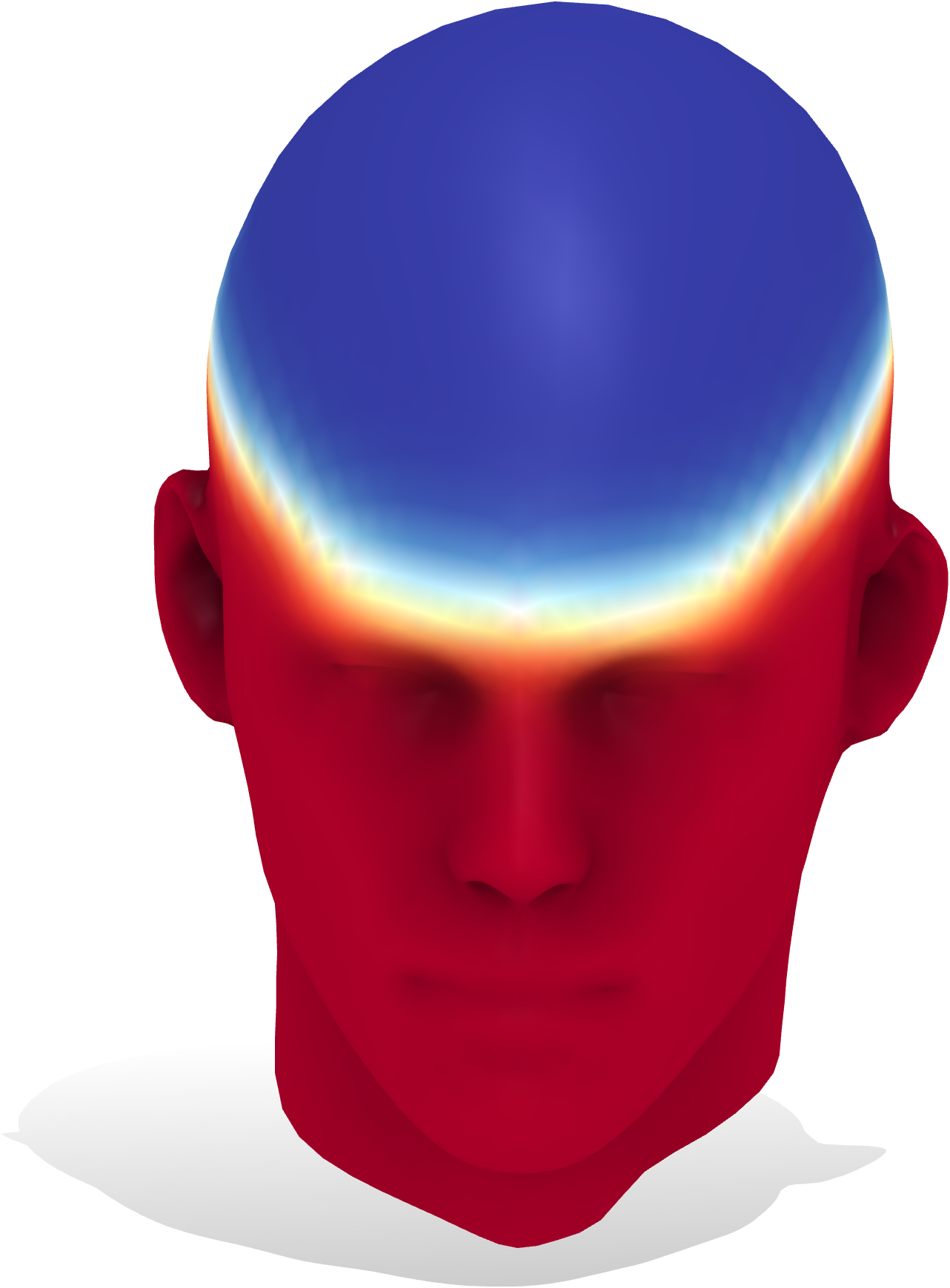} &
\includegraphics[width=\imgwidth,align=c]{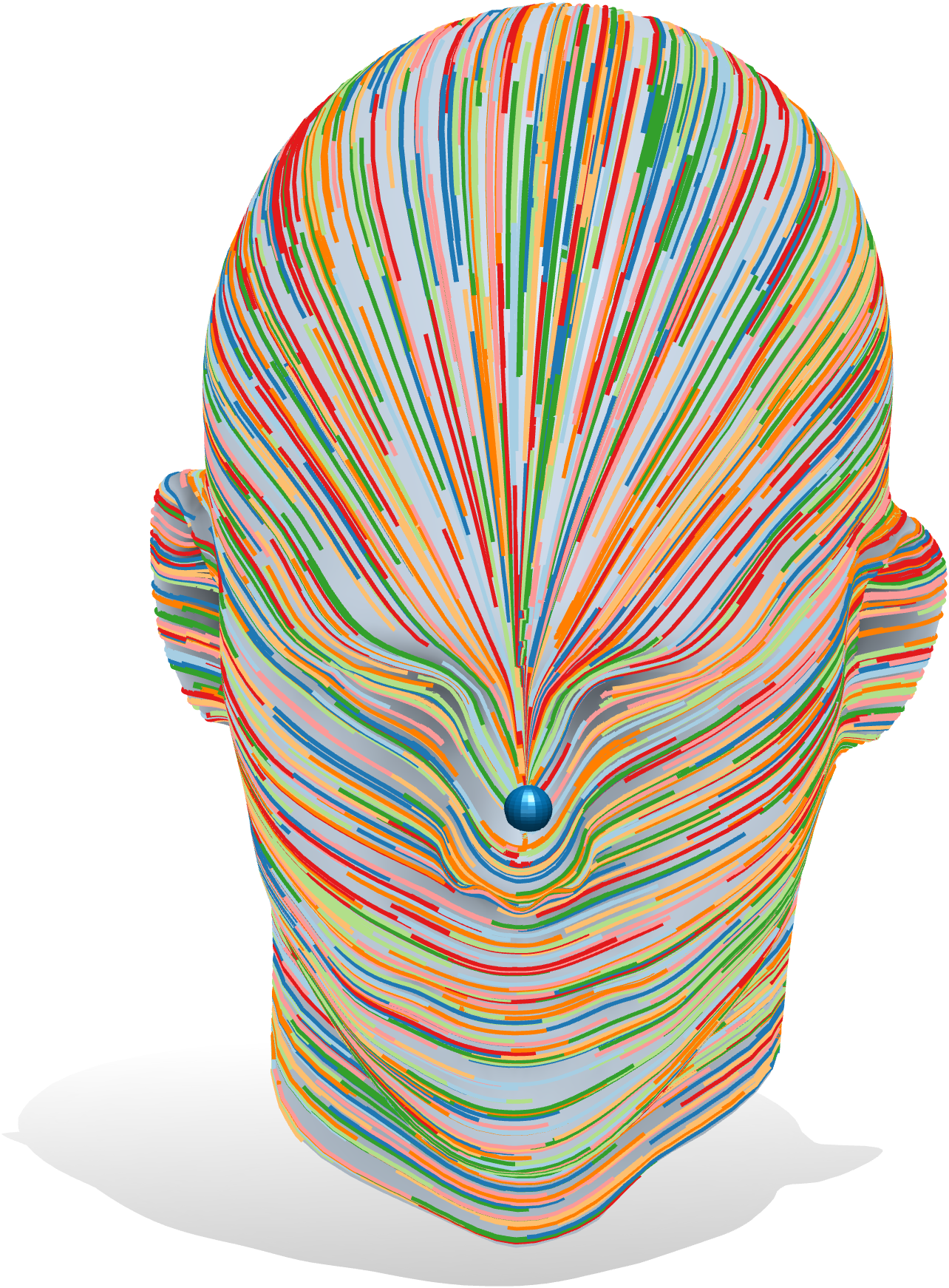} \\
Unmasked & Soft Mask & Masked
\end{tabulary}
\caption{Restricting singularities to a subdomain (shown in red) is as simple as adding a simple linear constraint in our method. The resulting constrained fields are shown at right. Soft constraints on singularities can be implemented by promoting the regularization parameter $\lambda$ to a spatially-varying scalar field (bottom row).}
\label{minsec:fig:hand-mask}
\end{figure}

\paragraph{Comparison.} In \Cref{minsec:fig:faces-cross}, we compare minimal section relaxation to prior methods for computing directional fields.
\begin{itemize}
\item \citet{knoppel_globally_2013} employ a convex relaxation relaxing the unit norm constraint. We test their method as implemented by the \textsc{geometry central} library \cite{geometrycentral}, as well as a version using our own covariant operators.
\item The \textsc{mbo} diffusion-projection method proposed by \citet{viertel_approach_2019} minimizes the nonconvex Ginzburg-Landau functional. Its diffusion time parameter $\tau$ controls the length scale of smoothing, reducing to the method of \cite{knoppel_globally_2013} as $\tau \to \infty$. \citet{viertel_approach_2019} recommend setting $\tau$ relative to the (connection) Laplacian spectrum; we use $\tau \in \{0.1, 0.01\}\lambda_2(\dscm{L}_1)$.
\end{itemize}
Compared to both methods, minimal section relaxation yields fewer singular points and more-consistent fields independent of small variations in the background mesh. Unlike the non-convex \textsc{mbo} method, it is independent of initialization.
\begin{figure*}
\centering
\newcommand{\imgwidth}{0.145\textwidth}
\begin{tblr}{width=\textwidth,colspec={X[l,m]X[c,m]X[c,m]X[c,m]X[c,m]X[c,m]}}
&%
\includegraphics[width=\imgwidth,align=c]{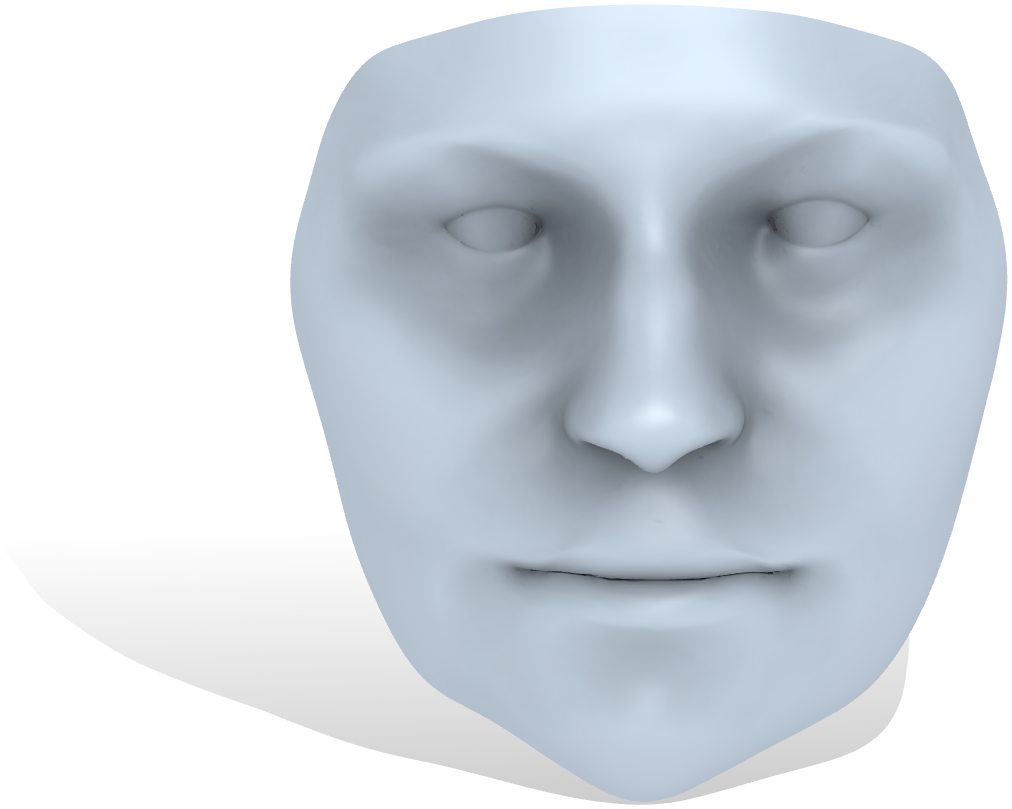} &
\includegraphics[width=\imgwidth,align=c]{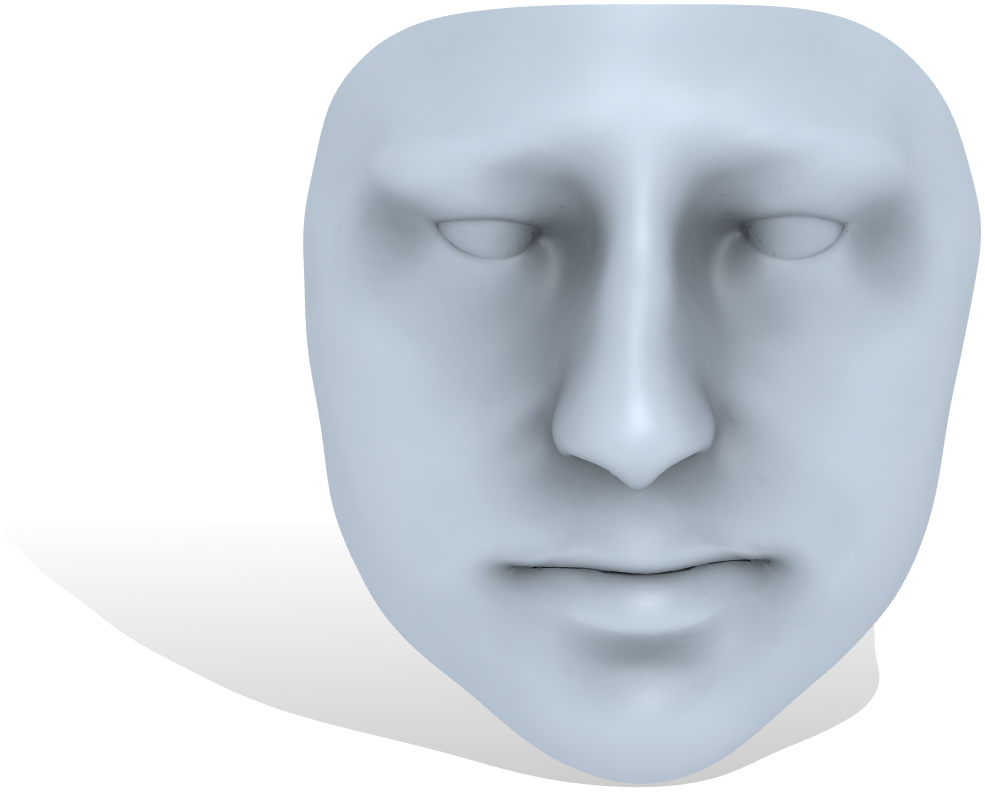} &
\includegraphics[width=\imgwidth,align=c]{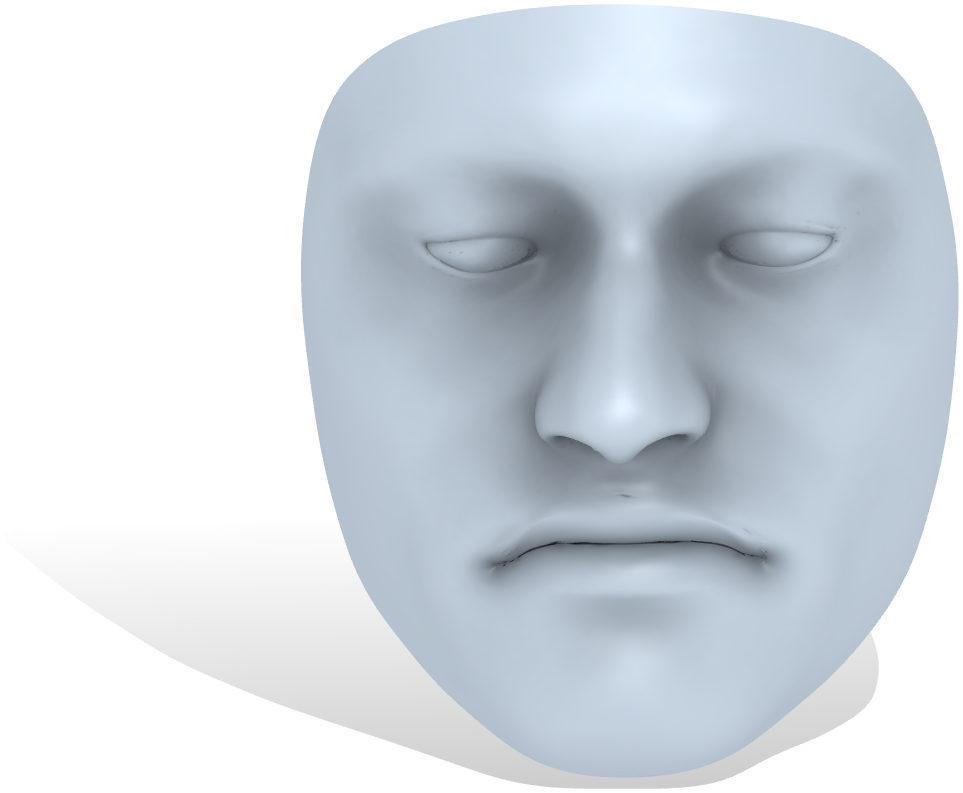} &
\includegraphics[width=\imgwidth,align=c]{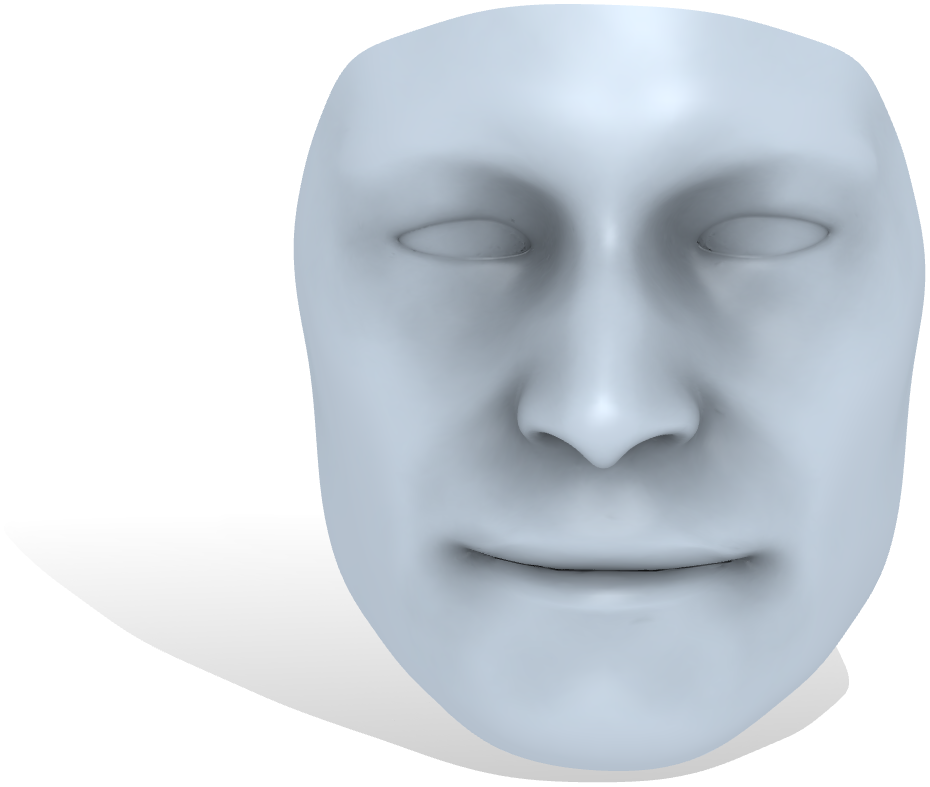} &
\includegraphics[width=\imgwidth,align=c]{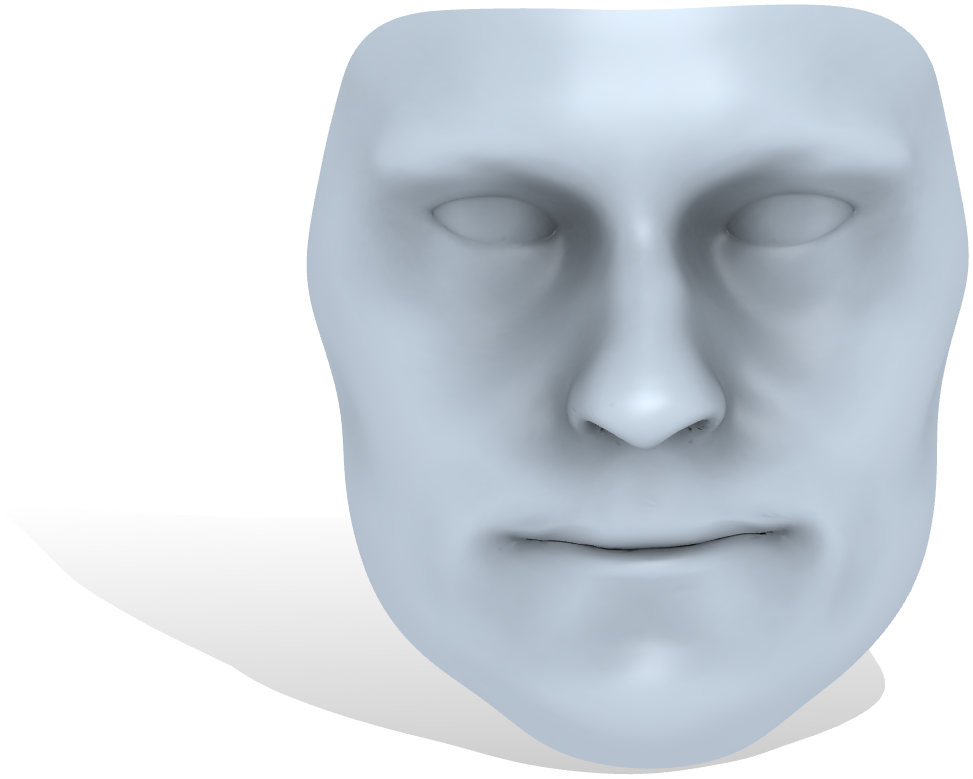} \\
{\small Ours ($\lambda = r = 0.1$)} &
\includegraphics[width=\imgwidth,align=c]{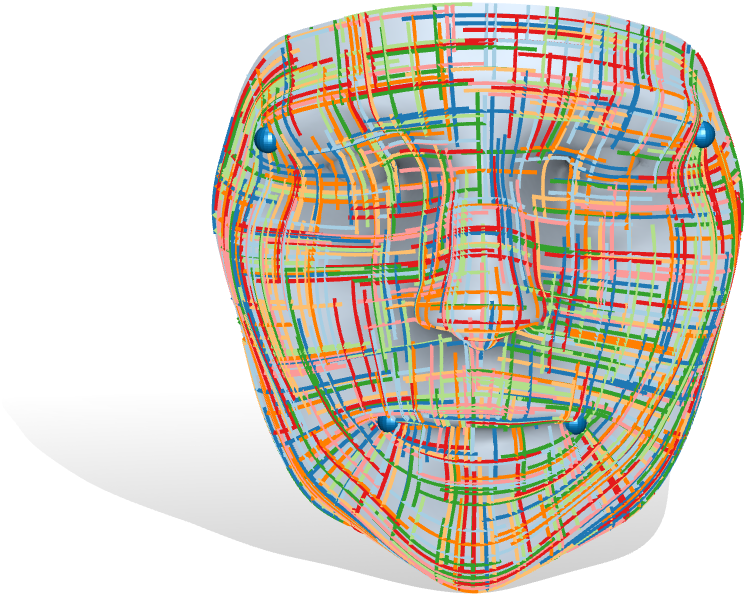} &
\includegraphics[width=\imgwidth,align=c]{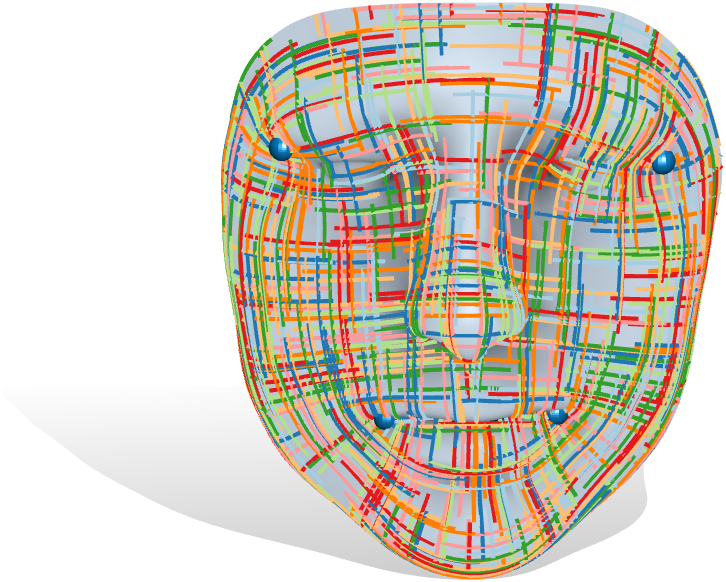} &
\includegraphics[width=\imgwidth,align=c]{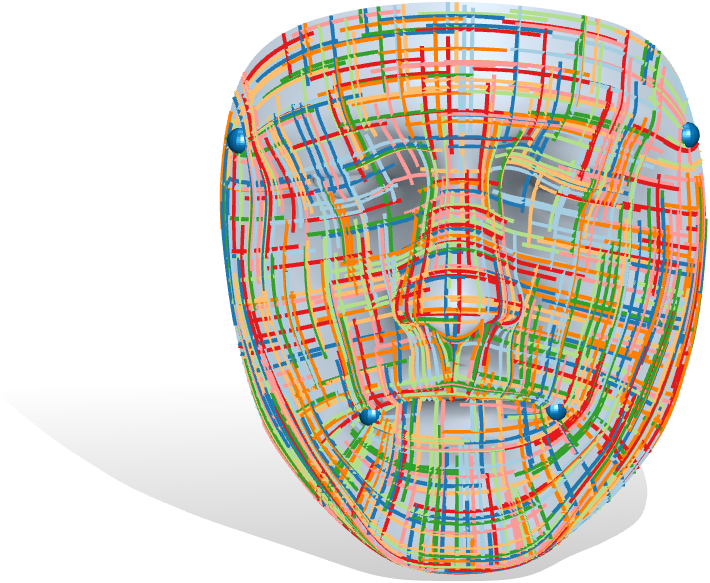} &
\includegraphics[width=\imgwidth,align=c]{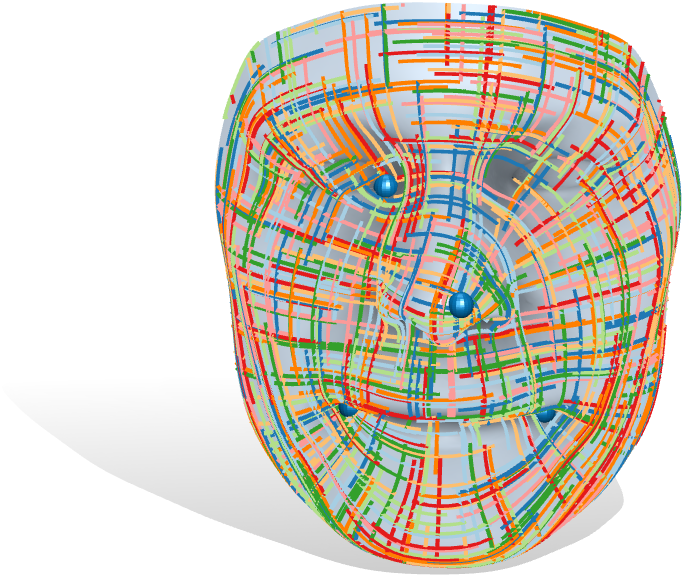} &
\includegraphics[width=\imgwidth,align=c]{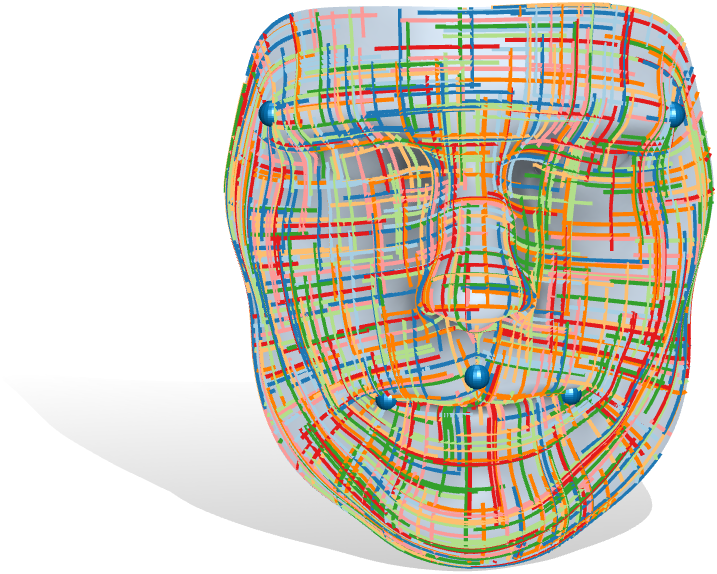} \\
{\small \textsc{mbo} ($\tau = 0.1 \lambda_2$)} &
\includegraphics[width=\imgwidth,align=c]{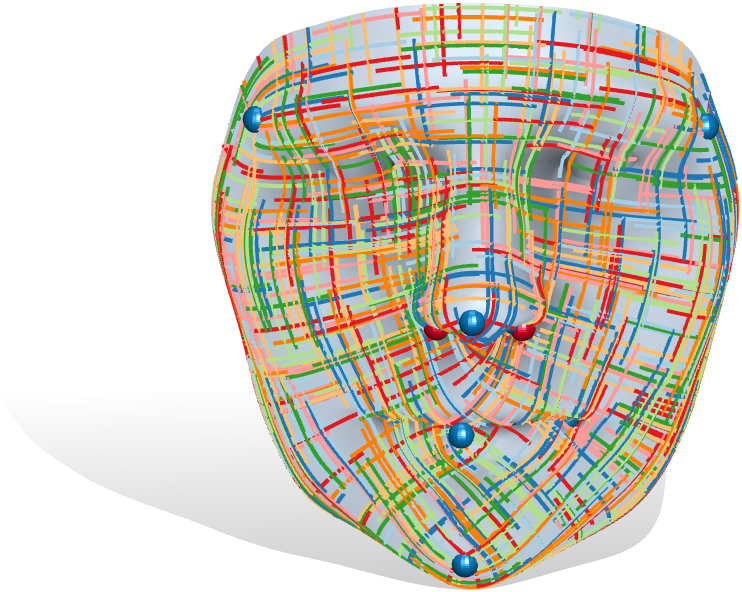} &
\includegraphics[width=\imgwidth,align=c]{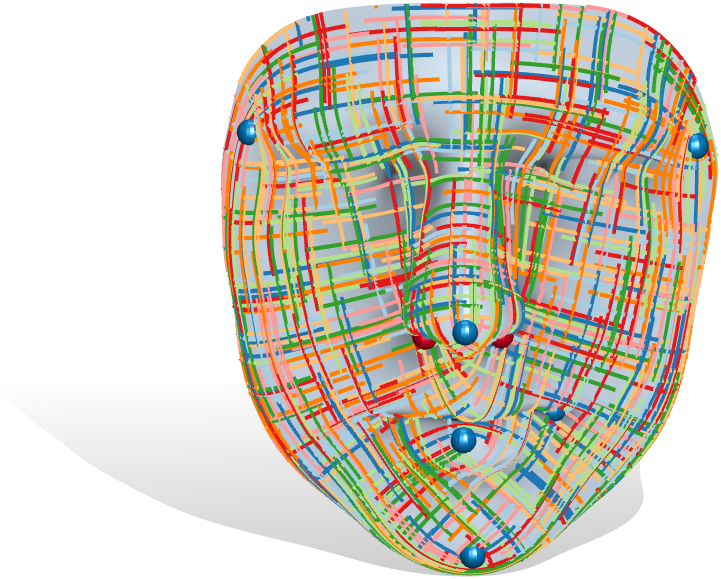} &
\includegraphics[width=\imgwidth,align=c]{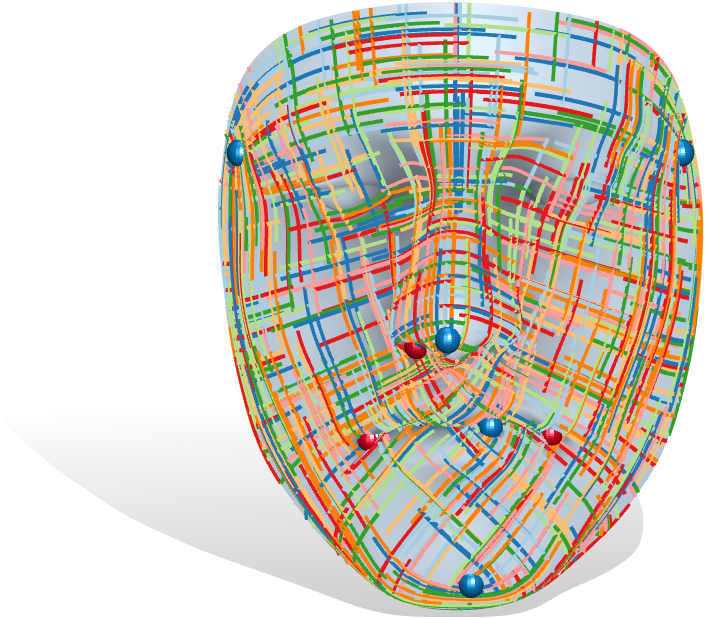} &
\includegraphics[width=\imgwidth,align=c]{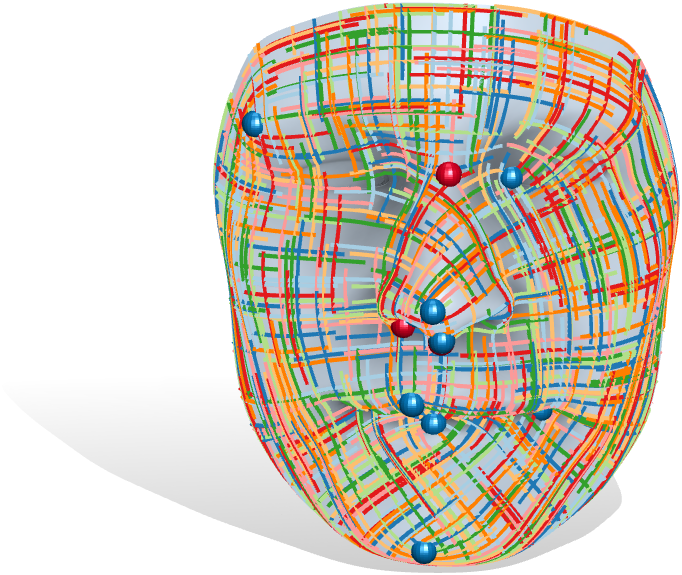} &
\includegraphics[width=\imgwidth,align=c]{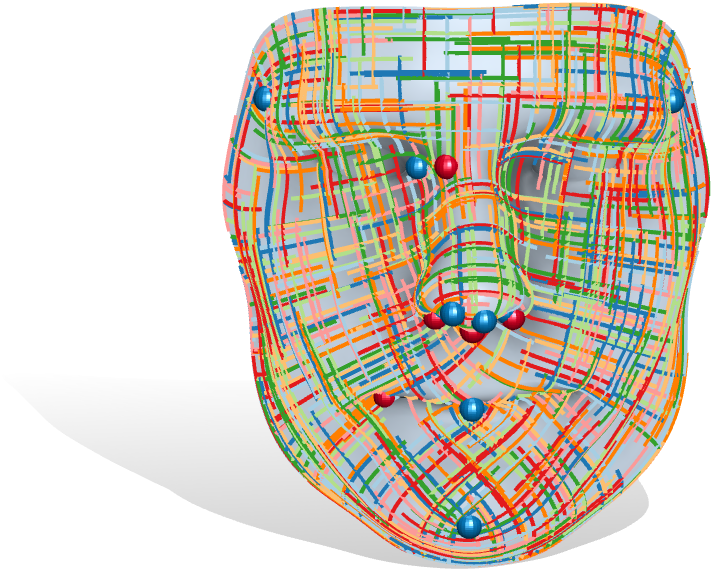} \\
{\small \textsc{mbo} ($\tau = 0.01 \lambda_2$)} &
\includegraphics[width=\imgwidth,align=c]{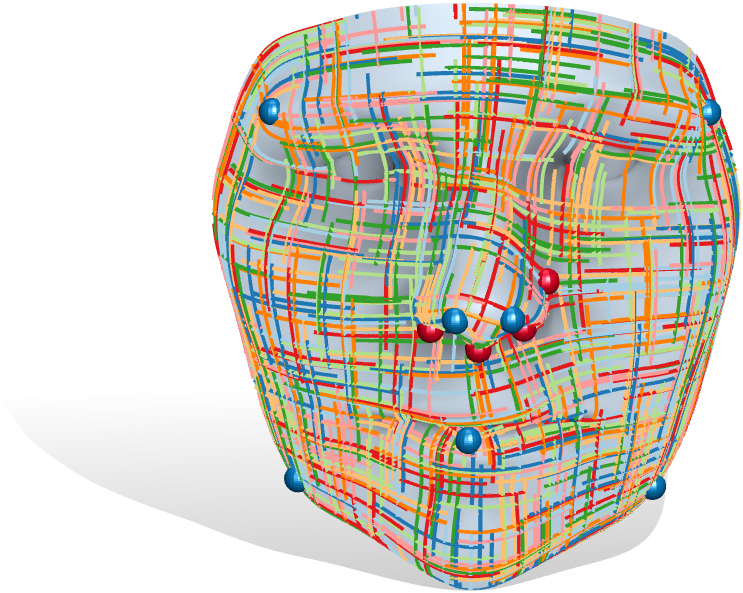} &
\includegraphics[width=\imgwidth,align=c]{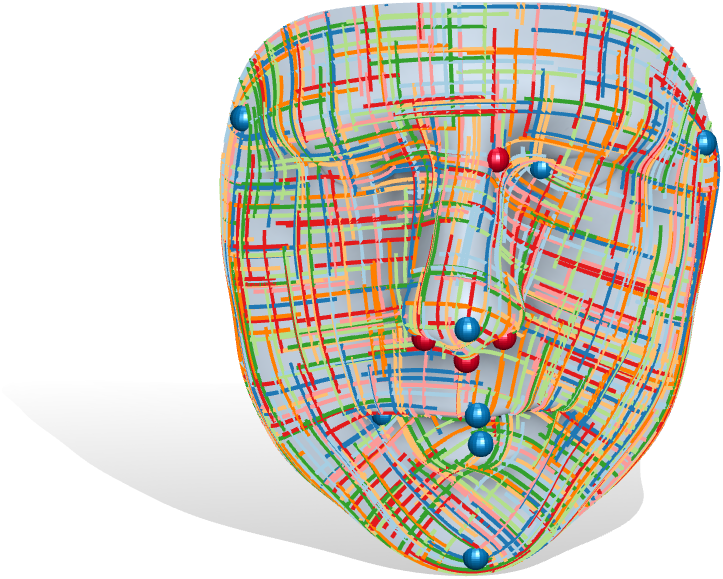} &
\includegraphics[width=\imgwidth,align=c]{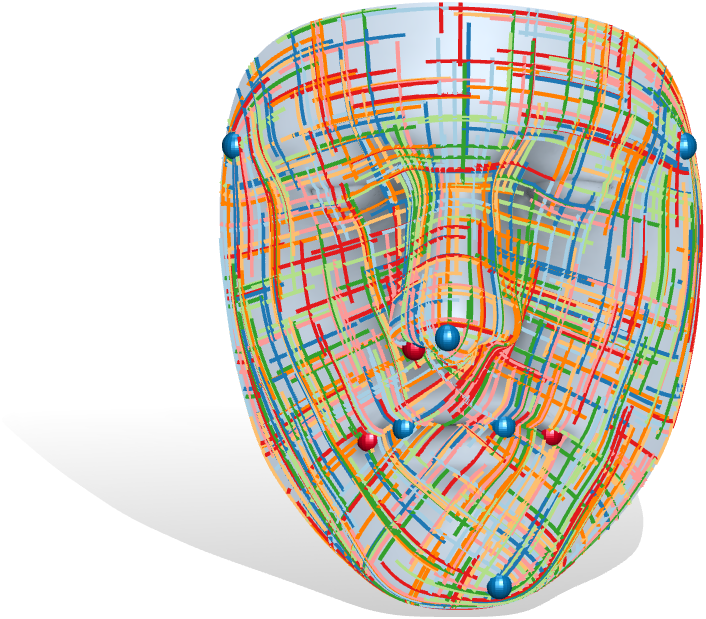} &
\includegraphics[width=\imgwidth,align=c]{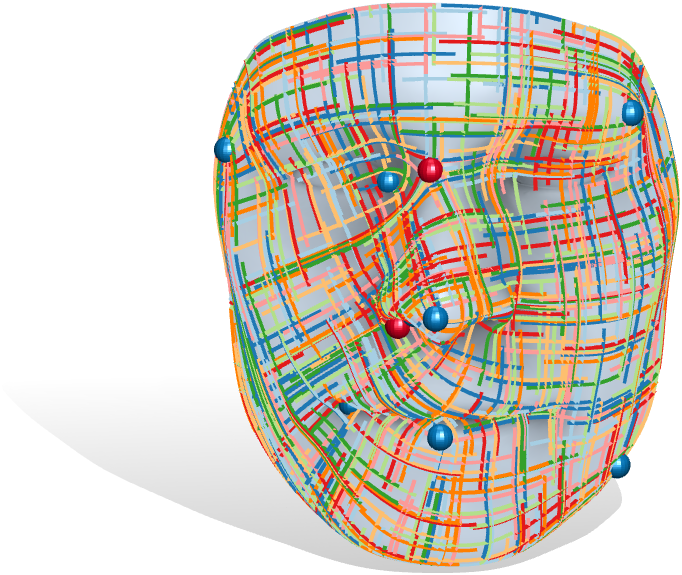} &
\includegraphics[width=\imgwidth,align=c]{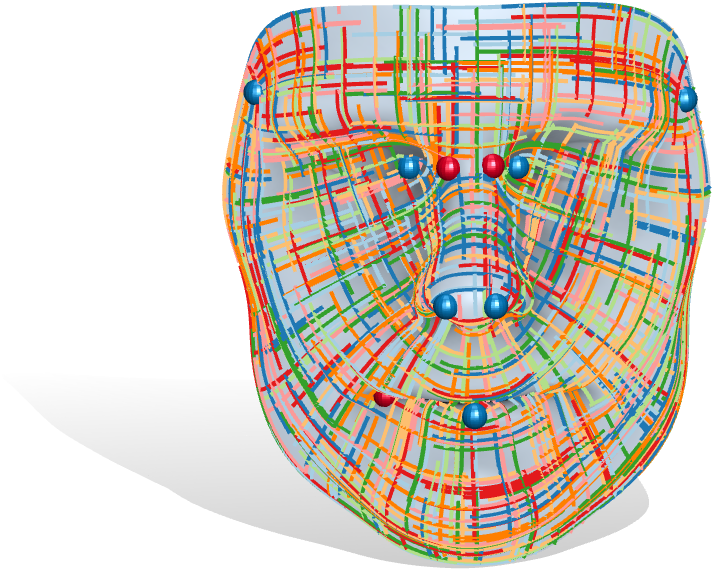} \\
{\small \cite{knoppel_globally_2013}\\(\textsc{geometry central})} &
\includegraphics[width=\imgwidth,align=c]{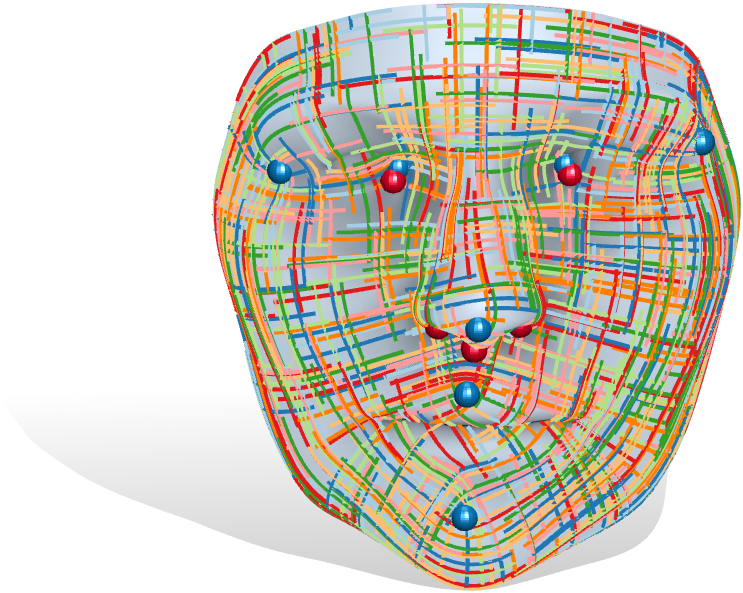} &
\includegraphics[width=\imgwidth,align=c]{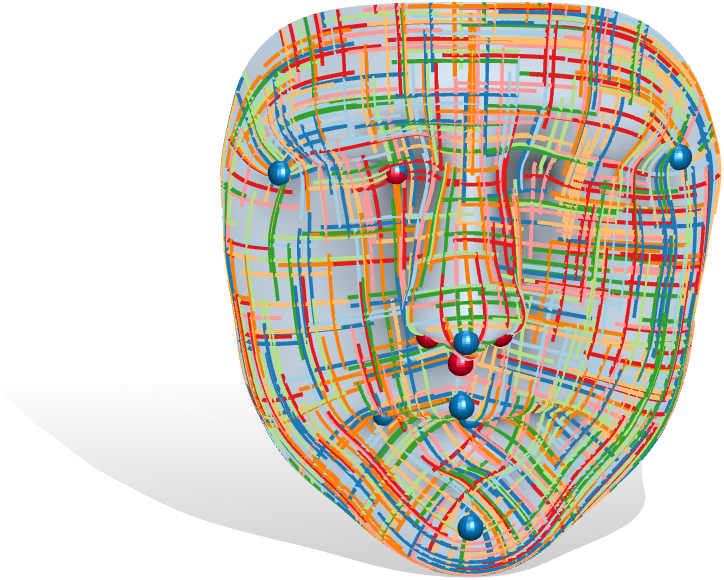} &
\includegraphics[width=\imgwidth,align=c]{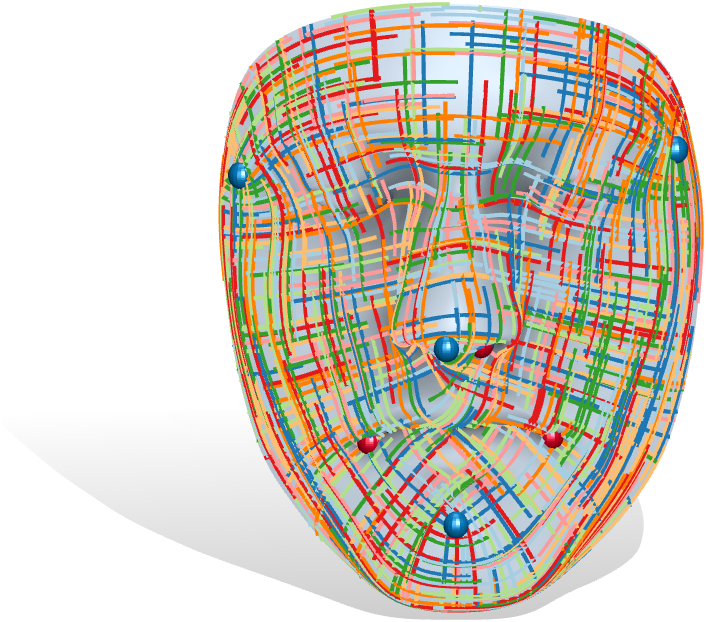} &
\includegraphics[width=\imgwidth,align=c]{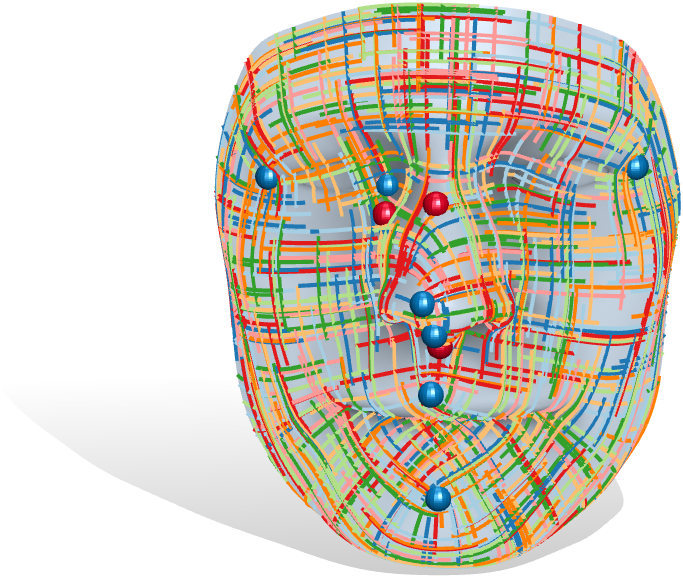} &
\includegraphics[width=\imgwidth,align=c]{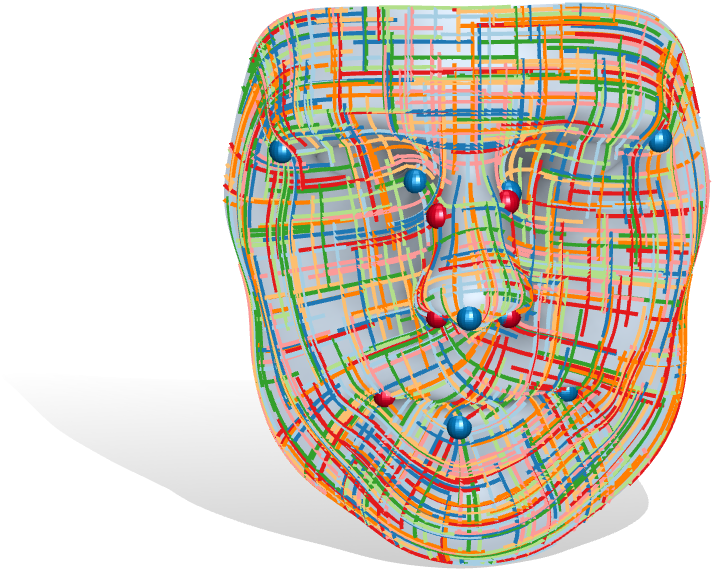} \\
{\small \cite{knoppel_globally_2013}\\(our operators)} &
\includegraphics[width=\imgwidth,align=c]{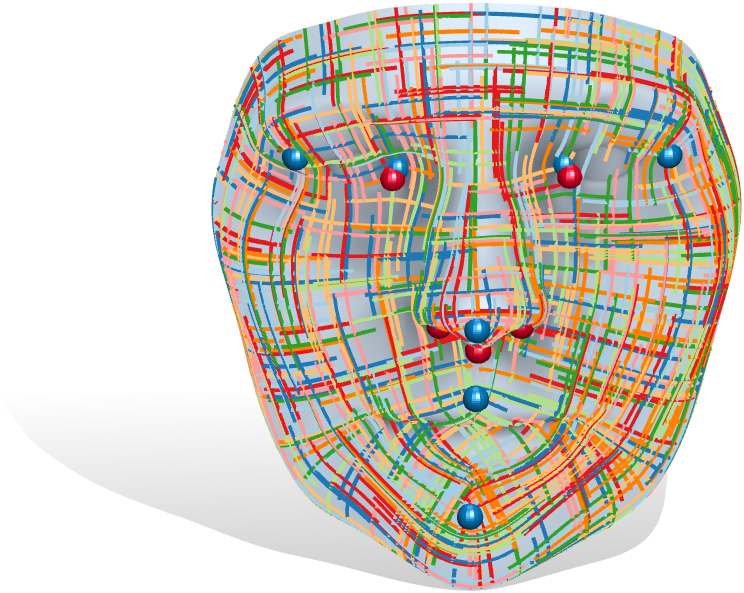} &
\includegraphics[width=\imgwidth,align=c]{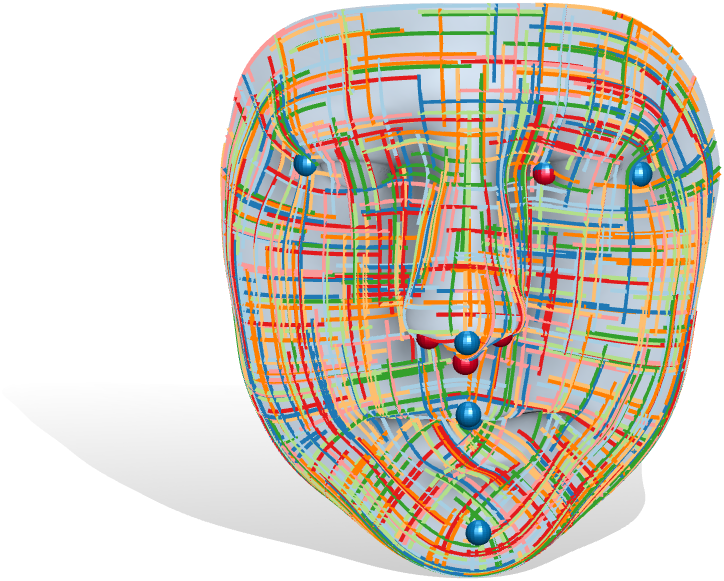} &
\includegraphics[width=\imgwidth,align=c]{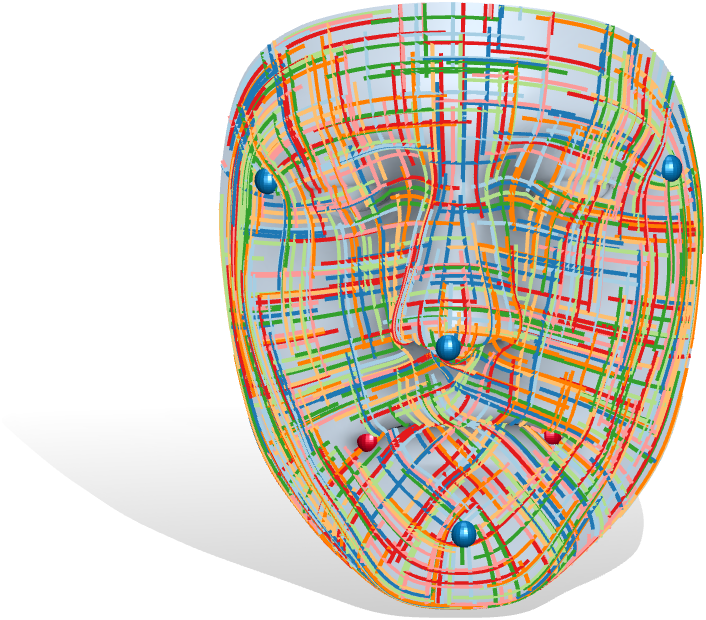} &
\includegraphics[width=\imgwidth,align=c]{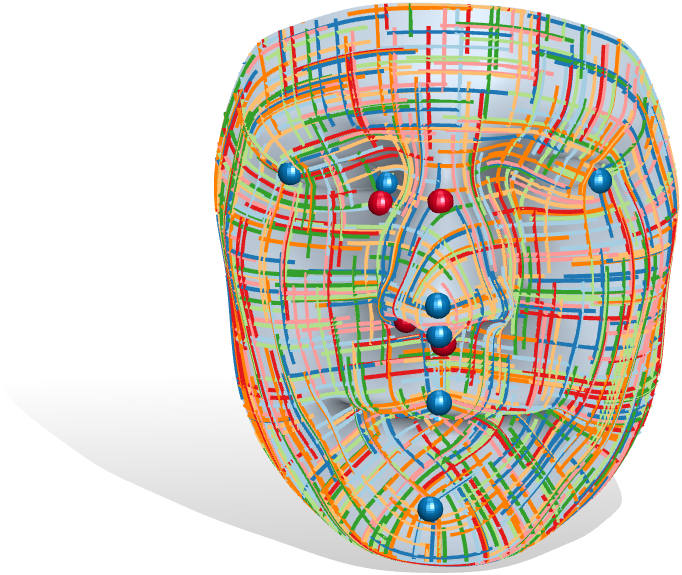} &
\includegraphics[width=\imgwidth,align=c]{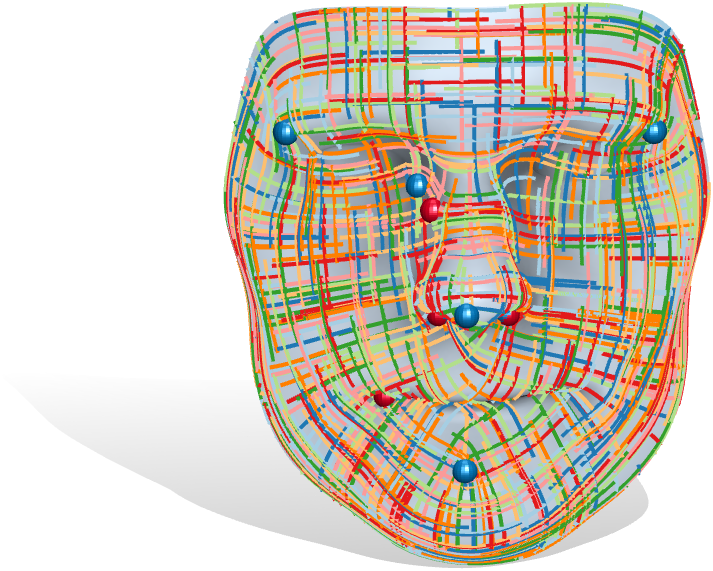} \\
{\small \cite{knoppel_globally_2013}\\(\textsc{geometry central})} &
\includegraphics[width=\imgwidth,align=c]{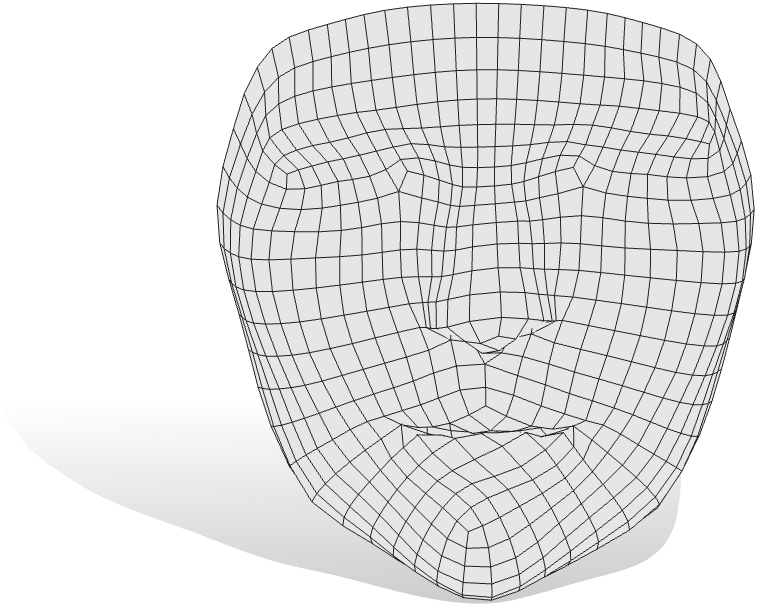} &
\includegraphics[width=\imgwidth,align=c]{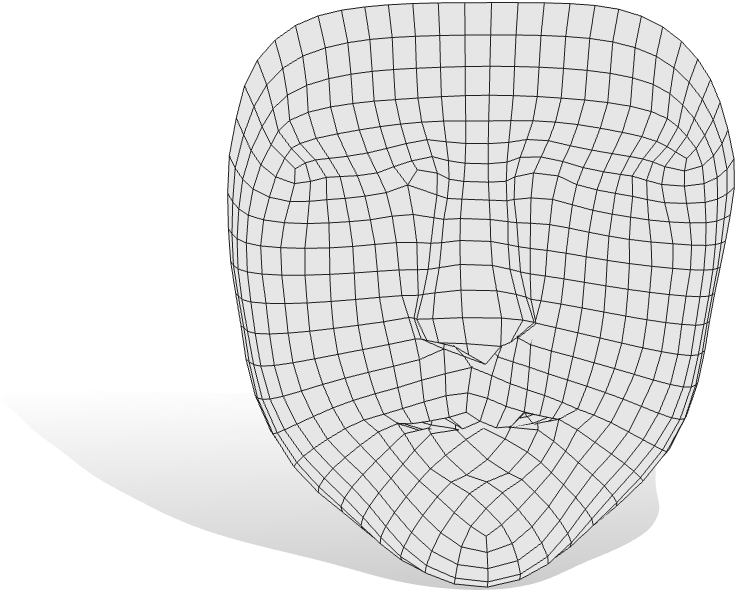} &
\includegraphics[width=\imgwidth,align=c]{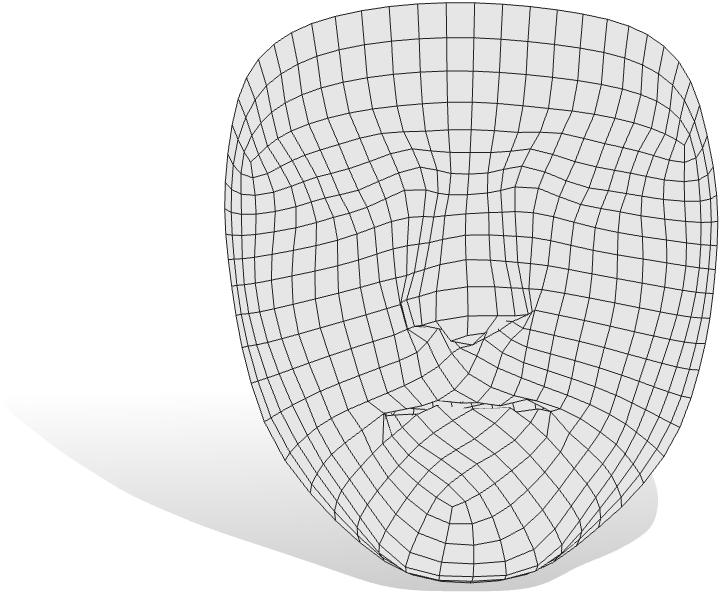} &
\includegraphics[width=\imgwidth,align=c]{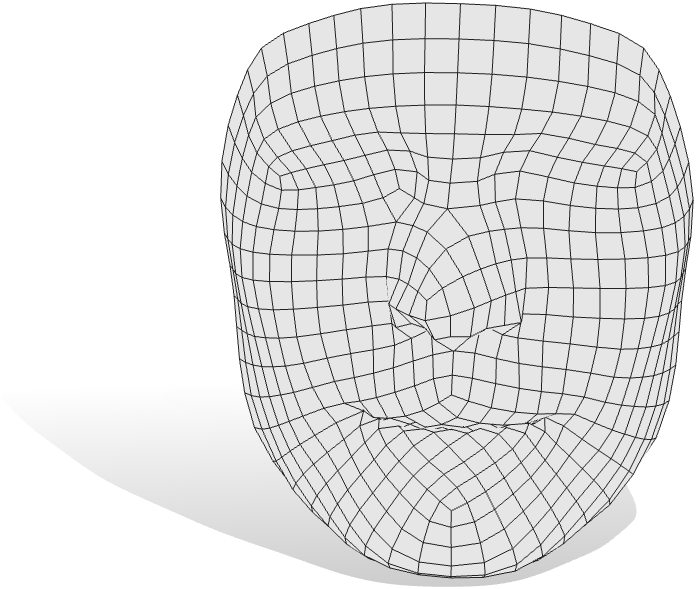} &
\includegraphics[width=\imgwidth,align=c]{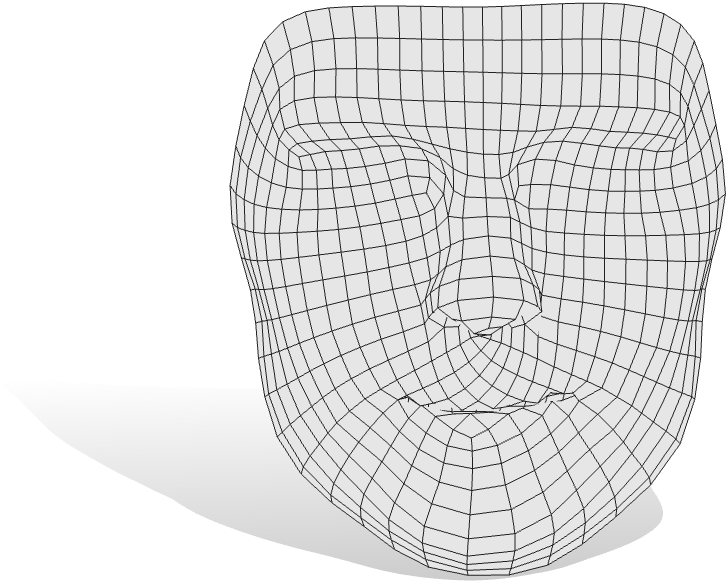} \\
{\small Ours} &
\includegraphics[width=\imgwidth,align=c]{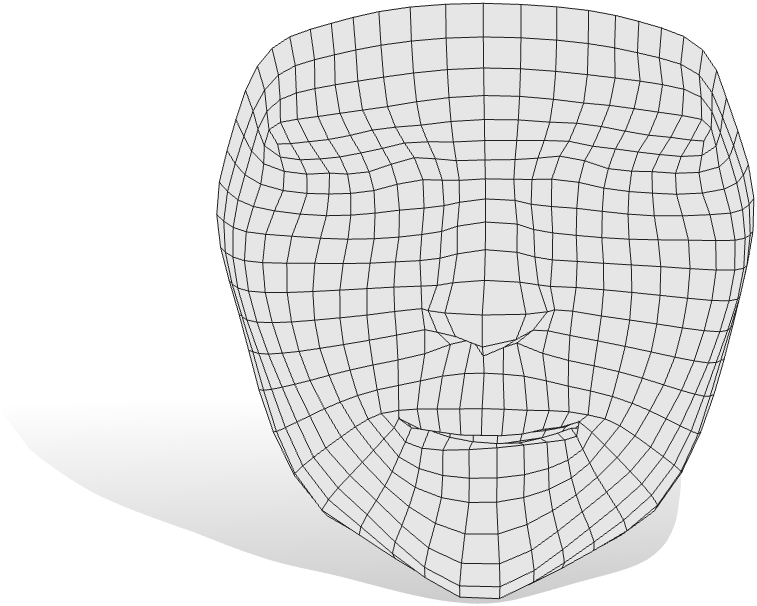} &
\includegraphics[width=\imgwidth,align=c]{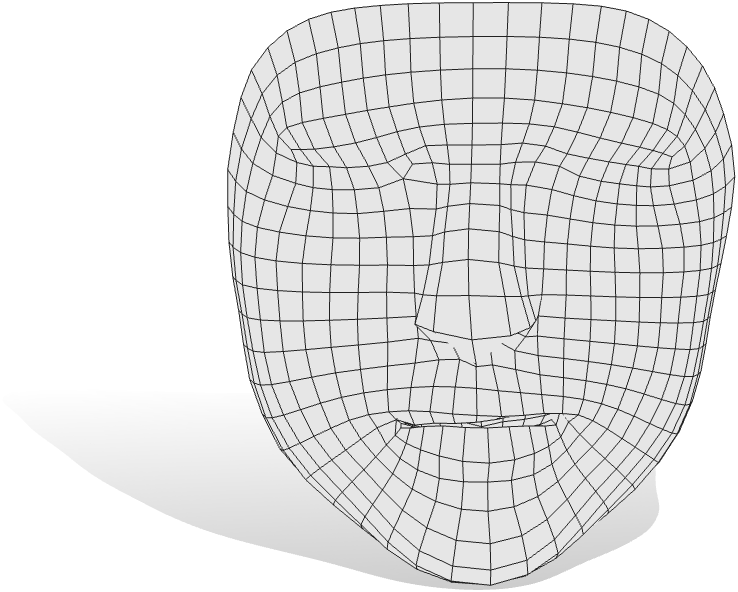} &
\includegraphics[width=\imgwidth,align=c]{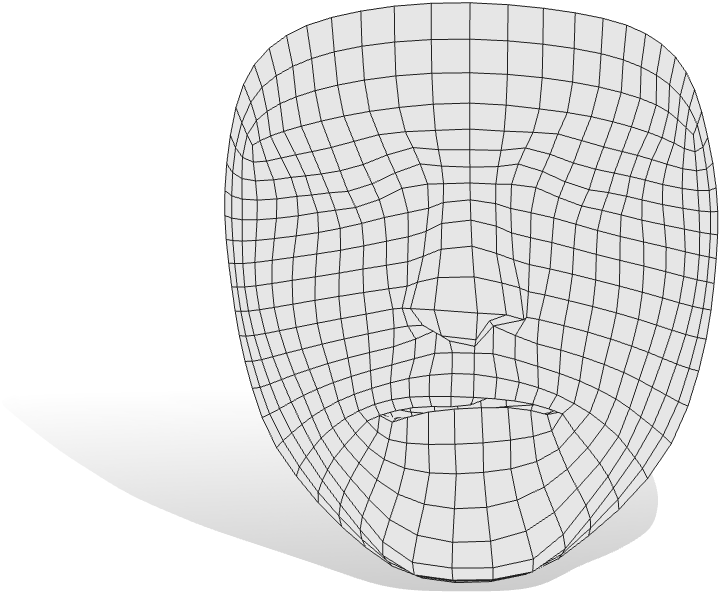} &
\includegraphics[width=\imgwidth,align=c]{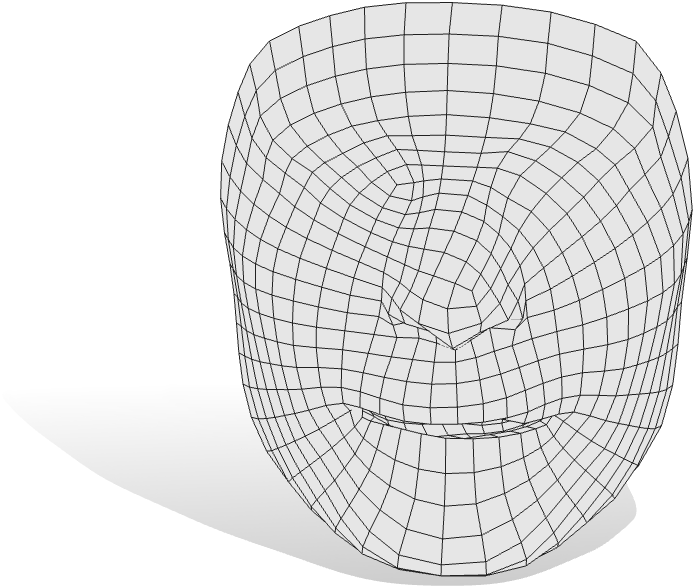} &
\includegraphics[width=\imgwidth,align=c]{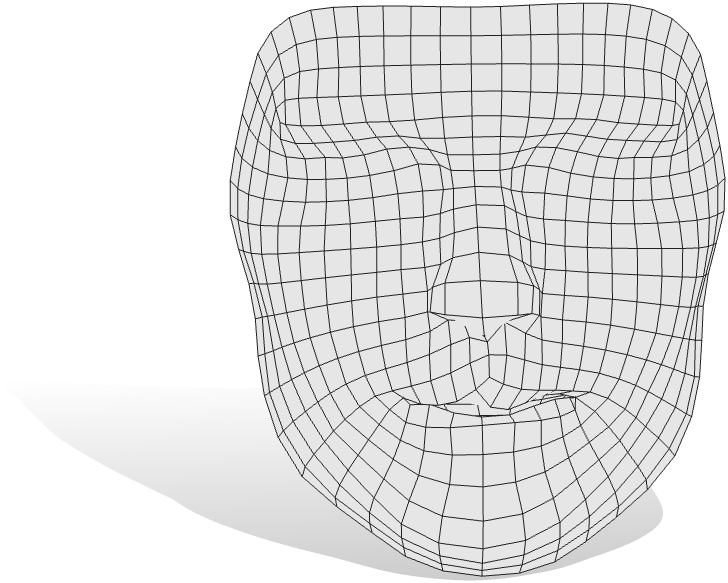}
\end{tblr}
\caption{Compared to the work of \citet{knoppel_globally_2013} and the \textsc{mbo} method of \citet{viertel_approach_2019}, minimal section relaxation yields fewer singularities and more consistent results independent of initialization or fine details of the underlying surface.}
\label{minsec:fig:faces-cross}
\end{figure*}

\paragraph{Convex Relaxation.} In the ideal case, the optimal solution to our convex relaxation is an integral current, and thus yields a globally minimal section. Due to discretization, this will never occur exactly, and the resulting current must be rounded to extract a directional field (see \Cref{minsec:subsec:extraction}). In practice, convergence toward an integral current is sufficient to extract a high-quality direction field (\Cref{minsec:fig:bdry-opening}). In some cases, however, the optimization converges to a superposition of integral solutions (\Cref{minsec:fig:superposition}). While our field extraction procedure still outputs a reasonable solution in such cases, we leave to future work the problem of bounding the rounding error.
\begin{figure}
\newcommand{\imgwidth}{0.3\columnwidth}
\begin{tabulary}{0.95\columnwidth}{@{}CCC@{}}
	\includegraphics[width=\imgwidth]{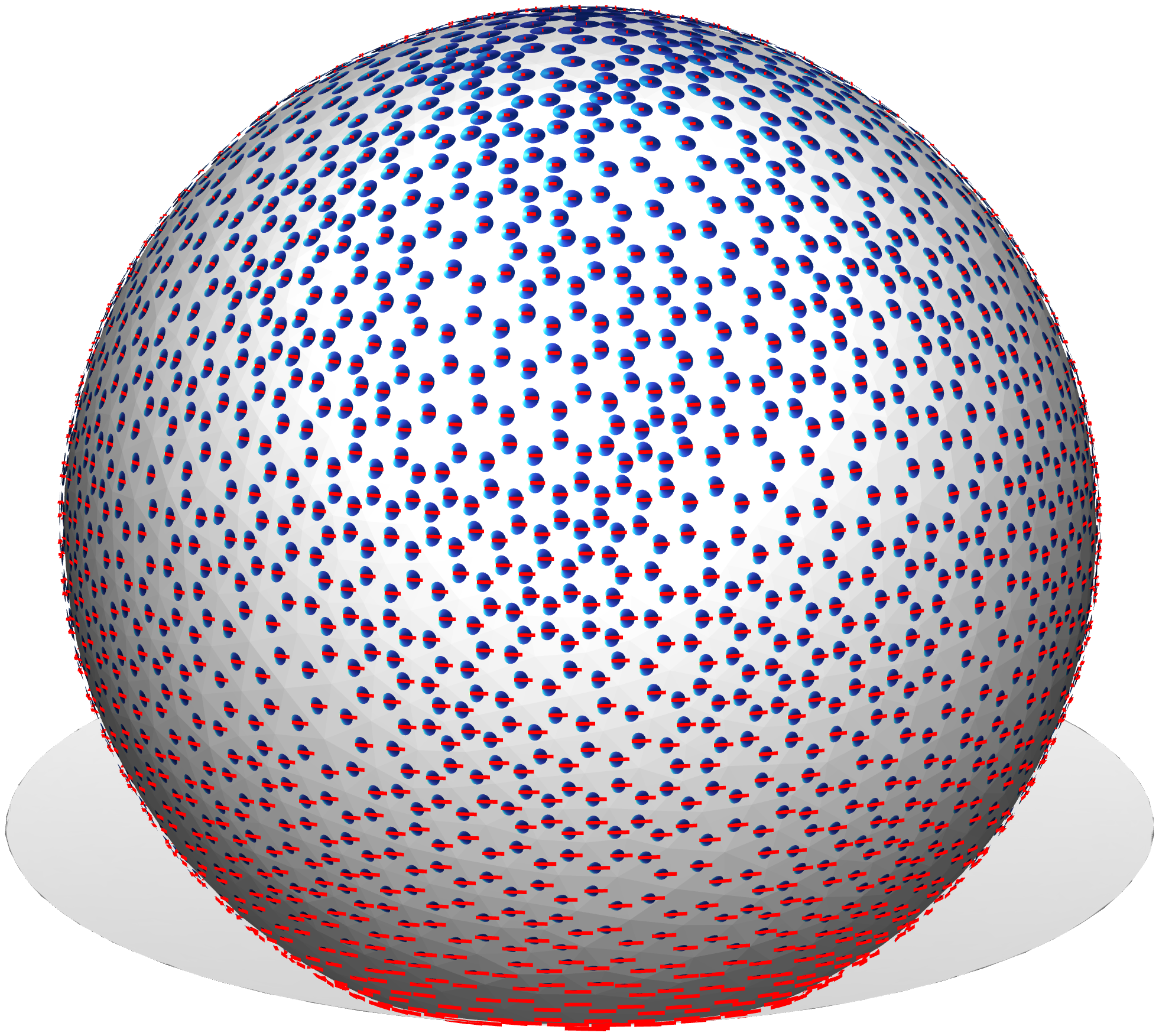} &
	\includegraphics[width=\imgwidth]{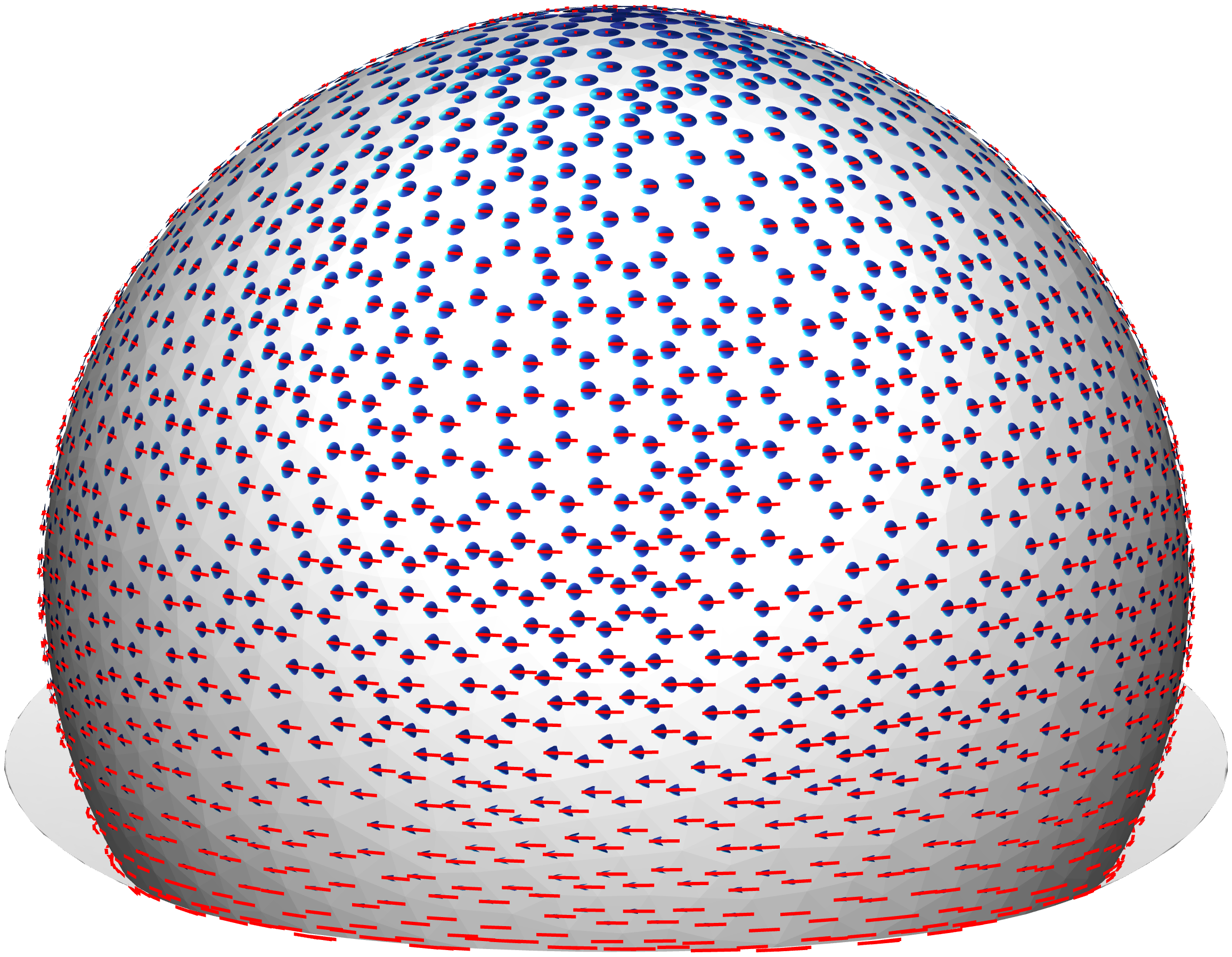} &
	\includegraphics[width=\imgwidth]{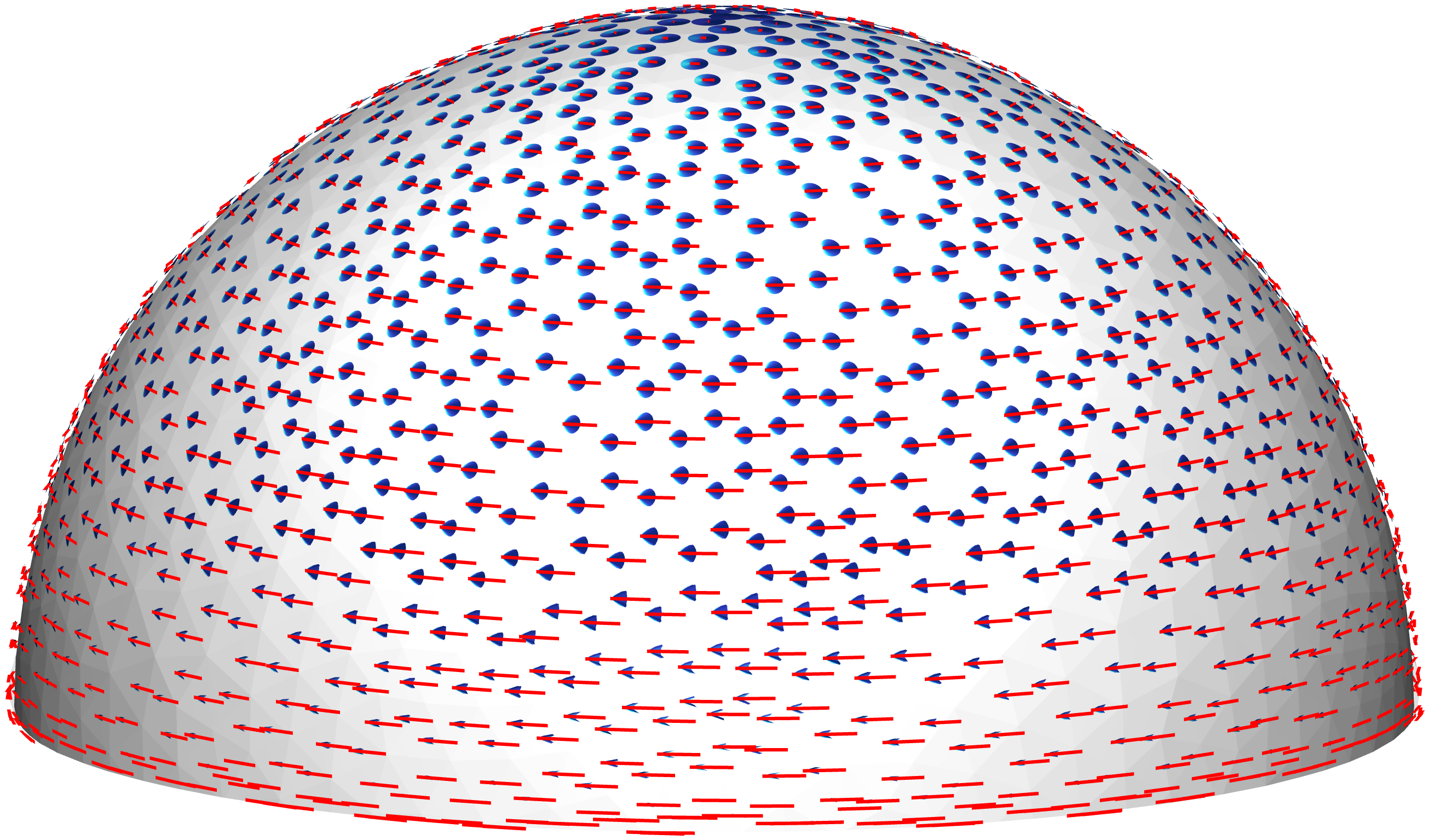} \\
	\includegraphics[width=\imgwidth]{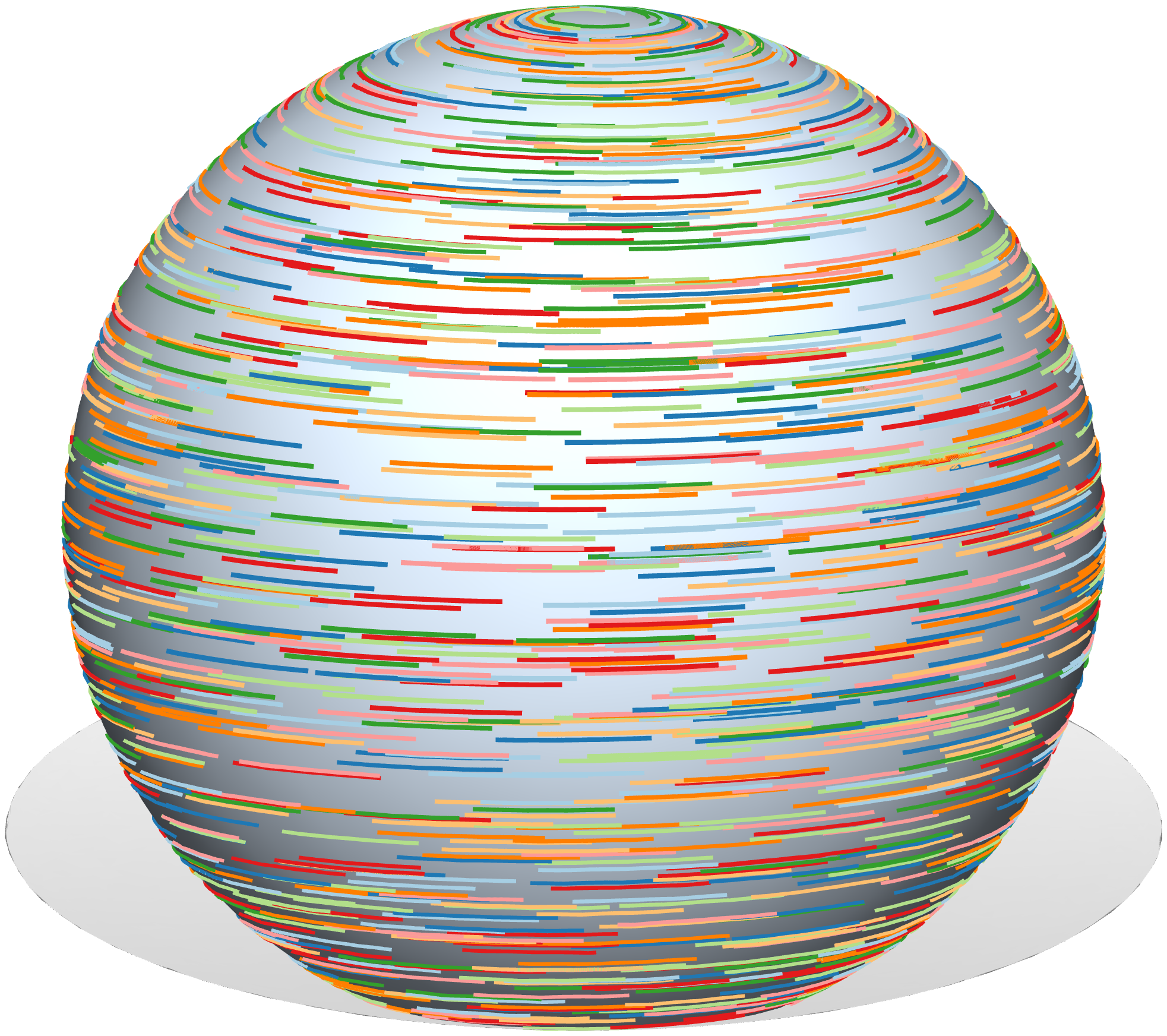} &
	\includegraphics[width=\imgwidth]{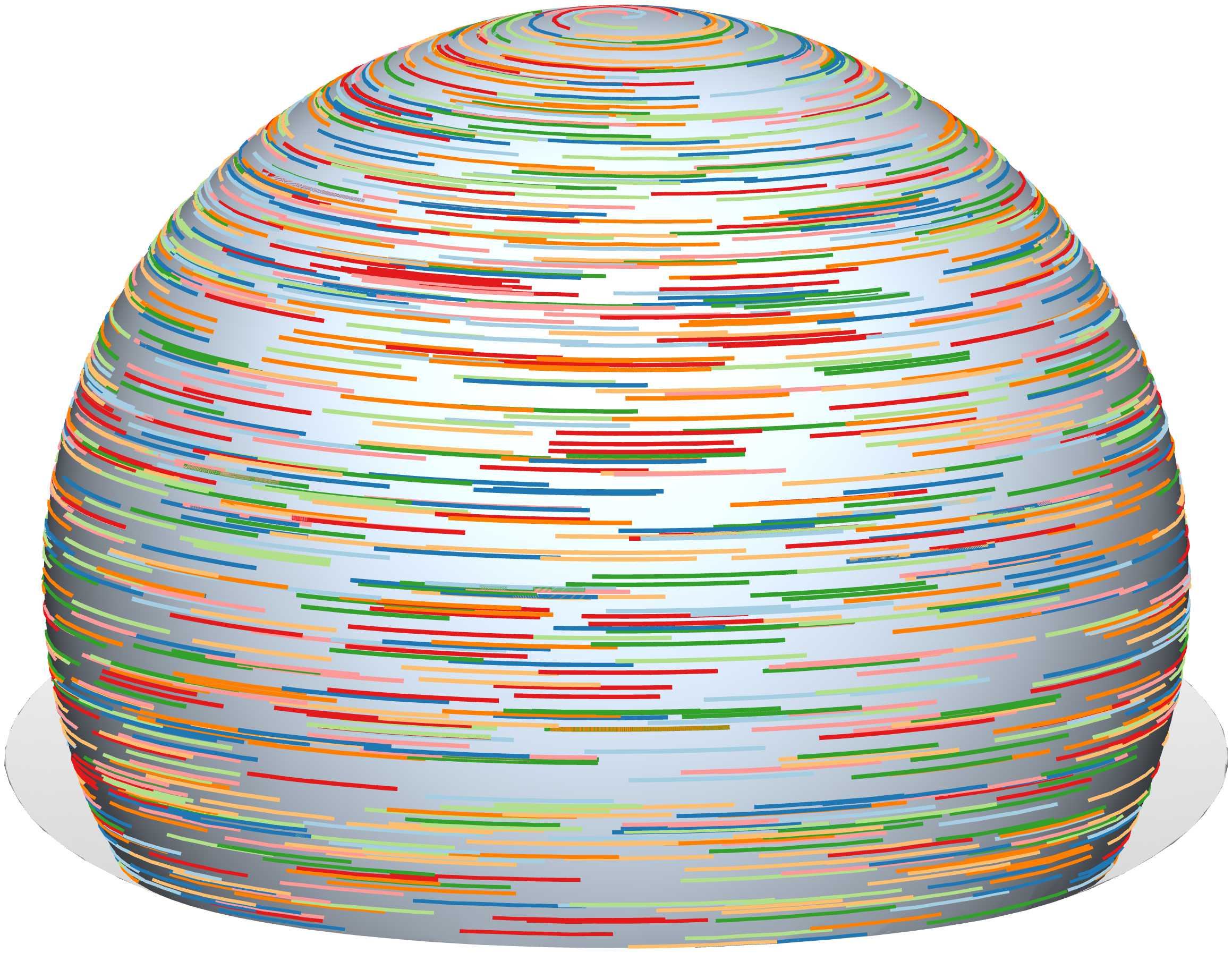} &
	\includegraphics[width=\imgwidth]{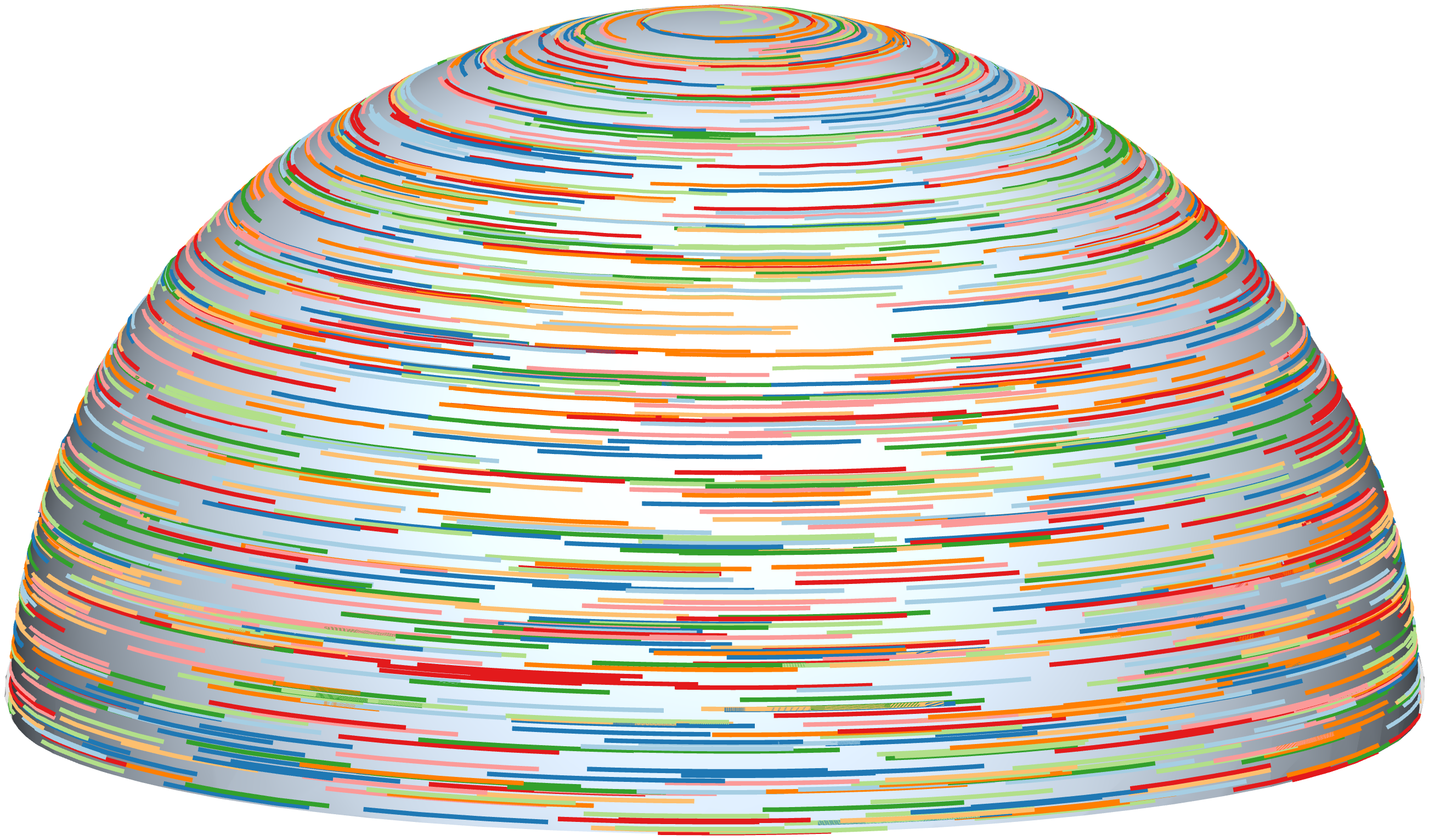}
\end{tabulary}
\caption{The current can become diffuse far away from the boundary (top). However, the extracted field is still well-defined (bottom).}
\label{minsec:fig:bdry-opening}
\end{figure}
\begin{figure}
	\newcommand{\imgwidth}{0.3\columnwidth}
	\begin{tabulary}{0.95\columnwidth}{@{}CCC@{}}
	\includegraphics[width=\imgwidth]{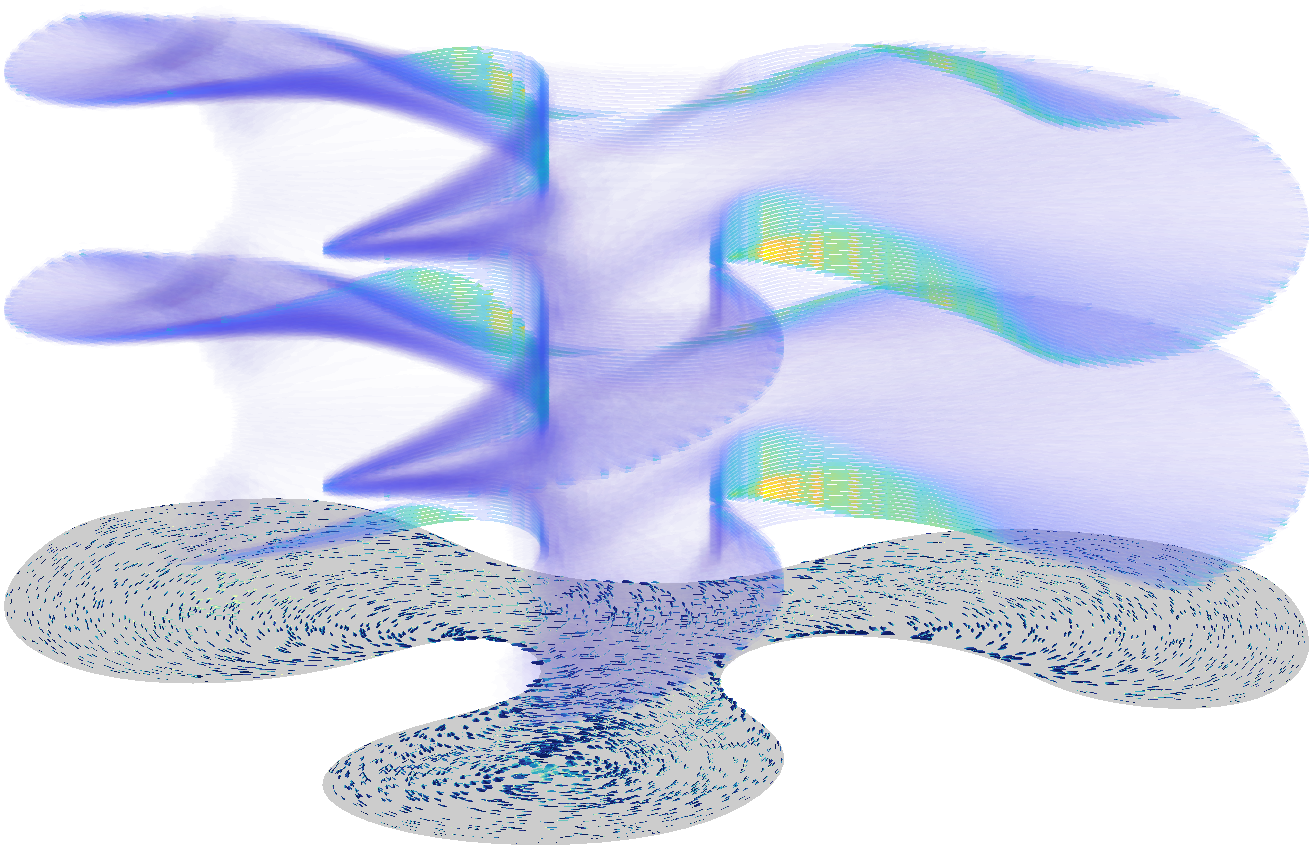} &
	\includegraphics[width=\imgwidth]{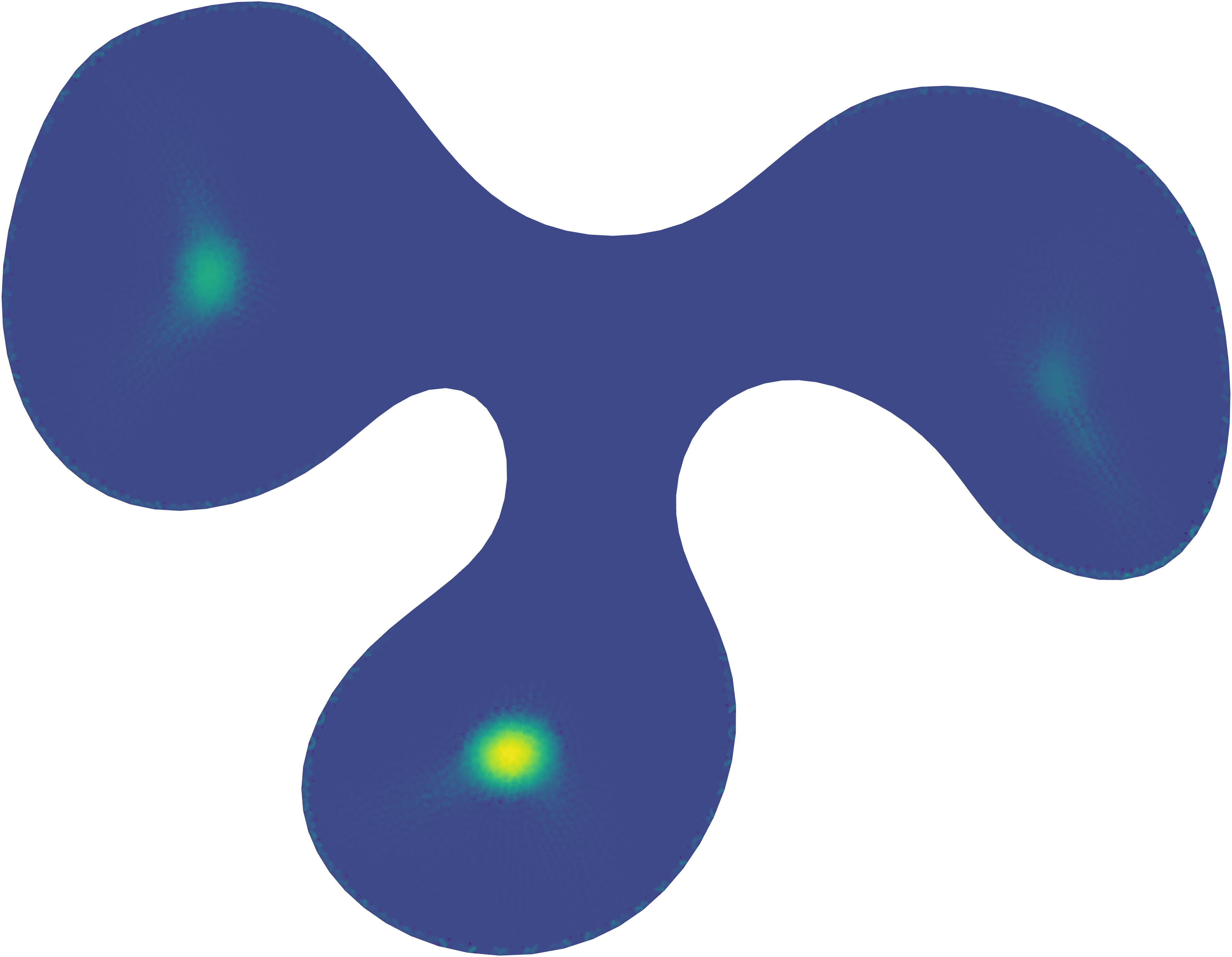} &
	\includegraphics[width=\imgwidth]{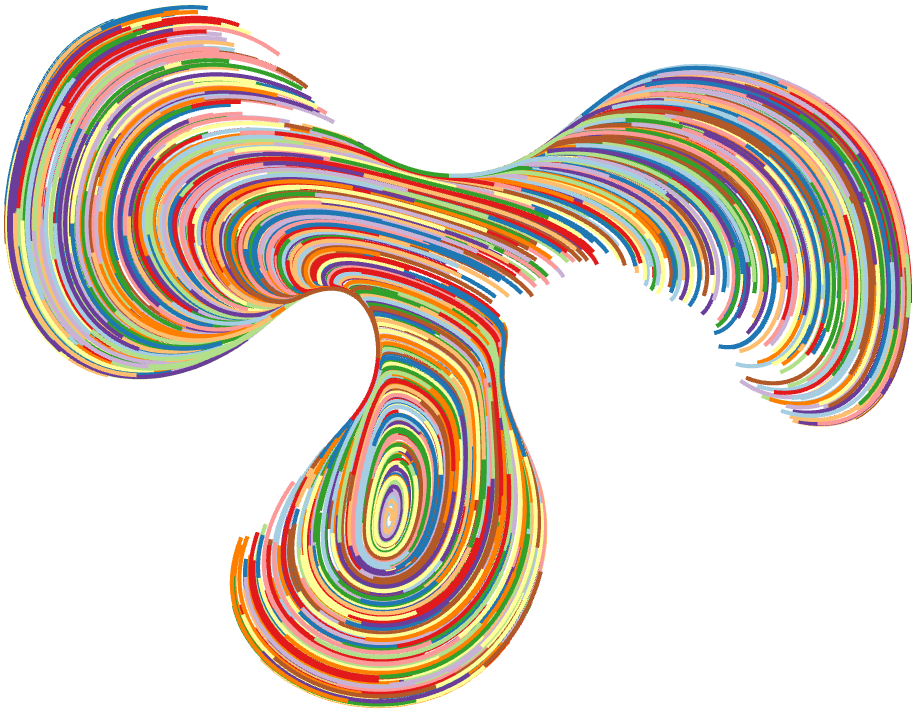}
\end{tabulary}
\caption{In some cases the current converges to a superposition of integral currents (left) with non-quantized singularities (center). However, one can still extract a directional field (right). This is a degree $1$ field computed with $\lambda = r = 1$ and base area $1$.}
\label{minsec:fig:superposition}
\end{figure}

To investigate the robustness of our method, we test it on models from the the dataset of \citet{Myles:RFAGP}. Since our method requires a surface with boundary, we preprocess closed surfaces by slicing along a plane through the centroid and taking the largest connected component. Five of the $116$ models in the dataset yielded errors during this preprocessing step, leaving us with $111$ models. The parameters $\lambda = 0.1, r = 0.2$ with $N = 64$ fiber increments were used for all models. \Cref{minsec:fig:dataset-mosaic} shows assorted results of our method and that of \citet{knoppel_globally_2013}.
\begin{figure*}
	\centering
	\newcommand{\imgwidth}{0.45\textwidth}
	\begin{tabulary}{\textwidth}{@{}C|C@{}}
	\includegraphics[width=\imgwidth]{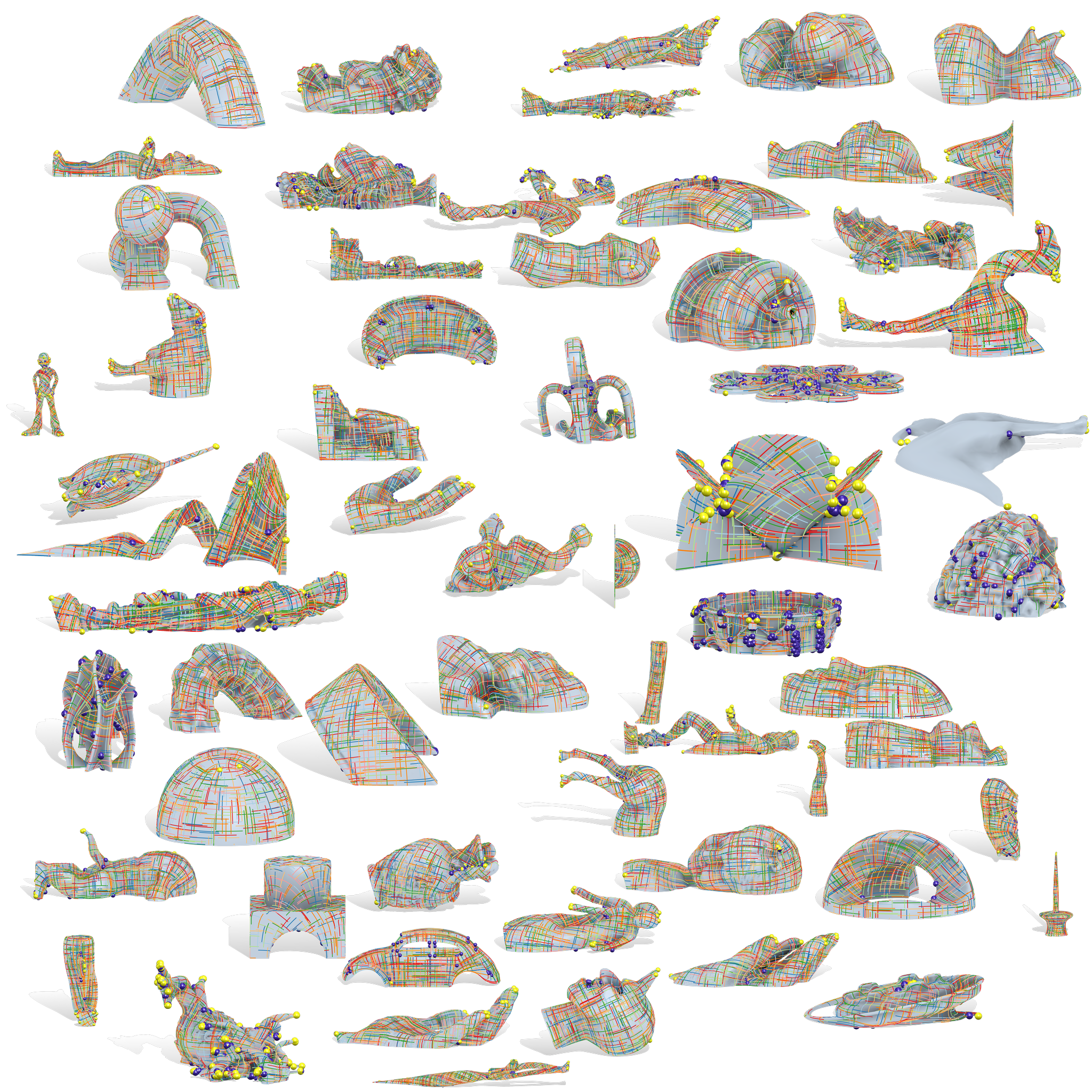} &
	\includegraphics[width=\imgwidth]{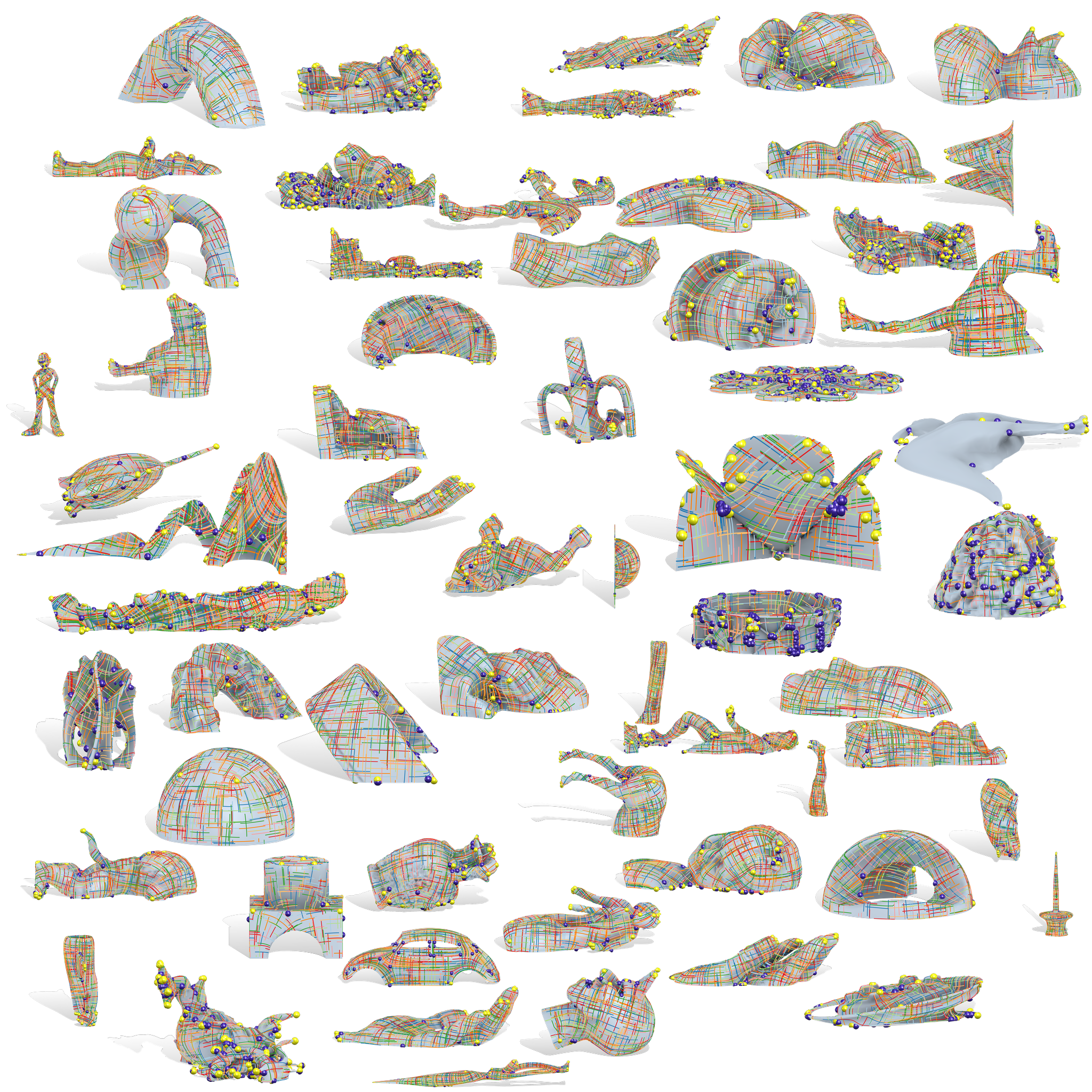} \\
	Ours & \cite{knoppel_globally_2013}
	\end{tabulary}
	\caption{Assorted results on the dataset of \citet{Myles:RFAGP}.}
	\label{minsec:fig:dataset-mosaic}
\end{figure*}

As a rough measure of exact recovery, we examine how tightly the mass of the computed current $\Sigma$ concentrates around the extracted section $\sigma$. For absolute angles $\theta \in [0, \pi]$, we measure how much of the mass of $\Sigma$ in each fiber is within a vertical offset of $\theta$ from the value of $\sigma$ in that fiber, then average over the base:
\begin{equation}
	\Mass_{\sigma,\theta}(\Sigma) = A^{-1} \sum_v A_v \sum_{x \in \pi^{-1}(v)} \{|\Sigma_x|_g : \dist_{\Sph^1}(x, \sigma_v) < \theta \}
\end{equation}
This cumulative distribution function (\textsc{cdf}) is shown in \Cref{minsec:fig:dataset-fiber-dist} for each surface in the dataset. Some currents are highly concentrated near the extracted section while others are diffusely spread out in the vertical direction. We consider a current to be meaningfully concentrated if $65\%$ of the mass is within $\pi/2$ of the extracted section. This is the case for $70$ of the $111$ models tested---a bit more than $60\%$ of the models.

The vertical \textsc{cdf} gives us a global measure of how close the relaxation is to exact recovery. To get a clearer picture of how the exactness varies over each surface, we measure the optimal transport distance $W_2$ on each fiber between $|\Sigma|$ restricted to that fiber and a Dirac $\delta$ measure at the value of $\sigma$ on that fiber:
\begin{equation}
	W_2(|\Sigma|_{\pi^{-1}(v)}, \delta_{\sigma_v}).
\end{equation}
In turn, \Cref{minsec:fig:dataset-W2-dist-cdf} shows a \textsc{cdf} of this pointwise measure of exact recovery over each base surface. While some models feature mass concentrated near the extracted section on most fibers, other models feature a significant fraction of fibers for which the mass is spread out diffusely. We suspect this is partly due to numerical decay of the high vertical frequencies away from the boundary, especially on highly curved surfaces. \Cref{minsec:fig:W2-vs-mesh} plots the mean fiber-wise $W_2$ distance against statistics of the vertex-wise Gaussian curvature $|\kappa|$. The results are closer to exact recovery on more gently curved meshes. Some of the meshes in the dataset have extremely high values of pointwise curvature---as high as $\num{1e8}$---due to nearly-degenerate triangles in curved regions. This suggests that remeshing the input surface prior to computing a minimal section may significantly improve exact recovery.
\begin{figure}
	\centering
	\pgfplotstableread[col sep=comma]{figures/dataset/fiber_dist_cdf.csv}{\datasetStat}
	\pgfplotstablegetcolsof{\datasetStat}
	\pgfmathsetmacro{\N}{\pgfplotsretval-1}
	\tikzsetnextfilename{fiber-dist-cdf}
	\begin{tikzpicture}
	\begin{axis}[
	        mark size = 0.5pt,
	        width = \columnwidth,
	        height = 0.6\columnwidth,
	        enlarge x limits = 0,
	        enlarge y limits = 0,
	        grid = none,
	        xlabel = {\footnotesize Angle Offset $\theta$},
	        ylabel = {\footnotesize Cumulative Mass Fraction $\Mass_{\sigma,\theta}(\Sigma)$},
	        xlabel near ticks,
	        ylabel near ticks,
	        every tick label/.append style = {font=\tiny}
	    ]
	    \foreach \column in {1,...,\N}{%
  			\addplot+[no markers, solid, ultra thin] table [x={dist_level},y index=\column] {\datasetStat};
		}
	\end{axis}
	\end{tikzpicture}
	\caption{The mass of $\Sigma$ is always biased toward the extracted section, though the degree of concentration varies with the base model.}
	\label{minsec:fig:dataset-fiber-dist}
\end{figure}
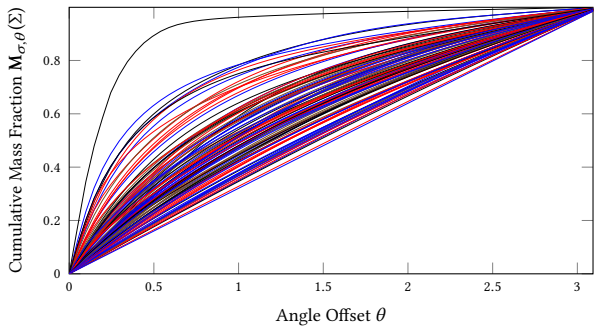
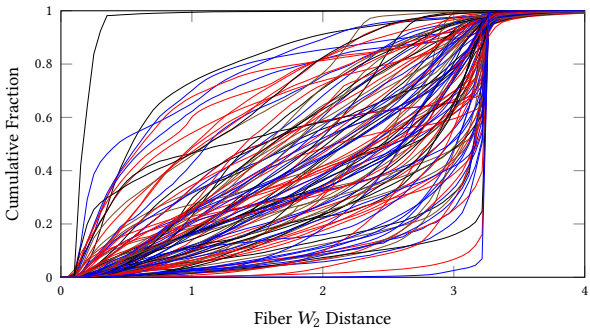
\begin{figure}
	\centering
	\pgfplotstableread[col sep=comma]{figures/dataset/W2_cdf_resaved.csv}{\datasetStat}
	\pgfplotstablegetcolsof{\datasetStat}
	\pgfmathsetmacro{\N}{\pgfplotsretval-1}
	\tikzsetnextfilename{W2-cdf}
	\begin{tikzpicture}
	\begin{axis}[
	        mark size = 0.5pt,
	        width = \columnwidth,
	        height = 0.6\columnwidth,
	        xmin = 0,
	        xmax = 4,
	        enlarge x limits = 0,
	        enlarge y limits = 0,
	        grid = none,
	        xlabel = {\footnotesize Fiber $W_2$ Distance},
	        ylabel = {\footnotesize Cumulative Fraction},
	        xlabel near ticks,
	        ylabel near ticks,
	        every tick label/.append style = {font=\tiny}
	    ]
	    \foreach \column in {1,...,\N}{%
  			\addplot+[no markers, solid, ultra thin] table [x={bin},y index=\column] {\datasetStat};
		}
	\end{axis}
	\end{tikzpicture}
	\caption{The concentration of $\Sigma$ toward the extracted section is heterogeneous over the base surface, as measured by fiberwise $W_2$ distance.}
	\label{minsec:fig:dataset-W2-dist-cdf}
\end{figure}
\begin{figure}
	\pgfplotstableread[col sep=comma]{figures/dataset/global_stats.csv}{\datasetStat}
	\tikzsetnextfilename{W2-vs-mean-abs-curv}
	\begin{tikzpicture}
	\begin{semilogxaxis}[
	        mark size = 0.5pt,
	        width = \columnwidth,
	        height = 0.5\columnwidth,
	        grid = none,
	        xlabel = {\footnotesize Mean Vertex Curvature},
	        ylabel = {\footnotesize Mean Fiber $W_2$ Distance},
	        xlabel near ticks,
	        ylabel near ticks,
	        legend pos = north west,
	        legend cell align = left,
	        legend style = {font=\footnotesize, row sep=0.1pt},
	        every tick label/.append style = {font=\tiny}
	    ]	    
	    \addplot [only marks,black,forget plot] table [x = mean_abs_curvature, y = mean_W2] {\datasetStat};
		\addplot [blue] table[x = mean_abs_curvature, y={create col/linear regression={x = mean_abs_curvature, y = mean_W2}}] {\datasetStat};
	\end{semilogxaxis}
	\end{tikzpicture}\\%
	\tikzsetnextfilename{W2-vs-max-abs-curv}
	\begin{tikzpicture}
	\begin{semilogxaxis}[
	        mark size = 0.5pt,
	        width = \columnwidth,
	        height = 0.5\columnwidth,
	        grid = none,
	        xlabel = {\footnotesize Maximum Vertex Curvature},
	        ylabel = {\footnotesize Mean Fiber $W_2$ Distance},
	        xlabel near ticks,
	        ylabel near ticks,
	        legend pos = north west,
	        legend cell align = left,
	        legend style = {font=\footnotesize, row sep=0.1pt},
	        every tick label/.append style = {font=\tiny}
	    ]	    
	    \addplot [only marks,black,forget plot] table [x = max_abs_curvature, y = mean_W2] {\datasetStat};
		\addplot [blue] table[x = max_abs_curvature, y={create col/linear regression={x = max_abs_curvature, y = mean_W2}}] {\datasetStat};
	\end{semilogxaxis}
	\end{tikzpicture}
	\caption{The computed current tends to be more diffuse when the input mesh is highly curved.}
	\label{minsec:fig:W2-vs-mesh}
\end{figure}
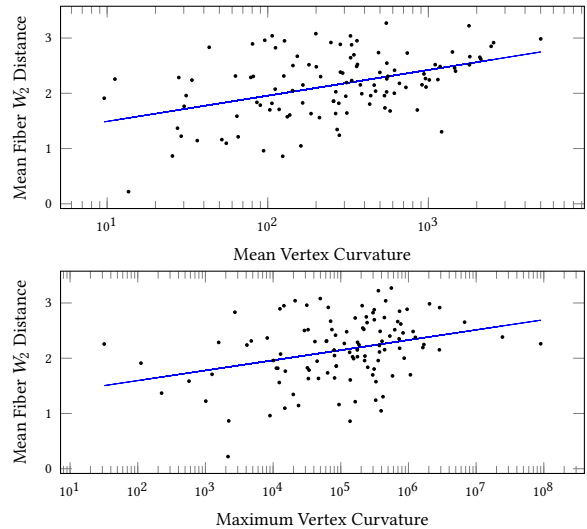

\paragraph{Performance.} \Cref{minsec:fig:dataset-timing} depicts the runtime scaling of our method as a function of base mesh size (number of vertices). The method is slightly superlinear. Experiments were done on a system with eight Intel Ice Lake cores and an Nvidia A100 MIG GPU.
\begin{figure}
	\pgfplotstableread[col sep=comma]{figures/dataset/global_stats.csv}{\datasetStat}
	\tikzsetnextfilename{performance-scaling}
	\begin{tikzpicture}
	\begin{loglogaxis}[
	        mark size = 0.5pt,
	        width = \columnwidth,
	        height = 0.6\columnwidth,
	        grid = none,
	        xlabel = {\footnotesize \# Vertices},
	        ylabel = {\footnotesize Time (s)},
	        xlabel near ticks,
	        ylabel near ticks,
	        legend pos = north west,
	        legend cell align = left,
	        legend style = {font=\footnotesize, row sep=0.1pt},
	        every tick label/.append style = {font=\tiny}
	    ]	    
	    \addplot [only marks,black,forget plot] table [x = nv, y = time] {\datasetStat};
		\addplot [blue] table[x = nv, y={create col/linear regression={x = nv, y = time}}] {\datasetStat};
		\xdef\scalingExponent{\pgfmathprintnumber{\pgfplotstableregressiona}}
		\addlegendentry{$\propto V^{\scalingExponent}$}
	\end{loglogaxis}
	\end{tikzpicture}
	\caption{Overall performance of our \textsc{admm}-based method scales nearly linearly with base mesh size.}
	\label{minsec:fig:dataset-timing}
\end{figure}
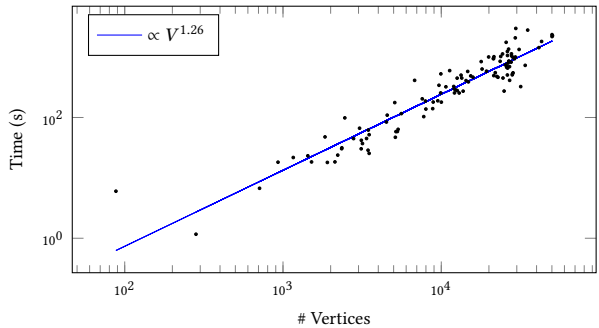

\paragraph{Meshing.} Our direction fields are ready to use for downstream applications such as quadrilateral meshing. Examples of meshes computed from our fields are shown in \Cref{minsec:fig:faces-cross} and compared against those computed with the method of \cite{knoppel_globally_2013}. All meshes were computed by the seamless integration method of \citet{bommes2013integer}.

\section{Discussion}
Lifting field optimization from maps to their graphs unlocks the rich tools of geometric measure theory, allowing us to precisely characterize the relationship between fields and their singularities as a linear boundary condition. The lifted formulation lends itself to a convex relaxation, and our hybrid Fourier discretization avoids the problem of discretizing a curved bundle.

This paper explores the case of fields valued in circle bundles, but a larger class of problems in geometry processing and computational engineering can be formulated as optimization problems over sections with singularities. The most obvious generalization is to volumetric frame fields \cite{solomon_boundary_2017,liu_singularity-constrained_2018,palmer_algebraic_2020}. More surprisingly, it may be possible to formulate end-to-end meshing problems as optimization problems over sections, with the right choice of fiber. For example, a quadrilateral mesh may be viewed as a section of a special bundle whose fiber is a certain $3$-manifold, itself in turn constructed as a twisted torus bundle (\emph{cf.} \cite{hocking_topological_2022}, which explores this idea from the perspective of modeling crystal defects in physics).

Supposing that the appropriate minimal section problems can be formulated in the language of currents, the next big question is whether convex relaxation from integral to normal currents remains useful. Even if exact recovery is achieved, it may not be worth the computational cost of discretizing a bundle of much higher dimension than the base domain, as the size of the discretization will generally scale exponentially with the dimension. Given this, an appealing alternative might be to use a neural implicit representation for currents in high dimension, \emph{a la} \cite{palmer_deepcurrents_2022}. Even then, working with the mass norm in dimension greater than three may be challenging due to issues arising from non-simple $p$-forms.

Alternatively, one could take inspiration from the simpler sparse deconvolution problem, where the conic particle optimization me\-thod of \citet{chizat_sparse_2020} converges exponentially to the global optimum given suitable initialization. If such Lagrangian methods could be extended from measures to currents, they could provide a computationally efficient alternative to discretizing the ambient space.

Along the same lines, a tantalizing approach is to develop a \emph{reduced-order model} of singularities. In such a model, upon ``integrating out'' the degrees of freedom of the field, the field energy would induce an effective energy encompassing interactions between the singularities. This would dramatically trim the computational complexity because only degrees of freedom on the singularities would need to be represented. Reduced-order models of elastic defects in two dimensions have been developed \cite{moshe_elastic_2015}, but they rely on lucky linearity properties that do not generalize to the volumetric case. Developing reduced-order models of volumetric frame field and mesh singularities remains an appealing challenge for future work.

\begin{acks}
We are grateful to David Bommes, Paul Zhang, Etienne Vouga, and Josh Vekhter for many interesting discussions. We thank David Bommes for helping with quad meshing experiments and Keenan Crane for providing surface models.

David Palmer acknowledges the generous support of the Fannie and John Hertz Foundation Fellowship
and the \grantsponsor{NSF}{NSF}{https://nsf.gov} Mathematical Sciences Postdoctoral Research Fellowship under award \#\grantnum{NSF}{2303403}. Albert Chern appreciates the generous support of the \grantsponsor{NSF}{National Science Foundation}{https://nsf.gov} CAREER Award \#\grantnum{NSF}{2239062}.
The MIT Geometric Data Processing group acknowledges the generous support of \grantsponsor{ARO}{Army Research Office}{https://arl.devcom.army.mil/who-we-are/aro/} grants \grantnum{ARO}{W911NF2010168} and \grantnum{ARO}{W911NF2110293}, of \grantsponsor{AFOSR}{Air Force Office of Scientific Research}{https://www.afrl.af.mil/AFOSR/} award \grantnum{AFOSR}{FA9550-19-1-031}, of \grantsponsor{NSF}{National Science Foundation}{https://nsf.gov} grant \grantnum{NSF}{CHS-1955697}, from the CSAIL Systems that Learn program, from the MIT–IBM Watson AI Laboratory, and from the Toyota--CSAIL Joint Research Center.
\end{acks}

\printbibliography

\appendix
\section{Mathematical Background} \label{minsec:sec:math-prelim}
Our convex relaxation uses the concept of \emph{current} from geometric measure theory. In \Cref{minsec:subsec:gmt}, we recall the pertinent definitions. In geometric measure theory, currents are constructed as dual to exterior forms, and formulating a minimal area problem requires defining a metric in the ambient space, here a circle bundle over a base surface. To that end, we introduce the basics of exterior calculus and Riemannian geometry on a circle bundle. In \Cref{minsec:subsec:ehresmann}, we examine the horizontal-vertical decomposition comprising the connection on the bundle. In \Cref{minsec:subsec:bundle-metric}, we describe the \emph{Sasaki} or \emph{Kaluza-Klein} metric, a natural Riemannian metric on the total space of the bundle. We discuss forms on the bundle in \Cref{minsec:subsec:bundle-decomp} and the metric on forms in \Cref{minsec:subsec:forms-metric}. 

\subsection{Currents} \label{minsec:subsec:gmt}
We briefly recall the essential definitions of geometric measure theory. We refer the interested reader to \cite{lang2005,federer_gmt,simon2014introduction} for a deeper introduction to this rich theory.

Currents on a manifold $X$ are dual to differential forms, generalizing the duality between measures and smooth functions:
\begin{definition}[$p$-current]
	The space $\mathcal{D}_p(X)$ of \textbf{$p$-currents} on $X$ is defined as the dual to the space of smooth compactly-supported $p$-forms on $X$:
	\[ \mathcal{D}_p(X) \coloneqq (\Omega^p_c(X))^* \]
	That is, a $p$-current $\Sigma \in \mathcal{D}_p(X)$ is a continuous linear functional on compactly-supported $p$-forms, and we denote the duality pairing by $\langle \Sigma, \omega \rangle$ for $\omega \in \Omega^p_c(X)$.
\end{definition}

The spaces of currents are vector spaces, and in particular $0$-currents are measures as implied above, $\mathcal{D}_0(X) = \mathcal{M}(X)$. Just as measures can be viewed as generalized point sets, higher-dimensional currents can be viewed as generalized submanifolds. In particular, any smooth surface $\Sigma$ defines a current, as
\begin{equation}
	\omega \mapsto \int_{\Sigma}\omega \in \R
\end{equation}
is a continuous linear functional on $2$-forms. Informally, one way to view currents is as ``objects that can be integrated over.''

Given a metric on the ambient manifold $X$, one can define a norm on currents by duality:
\begin{definition}[mass norm]
The \textbf{mass} of a $p$-current $\Sigma$ is given by
\[
	\Mass(\Sigma) \coloneqq \sup_{\|\alpha\|_\infty \le 1} \langle \Sigma, \alpha\rangle,
\]
where the supremum is over $p$-forms whose pointwise norm is uniformly bounded.\footnote{Strictly speaking, the pointwise norm in question should be the \emph{comass norm}. As this distinction only becomes relevant in ambient dimensions greater than $3$, we will not discuss it further here.}
\end{definition}
In case $\Sigma$ is a surface, $\Mass(\Sigma) = \Area(\Sigma)$ (this holds more generally for higher-dimensional submanifolds). Thus the mass norm can be seen as a generalized area functional. Similarly, one can generalize the topological notion of boundary to currents:
\begin{definition}[boundary operator]
	The \textbf{boundary} of a $p$-current $\Sigma$ is a $(p-1)$-current $\partial\Sigma$ defined by its action on $(p-1)$-forms as follows:
	\[ \langle \partial \Sigma, \omega \rangle \coloneqq \langle \Sigma, \dext \omega \rangle. \]
\end{definition}
Note that in case $\Sigma$ is a smooth submanifold, this is just Stokes' Theorem. So the boundary of a current is given by extending Stokes' Theorem to currents as a definition. Similarly, one can define a Cartesian product of currents, which generalizes the Cartesian product on submanifolds. We will use this notion along with the boundary operator to express singularities of a field as vertical boundary components of its graph:
\begin{definition}[Cartesian product of currents]
	Given currents $\Sigma \in \mathcal{D}_p(X), \Xi \in \mathcal{D}_q(Y)$, their \textbf{Cartesian product} is a $(p+q)$-current $\Sigma \times \Xi \in \mathcal{D}_p(X \times Y)$ given by
	\[ \langle\Sigma \times \Xi, \phi (\pi_X^* \alpha) \wedge (\pi_Y^*\beta)\rangle =
	\begin{cases} \langle \Sigma, \alpha \langle \Xi, \phi \beta \rangle\rangle  & \text{if }\alpha \in \Omega^p(X), \beta \in \Omega^q(Y) \\
	0 & \text{otherwise},
	\end{cases} \]
	and extended by linearity to general forms on the bundle. Here $\pi_X : X \times Y \to X$ and $\pi_Y : X \times Y \to Y$ are the projections.
\end{definition}

Geometric measure theory considers various formal optimization problems over spaces of currents. For example, the classical minimal surface problem or \emph{Plateau's problem} may be formulated as
\begin{equation}
	\inf_\Sigma \{\Mass(\Sigma) : \partial \Sigma = \Gamma\}.
\end{equation}
Theoretical work on existence and regularity of optimal solutions often finds it convenient to restrict the domain of optimization from all currents to a subspace of \emph{integral currents} $\mathcal{I}_p(X) \subset \mathcal{D}_p(X)$ satisfying additional technical conditions we will not expand on herein. Suffice it to say that optimization over integral currents is to optimization over general currents as integer programming is to linear programming. As such, integral currents are difficult to work with numerically. But in some cases, a convex relaxation to general currents yields integral solutions: a case of \emph{exact recovery} \cite{federer_real_1974,brezis_plateau_2019}. Our formulation of the minimal section problem in \Cref{minsec:sec:problem} adopts a relaxation to general currents.

\subsection{Ehresmann Connection}
\label{minsec:subsec:ehresmann}
Our minimal section problem \eqref{minsec:prob.sym} takes place in a circle bundle. To construct the mass norm, we need to metrize the total space of the circle bundle. Note that the fiber and base metrics alone are insufficient to uniquely specify a metric on the total bundle. The total metric also encodes which \emph{horizontal} directions are orthogonal to the vertical fibers. This choice determines the slope of a section graph relative to the base, which will give our area functional differential-geometric meaning.

Such a choice of horizontal reference directions in each fiber is precisely the information encoded in a \emph{connection} on the bundle. As a natural connection inherited from the base surface is readily available, we use this connection to construct the Riemannian metric on the total space of the bundle. We first recall the notion of connection appropriate to a circle bundle. In \Cref{minsec:subsec:bundle-metric} we will introduce the natural metric on the bundle induced by the connection.

\begin{figure}
\centering
	\tikzsetnextfilename{bundle-schematic}
	\begin{tikzpicture}[scale=\columnwidth/5.5cm]
		\fill[gray,opacity=0.2] (-1.5, .5) -- (-2.5, -.5) .. controls (-1, -1) and (0, 0) .. (1.5, -.5) -- (2.5, .5) .. controls (1, 1) and (0, 0) .. cycle;
		\fill[cyan!5] (-2.5, -.5) -- (-1.5, .5) .. controls (0, 0) and (1, 1) .. (2.5, .5) -- (2.5, 2.5) .. controls (1, 3) and (0, 2) .. (-1.5, 2.5) -- (-2.5, 1.5) -- cycle;
		\draw[dashed] (-2, 0) .. controls (-.5, -.5) and (.5, .5) .. (2, 0) -- (2, 2) .. controls (.5, 2.5) and (-.5, 1.5) .. (-2, 2) -- cycle;
		\draw[thick,cyan,>->] (-1, -.3) -- (-1, 2) node[pos=1,label={[label distance=-8pt]above right:{$\pi^{-1}(p)\simeq \Sph^1$}}] {};
		\node at (1, -.2) {$B$};
		\draw[thick,->] (-2, 0) .. controls (-.5, -.5) and (.5, .5) .. (2, 0) node[pos=0.8,label={[label distance=-4pt]above:$\gamma$}] {};
		\node[cyan] at (2.2, 2.2) {$E$};
		\draw[magenta,thick,->] (-2, 1) .. controls (-.5, .28) and (1, 1) .. (2, 1) node[pos=0.8,label={[label distance=-4pt]above:{$\tilde\gamma$}}] {};
		\draw[cyan,very thick,->] (-1, .7) -- (-1, 1.2) node[label={right:{$V_x E$}}] {};
		\fill[magenta,opacity=0.2] (-1.3, .95) -- (-.1, .85) -- (-.7, .45) -- (-1.9, .55) -- cycle;
		\node[magenta] at (-.7, .8) {$H_x E$};
		\node[circle,fill=black,inner sep=0pt,minimum size=3pt,label={[label distance=-2pt]above right:$p$}] at (-1, -.14) {};
		\node[circle,fill=cyan,inner sep=0pt,minimum size=3pt,label={[label distance=-2pt,cyan]below right:$x$}] at (-1, .7) {};
	\end{tikzpicture}
	\caption{The connection specifies a horizontal subspace of the tangent space to the bundle $E$ at each point $x \in E$. This uniquely determines horizontal lifts, i.e., parallel transport of points in the bundle.}
	\label{minsec:fig:bundle-schematic}
\end{figure}
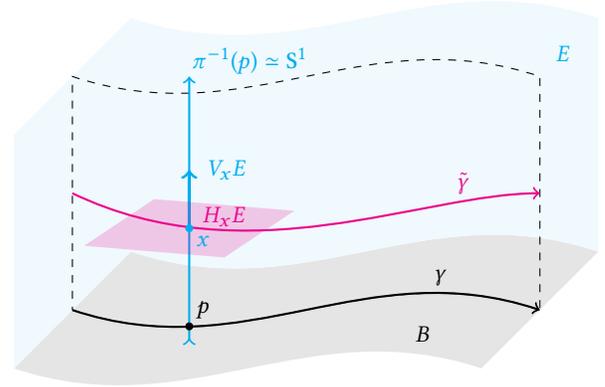
Let $\pi: E \to B$ be an orientable circle bundle over a surface $B$. Because the fibers are circles, each point $x \in E$ has a distinguished \emph{vertical} direction, given by the tangent vector $\vvf_x$ to the fiber. $\vvf_x$ spans the \textbf{vertical subspace} $V_x E \subset T_x E$, and we will refer to $\vvf$ as the \textbf{vertical vector field}. Unlike vertical directions, horizontal directions are not canonically defined by the bundle geometry. The additional structure of a connection is needed to explicitly fix them. The following definition formalizes this idea, applicable to fiber bundles more generally.

\begin{definition}[\protect{\cite[Definition 10.1]{nakaharaGeometryTopologyPhysics2003}}]
An \textbf{Ehresmann connection} on $E$ is a smoothly varying selection of a subspace $H_x E \subset T_x E$ of the tangent space to the bundle at each point $x \in E$, transverse to the vertical subspace $V_x E$, and preserved under vertical translation. $H_x E$ is called a \textbf{horizontal subspace}, $HE$ is accordingly the \textbf{horizontal sub-bundle}, and vectors in $HE$ are referred to as \textbf{horizontal}.
\end{definition}
 Intuitively, the connection specifies the rules for parallel transport. Just as the Levi-Civita connection allows one to parallel transport a tangent vector (i.e., a point in the tangent bundle) along a Riemannian manifold, an Ehresmann connection describes how to parallel transport a point in a fiber bundle $E$ along a path in the base space.
To parallel transport a point $x \in E$ along a curve $\gamma : [0, 1] \to B$ with $\pi(x) = \gamma(0)$, one chooses a \textbf{horizontal lift} of the curve $\tilde{\gamma} : [0, 1] \to E$, tangent to $HE$:
\begin{equation}
	\pi \circ \tilde\gamma = \gamma \qquad \dot{\tilde\gamma}(t) \in H_{\tilde\gamma(t)}E
\end{equation}
This amounts to solving a first-order \textsc{ode} on $E$ and provides a geometric notion of parallel transport on fiber bundles (\Cref{minsec:fig:bundle-schematic}).

The connection can be specified more concretely in local coordinates. Take an open set $U \subset B$ endowed with a local trivialization $U \times \Sph^1 \cong E_U \to U$. 
Locally on $E_U$, one can define the horizontal subspaces by a connection form $\tau \in \Omega^1(E_U)$ satisfying the transversality condition $\iota_{\vvf} \tau = 1$, where $\iota$ denotes the interior product or pairing between vector fields and $1$-forms. One takes $H_x E$ to be the subspace of $T_x E$ annihilated by $\tau$, and the transversality condition ensures that $\vvf_x$ does not fall in $H_x E$.
Finally, invariance under vertical translations amounts to the requirement that $\Lie_{\vvf} \tau = 0$. Using Cartan's formula, we can rewrite this as
\begin{equation}
	0 = \Lie_{\vvf} \tau = \dext \iota_{\vvf} \tau + \iota_{\vvf} \dext \tau = \iota_{\vvf} \dext \tau. \label{minsec:eq.dtau-horiz}
\end{equation}
When $E$ is the unit tangent bundle $\Sph TB$, the natural choice of connection is the restriction of the Levi-Civita connection on $TB$.

\subsection{Metric on the Bundle}
\label{minsec:subsec:bundle-metric}
In this section, we use the machinery of \Cref{minsec:subsec:ehresmann} to describe the concrete metric on our circle bundle. We confirm that, with this choice of metric, the area functional indeed measures field energy, proving \Cref{minsec:prop:area-func}.

The \textbf{Sasaki metric} \cite{sasaki_differential_1958} is a Riemannian metric $g$ induced by the connection on a fiber bundle $E$. When $E$ is a circle bundle, it is also known as the Kaluza-Klein metric after the famous work of Kaluza and Klein on geometrizing electromagnetism \cite{kaluzaUnitatsproblemPhysik1921,kleinQuantentheorieUndFuenfdimensionale1926}. See also \citet{padilla_bubble_2019} for a recent application in graphics. The metric is specified by stipulating that it is homogeneous along the fibers, makes the vertical and horizontal subspaces orthogonal, and agrees with the metric on the base:
\begin{equation} VE \perp_g HE, \qquad g\mid_{HE} = \pi^* g_B. \end{equation}
Together, these constraints ensure that the area of a \emph{parallel} or horizontal section is minimal.
For a circle bundle, these constraints leave one free parameter, the length of the fiber $\ell$, or equivalently, the fiber radius $r = \ell / 2\pi$.
Choosing $r$ fixes the pointwise norm $|\vvf|_g = r$ as
\begin{equation}
	2\pi r = \ell = \int_{\pi^{-1}(x)} |\vvf|_g \tau = 2\pi |\vvf|_g,
\end{equation}
and in matrix form, the resulting metric is
\begin{equation} 	\mathbf{v}^\top g \mathbf{w} = \begin{pmatrix}
	\mathbf{v}_H \\ v_V
\end{pmatrix} \begin{pmatrix}
g_B & 0 \\ 0 & r^2
\end{pmatrix} \begin{pmatrix}
	\mathbf{w}_H \\ w_V
\end{pmatrix}.
\end{equation}
Henceforth, we will assume a radius $r$ has been chosen and drop the subscript $g$ where understood. In \Cref{minsec:sec:intuition} we saw that changing $r$ changes the behavior of the area functional from Dirichlet-like to total variation--like. In \Cref{minsec:sec:results} we explore the practical implications of this parameter.

Having introduced the Sasaki metric, we can now prove \Cref{minsec:prop:area-func}, demonstrating how the area functional generalizes typical field energies:
\begin{proof}[Proof of \Cref{minsec:prop:area-func}]
Let $E_U \to U$ be a local trivialization with coordinates $(x^1, \dots, x^n, \theta)$, and parametrize the graph of $\sigma$ by $\tilde\sigma(x) = (x, \sigma(x))$. Denote $\partial_i \coloneqq \partial/\partial x^i$. Then
\[ \dext \tilde\sigma(\partial_i) = \partial_i + \partial_i\sigma \, \partial_\theta, \]
while the horizontal lift of $\partial_i$ can be expressed locally as
\[ (\partial_i)^H = \partial_i - \tau(\partial_i) \partial_\theta. \]
Evaluating the area form on the graph $\Sigma$, we obtain
\[
	\tilde\sigma^* \dext A_\Sigma(\partial_1, \dots, \partial_n) = \dext A_\Sigma(\dext \tilde\sigma(\partial_1), \dots, \dext \tilde\sigma(\partial_n)) = \sqrt{\det \tilde{g}},
\]
where
\begin{align*}
\tilde{g}_{ij} &= \langle \dext \tilde\sigma(\partial_i), \dext \tilde\sigma(\partial_j) \rangle \\
&= \langle (\partial_i)^H + (\partial_i\sigma + \tau(\partial_i)) \partial_\theta, (\partial_j)^H + (\partial_j\sigma + \tau(\partial_j)) \partial_\theta \rangle \\
&= \langle \partial_i, \partial_j\rangle_B + r^2 (\Dext_i \sigma)(\Dext_j \sigma)
\end{align*}
with the covariant derivative $\Dext_i \sigma \coloneqq \partial_i\sigma + \tau(\partial_i)$ induced by the connection.
Hence $\det \tilde{g} = (1 + r^2 \|\Dext \sigma\|_{g_B^{-1}}^2)\det g_B$, i.e.,
\[ \tilde\sigma^* \dext A_\Sigma = \sqrt{1 + r^2 |\Dext \sigma|^2} \; \dext A_B. \qedhere \]
\end{proof}

\subsection{Differential Forms on the Bundle}
\label{minsec:subsec:bundle-decomp}

As currents are dual to differential forms, we establish some helpful facts about forms on the bundle. Discretizing a curved bundle with finite elements is challenging. Our discretization bypasses this by decomposing forms on the bundle into their Fourier components whose coefficients are \emph{horizontal homogeneous} forms. As these are identified with forms on the base surface, we can do all computations on the base. Below, we introduce the horizontal forms and the decomposition of forms on $E$ into their horizontal and vertical parts.
\begin{definition}[Horizontal forms]
A $p$-form $\xi$ is \textbf{horizontal} if it satisfies $\iota_\vvf \xi = 0$. $\xi$ is \textbf{horizontal homogeneous} if it additionally satisfies $\iota_{\vvf} \dext \xi = 0 \implies \Lie_{\vvf} \xi = \dext \iota_{\vvf} \xi + \iota_{\vvf} \dext \xi = 0$.
We denote the horizontal homogeneous forms on $E_U$ by $\Omega_H^p(E_U)$.
\end{definition}
Horizontal and horizontal homogeneous forms are known elsewhere as \emph{semi-basic} and \emph{basic} respectively \cite{iveyCartanBeginnersDifferential2016}.
Horizontal forms can be identified pointwise with forms on the base, but their values may vary along the fiber direction.
The requirement that $\Lie_{\vvf} \dext \xi = 0$ ensures that a horizontal homogeneous form $\xi$ is invariant under global vertical translation.
As a result, horizontal homogeneous $p$-forms can be naturally identified with forms on the base:
\begin{proposition} \label{minsec:prop:horiz-base}
	$\pi^*$ is an isometry between horizontal homogeneous $p$-forms $\Omega^p_H(E)$ and $p$-forms on the base $\Omega^p(B)$.
\end{proposition}
\begin{proof}
Given a horizontal homogeneous $p$-form $\eta$ and $p$ vectors $\mathbf{v}_1, \dots, \mathbf{v}_p \in T_x B$, take horizontal lifts $\tilde{\mathbf{v}}_1, \dots, \tilde{\mathbf{v}}_p \in H_yE_U$ for each $y$ in the fiber $\pi^{-1}(x)$, and define the descent
\[ (P_{\tau} \eta)_x(\mathbf{v}_1, \dots, \mathbf{v}_p) \coloneqq \frac{1}{2\pi}\int_{\pi^{-1}(x)}\eta_y(\tilde{\mathbf{v}}_1, \dots, \tilde{\mathbf{v}}_p)\tau = \eta_x(\tilde{\mathbf{v}}_1, \dots, \tilde{\mathbf{v}}_p). \]
Then we claim that $P_{\tau} = (\pi^*)^{-1}$ on horizontal homogeneous forms.

To see this, let $\eta$ be a horizontal homogeneous form, and let $\mathbf{w}_1, \dots, \mathbf{w}_p \in T_x E$ be a collection of vectors. Then
\begin{align*} (\pi^* P_\tau \eta)(\mathbf{w}_1, \dots, \mathbf{w}_p) &= (P_\tau \eta)(\pi_\# \mathbf{w}_1, \dots, \pi_\# \mathbf{w}_p) \\
&= \eta(\tilde{\mathbf{w}}_1, \dots, \tilde{\mathbf{w}}_p),
\end{align*}
where $\tilde{\mathbf{w}}_j$ are horizontal and agree with $\mathbf{w}_j$ under $\pi_\#$. Thus, $\tilde{\mathbf{w}}_j - \mathbf{w}_j \in \ker\pi_\# = V_x E = \Span\{\vvf_x\}.$ As $\iota_{\vvf} \eta = 0$, we have $\pi^* P_\tau \eta = \eta$ as desired.

To see that this map is an isometry, consider a horizontal $1$-form $\omega \in \Omega_H^1(E_U)$. Then 
\[ |\omega|_g = \sup_{|\mathbf{v}|_g \le 1} \omega(\mathbf{v}) = \sup_{|\mathbf{v}|_g \le 1} (P_\tau \omega)(\pi_\#\mathbf{v}) = \sup_{|\mathbf{w}|_{g_B}\le 1} (P_\tau\omega)(\mathbf{w}) = |P_\tau \omega|_{g_B}, \]
where the first equality follows because $\omega$ is zero on the vertical part of $\mathbf{v}$. It follows by polarization that the induced metric $g$ on horizontal forms agrees with that induced by $g_B$ on their projections, and this extends to higher-degree forms.%
\end{proof}

In particular, \eqref{minsec:eq.dtau-horiz} along with $\dext^2 = 0$ implies that $\dext \tau$ is horizontal homogeneous. It is identified with the negative of the \textbf{curvature $2$-form} $\kappa$ on the base $B$, $\dext \tau = -\pi^* \kappa$.

Just as for vectors in $TE$, $p$-forms on $E$ can be decomposed into their horizontal and vertical parts. To see this decomposition, we define linear projection operators $\Pi_H \coloneqq \iota_{\vvf} \circ (\tau \wedge)$ and $\Pi_V \coloneqq (\tau \wedge) \circ \iota_{\vvf}$, which have the properties
\begin{equation}
	\Pi_H + \Pi_V = \id_{\Omega^p}, \qquad \Pi_H \circ \Pi_V = \Pi_V \circ \Pi_H = 0.
\end{equation} 
Then a $p$-form $\xi \in \Omega^p(E_U)$ decomposes as follows:
\begin{gather} \xi = \tau \wedge \xi_V + \xi_H \label{minsec:eq:form-decomp} \\
\xi_H \coloneqq \Pi_H \xi \qquad \tau \wedge \xi_V \coloneqq \Pi_V \xi, \end{gather}
i.e., $\xi_V = \iota_{\vvf} \xi$.
We call $\xi_V$ and $\xi_H$ the vertical and horizontal parts of $\xi$, respectively, though both are horizontal forms as $\iota_\vvf^2 = 0$ implies $\iota_{\vvf} \xi_V = \iota_{\vvf} \xi_H = 0$.

\subsection{Metric on Forms} \label{minsec:subsec:forms-metric}
The matrix of the metric on $1$-forms is the inverse of that on vectors. In particular, $|\tau|_g = r^{-1}$. We recall for convenience some metric-dependent operators relevant to subsequent derivations.
\paragraph{Hodge Star.}
The Hodge star on the total space $E$ decomposes into blocks consisting of Hodge stars on the base. Letting $r = \ell/2\pi$ be the fiber radius, we have
\begin{equation} \begin{pmatrix} \Pi_H \\ \iota_{\vvf} \end{pmatrix}\hodge^E_p = \begin{pmatrix} 0 & r^{-1}\hodge^B_{p-1} \\ (-1)^p r\hodge^B_p & 0 \end{pmatrix}\begin{pmatrix} \Pi_H \\ \iota_{\vvf} \end{pmatrix}. \end{equation}

\paragraph{Musical Isomorphisms.} We will use the musical isomorphisms $\sharp: \Omega^1(E) \to \Gamma(TE)$ and $\flat: \Gamma(TE) \to \Omega^1(E)$ to convert between $1$-forms and vector fields in our discretization. These have the form
\begin{equation}
	\begin{pmatrix}
		\xi_H \\ \xi_V
	\end{pmatrix}^\sharp = \begin{pmatrix}
		\xi_H^\sharp \\ r^{-2}\xi_V
	\end{pmatrix},
\end{equation}
where $\xi_H^\sharp$ is given by the corresponding $\sharp$ on the base.

\section{Operator Details}
\subsection{Derivation of Frequency Decomposition}
\label{minsec:app.op-decomp}
Writing $\xi$ in the form \eqref{minsec:eq.bundle-fourier}, the exterior derivative decomposes in a nice fashion:
\[ \dext \xi = \sum_k e^{ik\theta} [ik \dext \theta \wedge (\tau \wedge \alpha_k + \beta_k) + \dext \tau \wedge \alpha_k - \tau \wedge \dext \alpha_k + \dext \beta_k].\]
So 
\begin{align*}
\Pi_V \dext \xi &= \tau \wedge (\iota_{\vvf}\dext \xi) \\
&= \sum_k e^{ik\theta} \tau \wedge [ik (\tau - \dext\theta) \wedge \alpha_k + ik\beta_k - \dext \alpha_k] \\
\Pi_H \dext \xi &= \iota_{\vvf} (\tau \wedge \dext \xi) \\
&= \sum_k e^{ik\theta} \iota_{\vvf} [ik \tau \wedge \dext\theta \wedge \beta_k + \tau \wedge \dext \tau \wedge \alpha_k + \tau \wedge \dext \beta_k] \\
&= \sum_k e^{ik\theta} [ik (\dext\theta - \tau) \wedge \beta_k + \dext \tau \wedge \alpha_k + \dext \beta_k],
\end{align*}
where we have used that $\iota_{\vvf} \tau = \iota_{\vvf} \dext\theta = 1$ and $\iota_{\vvf} \dext \tau = 0$.

Let's express the Hodge Laplacians on the bundle in terms of operators on the base. We have (eliding the $\iota_{\vvf}$ and $\Pi_H$ operators and dropping various subscripts for convenience):
\begin{align*}
\delta_k^p &= (-1)^p(\hodge_E^{p-1})^{-1}\dext_k^{n+1-p}\hodge_E^p \\
&=
\begin{pmatrix} 0 & -r^{-1}\hodge^{-1} \\ (-1)^p r\hodge^{-1} & 0 \end{pmatrix}
\begin{pmatrix}
\Dext_k & \dext \tau \wedge \\
ik & -\Dext_k
\end{pmatrix}
\begin{pmatrix} 0 & r^{-1}\hodge \\ (-1)^p r\hodge & 0 \end{pmatrix} \\
&= \begin{pmatrix} 0 & -r^{-1}\hodge^{-1} \\ (-1)^p r\hodge^{-1} & 0 \end{pmatrix}
\begin{pmatrix}
	(-1)^p r(\dext \tau \wedge)\hodge & r^{-1}\Dext_k\hodge \\
	(-1)^{p+1}r\Dext_k\hodge & r^{-1}ik\hodge
\end{pmatrix} \\
&= \begin{pmatrix}
	(-1)^p \hodge^{-1} \Dext_k \hodge & -r^{-2}ik \\
	r^2 \hodge^{-1}(\dext\tau\wedge)\hodge & (-1)^p \hodge^{-1}\Dext_k\hodge
\end{pmatrix}
\end{align*}
Hence
\begin{align*}
&\dext_k^{p-1}\delta_k^p \\
&\;= 
\begin{pmatrix}
\Dext_k & \dext \tau \wedge \\
ik & -\Dext_k
\end{pmatrix}
\begin{pmatrix}
	(-1)^p \hodge^{-1} \Dext_k \hodge & -r^{-2}ik \\
	r^2 \hodge^{-1}(\dext\tau\wedge)\hodge & (-1)^p \hodge^{-1}\Dext_k\hodge
\end{pmatrix} \\
&\;= \begin{psmallmatrix}
	r^2 (\kappa\wedge)\hodge^{-1}(\kappa\wedge)\hodge + (-1)^p \Dext_k \hodge^{-1} \Dext_k \hodge &
	(-1)^{p+1}(\kappa\wedge)\hodge^{-1}\Dext_k\hodge - r^{-2} ik \Dext_k \\
	r^2\Dext_k\hodge^{-1}(\kappa\wedge)\hodge +(-1)^p ik\hodge^{-1}\Dext_k\hodge &
	r^{-2} k^2 + (-1)^{p-1} \Dext_k \hodge^{-1}\Dext_k\hodge
\end{psmallmatrix} \\
&\delta_k^{p+1} \dext_k^p \\
&\;=
\begin{pmatrix}
	(-1)^{p+1} \hodge^{-1} \Dext_k \hodge & -r^{-2} ik \\
	r^2\hodge^{-1}(\dext\tau\wedge)\hodge & (-1)^{p+1} \hodge^{-1}\Dext_k\hodge
\end{pmatrix}
\begin{pmatrix}
\Dext_k & \dext \tau \wedge \\
ik & -\Dext_k
\end{pmatrix} \\
&\;= \begin{psmallmatrix}
	r^{-2} k^2 + (-1)^{p+1}\hodge^{-1} \Dext_k \hodge \Dext_k &
	r^{-2} ik \Dext_k + (-1)^{p}\hodge^{-1} \Dext_k \hodge (\kappa\wedge) \\
	(-1)^{p+1}ik\hodge^{-1}\Dext_k\hodge - r^2 \hodge^{-1}(\kappa\wedge)\hodge \Dext_k &
	(-1)^p\hodge^{-1}\Dext_k\hodge \Dext_k + r^2 \hodge^{-1}(\kappa\wedge)\hodge(\kappa\wedge)
\end{psmallmatrix},
\end{align*}
where we've substituted $-\kappa$ for $\dext \tau$. Finally,
\begin{align*}
\Delta_E^p &= \dext \delta + \delta \dext \\
	&= \begin{psmallmatrix}
	\Delta_{\Dext_k}^p + r^{-2} k^2 + r^2(\kappa\wedge)\hodge^{-1}(\kappa\wedge)\hodge &
	(-1)^p[\hodge^{-1}\Dext_k\hodge, (\kappa\wedge)] \\
	r^2[\Dext_k, \hodge^{-1}(\kappa\wedge)\hodge] &
	\Delta^{p-1}_{\Dext_k} + r^{-2}k^2 + r^2\hodge^{-1}(\kappa\wedge)\hodge(\kappa\wedge)
\end{psmallmatrix}.
\end{align*}

\subsection{Parallel Transport}
\label{minsec:app:op-pt}
We discretize the Levi-Civita connection of the base in terms of discrete parallel transport operators. Each operator $\rho_{a\to T}$ takes vectors in the tangent space at a vertex $a$ to vectors in the tangent space to a triangle $T$ adjoining $a$. Each tangent space is assigned an (arbitrary) intrinsic frame and thus identified with the complex plane. So $\rho_{a \to T}$ is expressed as a complex phase.

To construct $\rho_{a \to T}$, we first compute the ``extrinsic'' parallel transport operator, i.e., the principal rotation $R_{a \to T}$ mapping the vertex normal $\vect{n}_a$ to the face normal $\vect{n}_T$:
\begin{equation}
	R_{a \to T} \coloneqq \argmax_{\substack{R \in \SO(3) \\ R \vect{n}_a = \vect{n}_T}} \quad \tr R
\end{equation}

Given frames $F_a$ and $F_T$, a vector $\vect{v} \in T_a B$ (expressed in $F_a$ coordinates) is sent to $F_T^\dag R_{a \to T} F_a \vect{v}$ in the tangent space to $T$ (expressed in $F_T$ coordinates). Thus, we have
\begin{equation}
	F_T^\dag R_{a \to T} F_a = \begin{pmatrix}
		u & -v \\ v & u
	\end{pmatrix}
\end{equation}
for some real values $u$ and $v$. This is nothing but the intrinsic parallel transport operator $\rho_{a \to T}$, expressed as a real matrix. Accordingly, we set $\rho_{a \to T} \coloneqq u + iv$.

\section{Vertical Symmetry and Relation to Conformal Flattening}
\label{minsec:subsec:symmetry}
We can formulate the minimal section problem \eqref{minsec:prob.sym} or \eqref{minsec:prob.primal} on a closed surface, but the resulting convex relaxation suffers from a symmetry that makes it hard to achieve exact recovery.
Eliminating the boundary conditions, the remainder of \eqref{minsec:prob.primal} is symmetric under global translation along the fiber direction. Recall a general principle of convex analysis: if a (lower semicontinuous) convex functional $f$ is invariant under a group of symmetries $G$ acting by $(g\cdot f)(x) \coloneqq f(g^{-1}x)$, then it has an invariant optimum  $x^* = G x^*$. Thus we can consider without loss of generality a solution $\Sigma$ invariant under vertical translation,
\begin{equation}
	0 = \Lie_\vvf \Sigma = \dext \iota_\vvf \Sigma + \iota_\vvf \dext \Sigma = \dext \iota_\vvf \dext f = \dext (\partial_\theta f).
\end{equation}
Since $\theta$ is periodic, $\partial_\theta f = 0$, i.e., $f$ must be constant along the fibers. Therefore, problem \eqref{minsec:prob.primal} becomes
\begin{equation}
\begin{aligned}
 \min \quad & \ell \int_B \hodge |\bar\tau + \dext f + \hodge \dext \phi|_2 + \lambda \int_B \hodge |\Gamma| \\
 \subj \quad &\Delta \phi = \Gamma - \bar\kappa,
\end{aligned}
\tag{P-\textsc{sym}}
\label{minsec:prob.primal-symm}
\end{equation}
where we've eliminated the constraint $\Sigma_V \ge 0$ as $\Sigma_V = \bar\tau$.
We can write the first integral as
\begin{equation}
\begin{aligned}
	\int_B \hodge |\bar\tau + \dext f + \hodge \dext \phi|_2 &= \int_B \hodge \sqrt{r^{-2}/4\pi^2 + |\dext f + \hodge \dext \phi|_2^2} \\
	&\approx \ell^{-1} \int_B \left[1 + \frac{\ell^2}{2}|\dext f + \hodge \dext \phi|_2^2 \right] \dext A_B
\end{aligned}
\end{equation}
when the fiber radius $r$ is small. In the limit as $r \to 0$, the overall problem then becomes
\begin{equation}
\begin{aligned}
 \argmin \quad & \int_B \left[|\dext \phi|^2_2 + \frac{2\lambda}{\ell^2} |\Gamma| \right] \dext A \\
 \subj \quad &\Delta \phi = \Gamma - \bar\kappa,
\end{aligned}
\tag{P-\textsc{red}}
\label{minsec:prob.primal-reduced}
\end{equation}
where we have eliminated $f$ from the optimization problem using Hodge decomposition.
Notice the similarity to \eqref{minsec:eq.graph-area-reduced}, where the potential $\phi$ has replaced the angle $\theta$, and covariant differentiation has been replaced by ordinary differentiation.

The reduced problem \eqref{minsec:prob.primal-reduced} is a sparse inverse Poisson problem like that studied by \citet{soliman_optimal_2018}. They solve a similar problem (but with $\|\phi\|_2^2$ replacing $\|\dext\phi\|_2^2$) to compute cone singularities for conformal flattening, a setting in which quantization of the cone angles is unimportant.

In our case, however, quantization of singularities is essential to computing a well-defined section. Conversely, it is the concentration of the current $\Sigma$ onto a surface that enforces quantization. The upshot of this is that, to solve the minimal section problem with quantized singularities, we need boundary conditions to break the vertical translation symmetry.

\section{Parameter Scaling}
\label{minsec:app:scaling}

In the following, we analyze the relative scaling of the parameters $r$ and $\lambda$. Suppose that we have a minimal section $\Sigma$ that is \emph{integral}, i.e., supported on a surface, which is a graph away from its singularities $\Gamma$. For simplicity, we will assume all the singularities are simple (i.e., have degree $\pm 1$). Then $\Mass(\Gamma)$ simply counts the singularities.

We estimate the mass norm of $\Sigma$ as follows. Around each singularity $z_i$, excise a disk $B_i = B_R(z_i)$. Then we have
\begin{equation}
	\Mass(\Sigma) = \int_{B \setminus \cup B_i} \sqrt{1 + r^2 |\Dext \sigma|^2} \dext A + \sum_i \Mass(\Sigma_i),
\end{equation}
where $\Sigma_i = \Sigma \cap E\mid_{B_i}$, and $\Sigma = \gr \sigma$ away from the singularities. We approximate the masses of the singular components $\Sigma_i$ by considering the idealized case of a helicoid. The area of a helicoid of height $2\pi r$ and radius $R = kr$ is
\begin{equation}
	\pi r^2 (k\sqrt{1 + k^2} + \arcsinh k).
\end{equation}
On the other hand, we have $|\Dext \sigma|^2 = R^{-2}$ on $\partial B_i$. Assuming that $|\Dext \sigma|^2 \le k^{-2}$ away from the singularities, we have the estimate
\begin{equation}
\begin{aligned}
	\Mass(\Sigma) &\le (\Area(B) - \pi r^2 k^2 \Mass(\Gamma))\left(1+\frac{k^{-2}}{2}\right) \\
	&\quad + \pi r^2 \Mass(\Gamma) 
	\left(k\sqrt{1 + k^2} + \arcsinh k \right).
\end{aligned}
\end{equation}
For $k = 1$, this becomes
\begin{equation}
\begin{aligned}
	\Mass(\Sigma) \le \frac{3}{2} \Area(B) + \pi r^2 \Mass(\Gamma)\left(
	\arcsinh 1 + \sqrt{2} - \frac{3}{2} \right).
\end{aligned}
\end{equation}
Hence in the objective function $\Mass(\Sigma) + \lambda \Mass(\Gamma)$, $\lambda$ is renormalized by the addition of a factor proportional to $r^2$.

\end{document}